\documentclass{article}

\usepackage{pgf, pgfplots}
\usetikzlibrary
{
	arrows, automata, chains, datavisualization.formats.functions, fit, positioning, shapes
}

\usepackage{amsthm}
\usepackage{amssymb}
\usepackage[bb = boondox]{mathalfa}
\usepackage{dsfont}
\usepackage{mathtools}
\usepackage{multicol}
\usepackage{physics}
\numberwithin{equation}{section}

\usepackage{yquant, braket}

\usepackage[a4paper, left = 2.5cm, right = 2.5 cm, top = 2.5cm, bottom = 3.0 cm]{geometry}

\usepackage{xcolor}
\usepackage{varwidth}
\usepackage[breakable, skins, theorems]{tcolorbox}
\usepackage{graphicx}
\usepackage{hyperref}
\hypersetup
{
	colorlinks,
	linkcolor = {WordRed!75!black},
	citecolor = {WordBlueDarker25},
	urlcolor = {WordAquaDarker50}
}
\usepackage{nicematrix}
\usepackage{colortbl}
\usepackage{float}
\usepackage{cellspace}
\usepackage{makecell}
\usepackage[boxed, linesnumbered, ruled, vlined]{algorithm2e}
\usepackage{appendix}

\newtheorem{definition}{Definition}[section]
\newtheorem{theorem}{Theorem}[section]
\newtheorem{lemma}[theorem]{Lemma}
\newtheorem{proposition}[theorem]{Proposition}
\newtheorem{corollary}[theorem]{Corollary}
\newtheorem{example}{Example}[section]
\newtheorem{ManualCorollaryInner}{Corollary}
\newenvironment{ManualCorollary}[1]{%
	\renewcommand\theManualCorollaryInner{#1}%
	\ManualCorollaryInner
}{\endManualCorollaryInner}
\newtheorem{ManualPropositionInner}{Proposition}
\newenvironment{ManualProposition}[1]{%
	\renewcommand\theManualPropositionInner{#1}%
	\ManualPropositionInner
}{\endManualPropositionInner}
\newtheorem{ManualTheoremInner}{Theorem}
\newenvironment{ManualTheorem}[1]{%
	\renewcommand\theManualTheoremInner{#1}%
	\ManualTheoremInner
}{\endManualTheoremInner}

\definecolor{MyLightRed}{RGB}{244, 213, 245}
\definecolor{WordRed}{RGB}{255, 0, 102}
\definecolor{RedDarkLightest}{HTML}{ff0088}
\definecolor{RedDarkLight}{HTML}{ea005f}
\definecolor{RedPurple}{HTML}{aa007f}
\definecolor{Purple}{HTML}{911146}
\definecolor{WordLightGreen}{RGB}{140, 214, 192}
\definecolor{WordGreen}{RGB}{0, 176, 80}
\definecolor{GreenLightest}{HTML}{00ffa0}
\definecolor{GreenLighter1}{HTML}{00b383}
\definecolor{GreenLighter2}{HTML}{00aa7f}
\definecolor{GreenDark}{HTML}{225522}
\definecolor{GreenTeal}{HTML}{008080}
\definecolor{WordIceBlue}{RGB}{223, 227, 229}
\definecolor{MyVeryLightBlue}{RGB}{211, 245, 247}
\definecolor{WordBlueVeryLight}{RGB}{0, 176, 240}
\definecolor{WordBlueLight}{RGB}{0, 112, 192}
\definecolor{WordBlueDark}{RGB}{46, 116, 181}
\definecolor{WordBlueDarker}{RGB}{31, 78, 121}
\definecolor{WordBlueDarker25}{RGB}{54, 96, 146}
\definecolor{WordBlueDarker50}{RGB}{36, 64, 98}
\definecolor{WordBlueDarkest}{RGB}{0, 32, 96}
\definecolor{WordBlue}{RGB}{19, 65, 99}
\definecolor{MyBlue}{RGB}{0, 64, 128}
\definecolor{MyDarkBlue}{RGB}{0, 51, 102}
\definecolor{BlueVeryDark}{HTML}{222255}
\definecolor{WordAquaLighter80}{RGB}{218, 238, 243}
\definecolor{WordAquaLighter60}{RGB}{183, 222, 232}
\definecolor{WordAquaLighter40}{RGB}{146, 205, 220}
\definecolor{WordAquaDarker25}{RGB}{49, 134, 155}
\definecolor{WordAquaDarker50}{RGB}{33, 89, 103}
\definecolor{WordVeryLightTeal}{RGB}{223, 236, 235}
\definecolor{WordLightTeal}{RGB}{160, 199, 197}
\definecolor{WordDarkTealLighter80}{RGB}{207, 223, 234}
\definecolor{WordDarkTeal}{RGB}{72, 123, 119}
\definecolor{WordDarkerTeal}{RGB}{48, 82, 80}
\definecolor{WordTurquoiseLighter80}{RGB}{209, 238, 249}

\title{The connection between the $PQ$ penny flip game and the dihedral groups}

\author{
	Theodore Andronikos$^1$ and Alla Sirokofskich$^2$ \\
	$^1$Department of Informatics, Ionian University, \\
	Corfu, Greece; andronikos@ionio.gr \\
	$^2$Department of History and Philosophy of Sciences, \\
	National and Kapodistrian University of Athens, \\
	Athens, Greece; asirokof@math.uoa.gr \\
}

\begin{document}

\maketitle

\begin{abstract}
		This paper is inspired by the PQ penny flip game. It employs group-theoretic concepts to study the original game and also its possible extensions. We show that the PQ penny flip game can be associated with the dihedral group $D_{8}$. We prove that within $D_{8}$ there exist precisely two classes of winning strategies for Q. We establish that there are precisely two different sequences of states that can guaranteed Q's win with probability $1.0$. We also show that the game can be played in the all dihedral groups $D_{8 n}$, $n \geq 1$, with any significant change. We examine what happens when Q can draw his moves from the entire $U(2)$ and we conclude that again, there are exactly two classes of winning strategies for Q, each class containing now an infinite number of equivalent strategies, but all of them send the coin through the same sequence of states as before. Finally, we consider general extensions of the game with the quantum player having $U(2)$ at his disposal. We prove that for Q to surely win against Picard, he must make both the first and the last move.

\end{abstract}

\section{Introduction} \label{sec:Introduction}

It is rather unnecessary to stress the importance of game theory. It has been extensively used for decades now to help researchers and practitioners make sense of situations involving conflict, competition, and cooperation. The abstraction of players who antagonize each other in a specified framework by devising elaborate strategies has been employed in the fields of economics, political and social sciences, biology, and, naturally, to computer science. Game theorists have developed an enormous technical machinery for the quantitative assessment of the players' strategies and their payoffs. Of particular significance is the assumption that the players are rational, which means that they seek to maximize their payoffs. In this paper we use only very basic and easy to grasp notions from game theory. These can be found in all standard textbooks, such as \cite{Maschler2020}, \cite{Dixit2015}, and \cite{Myerson1997}. The emergence of the quantum era in information and computation also brought about the creation of the field of quantum game theory. This recent field is devoted to the study of classical games in the quantum setting, giving an exciting new perspective and results that are beyond the grasp of the classical realm.

\subsection{Related work}

The year 1999 was an important milestone for the creation of the field of quantum games. In that year two influential works were published. In a seminal paper Meyer \cite{Meyer1999} introduced the PQ penny flip game, which can be considered the quantum analogue of the classical penny flip game. The other influential work from 1999 was by Eisert et al. in \cite{Eisert1999}. There the authors presented a novel technique, known now as the Eisert-Wilkens-Lewenstein protocol, that has gained wide acceptance in the field.

In Meyer's PQ penny flip game, the two players are the famous tv characters Picard and Q from the tv series Star Trek. They consecutively ``toss'' a quantum coin and if at the end of the game the coin is found heads up Q wins, otherwise Picard wins. There is a metaphor behind the two players: Picard represents the classical player and Q the quantum player. For Picard the game is perceived as the classical penny flip game, but for Q the quantumness of the coin is evident and can be exploited to his advantage. Meyer demonstrated that Q can always win with probability $1.0$ by employing the Hadamard operator. Afterwards, many researchers generalized this game to $n$-dimensional quantum systems. Important results in this direction were obtained by \cite{Wang2000}, \cite{Ren2008}, and \cite{Salimi2009}. These results indicated that under a specific set of rules, the quantum player does have an advantage over the classical player. Nonetheless, this need not always hold as the authors in \cite{Anand2015} pointed out. There it was shown that if the the rules of the PQ penny flip game are appropriately modified, it is even    possible that Picard may win the game. Another related problem, namely that of quantum gambling based on Nash-equilibrium was examined in \cite{Zhang2017}. The association of every finite variant of the PQ penny flip game with finite automata so that strategies are words accepted by the corresponding automaton was established in \cite{Andronikos2018}. In that work the underlying assumption was that Q will always use the Hadamard operator. The present paper is also focused on the PQ penny flip game and its possible extensions, but this time without any limitations, as Q is free to choose his moves from the entire $U(2)$.

With respect to the Eisert-Wilkens-Lewenstein scheme, many important results have been obtained. We mention that several quantum adaptations of the famous prisoners' dilemma have been defined and studied, giving quantum strategies that are better than any classical strategy (\cite{Eisert1999}). Some recent results were presented in \cite{Giannakis2019}, where the correspondence of typical conditional strategies used in the classical repeated prisoners' dilemma game to languages accepted by quantum automata was established, and in \cite{Rycerz2020}, where the Eisert–Wilkens–Lewenstein scheme was extended. Quantum games, especially coin tossing, have also been fruitfully utilized in many quantum cryptographic protocols. In such a setting Alice and Bob assume the role of remote parties that, despite not trusting each other, they have to agree on a random bit (see \cite{Bennett2014} and references therein). This has been extended in \cite{Aharon2010} to quantum dice rolling when multiple outcomes and parties are involved. In a different but quite similar line of thought, Parrondo games were studied via quantum lattice gas automata in \cite{Meyer2002} and in \cite{Giannakis2015a} it was shown that quantum automata accepting infinite words can capture winning strategies for abstract quantum games. Recently abstract sequential quantum game were investigated in \cite{Andronikos2021}. In passing we note that games have been cast not only in a quantum setting, but also in a biological setting. Some well-known classical games, such as the prisoners' dilemma, can be expressed via biological  bio-inspired concepts (see \cite{Kastampolidou2020}, \cite{Theocharopoulou2019} and \cite{Kastampolidou2020a} for more references).

\subsection{Contribution}

This paper is inspired by previous research on the PQ penny flip game. Its novelty is mainly attributed to its use of group-theoretic concepts to study the original game and its possible extensions. We show for the first time, to the best of our knowledge, that the original PQ penny flip game can be associated with the dihedral group $D_{8}$. Interpreting the game in terms of stabilizers and fixed sets, which are basic but helpful group notions, enables us to easily explain and replicate Q's strategy. First, we prove that within $D_{8}$ there exist precisely two classes of winning strategies for Q. Each class contains many different strategies, but all these strategies are equivalent in the sense that they drive the coin through the same sequence of states. We establish for the first time in the literature that there are precisely two different sequences of states that can guaranteed Q's win with probability $1.0$. We then proceed to show that the same game can be played in the all dihedral groups $D_{8 n}$, $n \geq 1$, with any significant change in the winning strategies of Q. This allows us to conclude that in a way the smallest group that captures the essence of the game is $D_{8}$. Subsequently, we examine what happens when Q can draw his moves from the entire $U(2)$. We provide the definitive answer that, perhaps surprisingly, the situation remains the same. Again, there are exactly two classes of winning strategies for Q, each class containing now an infinite number of equivalent strategies, but all of them send the coin through the same sequence of states as before. In a final analysis, the original PQ penny flip game can be succinctly summarized by saying that there are precisely two paths of states that lead to Q's win and, of course, no path that leads to Picard's win. Finally, we consider general extensions of the game without any restrictions in the number of rounds and with the quantum player having $U(2)$ at his disposal. Our examination, will uncover a very important fact, namely that for the quantum player to surely win against the classical player the tremendous advantage of in terms of available quantum actions is not enough. Q must also make both the first and the last move, or else he is not certain to win.

\subsection{Organization}

The paper is structured as follows. Section \ref{sec:Introduction} sets the stage and gives the most relevant references. Section \ref{sec:Background} introduces the notation and terminology used in this article. Section \ref{sec:Connecting PQG & $D_8$} proves the connection of the game with the dihedral group $D_{8}$ and Section \ref{sec:Group Theoretic Analysis of PQG} analyzes Q's strategy in terms of group concepts. Sections \ref{sec:Enlarging the Operational Space of the PQG} and \ref{sec:Extending the PQG} contain the most important results of this work, and, finally, Section \ref{sec:Conclusions} provides a summary of our conclusions and sketches some ideas for possible future work.

\section{Background} \label{sec:Background}

\subsection{The $PQ$ penny flip game}

In what is now regarded as a landmark paper~\cite{Meyer1999}, Meyer defined the \textit{penny flip} game between the famous television personas Picard and Q from the TV series Star Trek. From now on for brevity we shall refer to it as the $PQG$ (the Picard - Q game). This game is much more than a coin flipping game; its importance lies in the fact that it demonstrates the advantage of quantum strategies over classical strategies. The human player Picard can only employ classical strategies, while the quantum player, Q, is capable of using quantum strategies. This asymmetry is the reason why, no matter what Picard does, Q always wins with probability $1.0$. Picard is confined to just the two classical moves available in a 2-dimensional system: he can either do nothing, or he can flip the coin. Doing nothing means that the coin remains in its current state, while flipping the coin changes its state from heads to tails or vice versa. Q's advantage stems from the fact that he can potentially choose from an infinite pool of allowable moves; the only obvious restriction being that his move must be represented by a unitary operator. The game begins with the coin \emph{heads up} and the two players act on the coin following a predetermined order. Q acts first, then Picard and last Q again. If, after Q's last action, the coin is found \emph{heads up}, then Q wins. If the coin is found \emph{tails up}, then Picard wins.

In the context of the $PQG$ and its extensions, it is convenient to employ the terminology outlined in Definition \ref{def:Winning and Dominant Strategies}, adapted from \cite{Maschler2020} and \cite{Myerson1997}. Informally, the word \emph{strategy} implies a \emph{rational} plan on behalf of each player. This plan ultimately consists of actions, or moves that the player makes as the game evolves.

\begin{definition} [Winning and dominant strategies] \label{def:Winning and Dominant Strategies} \
	\begin{itemize}
		\item	A \emph{strategy} is a function that associates an admissible action to every round that the player makes a move. It is convenient to represent strategies as finite sequences of moves from the player's repertoire.
		\item	A strategy $\sigma_{P}$ for Picard is a \emph{winning strategy} if for \emph{every} strategy $\sigma_{Q}$ of Q, Picard wins the game with probability $1.0$.
		\item	Symmetrically, a strategy $\sigma_{Q}$ for Q is a \emph{winning strategy} if for \emph{every} strategy $\sigma_{P}$ of Picard, Q wins the game with probability $1.0$.
		\item	A strategy for Picard or Q is \emph{dominated} if there exists another strategy that has greater probability to win for \emph{every} strategy of the other player. A strategy that dominates all other strategies is called \emph{dominant}.
	\end{itemize}
\end{definition}

Of course, in the original $PQG$, Picard's strategy is just one move, e.g., $( F )$. Q's strategy on the other hand is a sequence of two moves: $(H, H)$. Moreover, a winning strategy for Q is also a dominant strategy. Meyer proved that the use of the Hadamard operator constitutes a \emph{winning} strategy for Q. If Q uses the Hadamard operator, he will win with probability $1.0$, irrespective of Picard's moves.

In more technical terms, the game takes place in the 2-dimensional complex Hilbert space $\mathcal{H}_2$. The computational basis of $\mathcal{H}_2$ is denoted by $B$ and consists of the kets $\ket{0} = \begin{bmatrix} 1 \\ 0 \end{bmatrix}$ and $\ket{1} = \begin{bmatrix} 0 \\ 1 \end{bmatrix}$:

\begin{align} \label{eq:Computational Basis B2}
	B = \{ \ket{0}, \ket{1} \} \ .
\end{align}

Typically, $\ket{0}$ and $\ket{1}$ capture the state of the coin being \emph{heads up} or \emph{tails up}, respectively. Picard's moves \emph{do nothing} and \emph{flip} the coin correspond to the \emph{identity} operator $I$ and the \emph{flip} operator $F$, respectively. As already mentioned, Q's winning strategy is the Hadamard operator $H$. In $\mathcal{H}_2$ the players' moves are represented by the following $2 \times 2$ matrices:

\begin{align} \label{eq:PQ Matrices}
	I =
	\begin{bmatrix}
		1 & 0 \\
		0 & 1
	\end{bmatrix}, \quad
	F =
	\begin{bmatrix}
		0 & 1 \\
		1 & 0
	\end{bmatrix}, \quad \text{and} \quad
	H =
	\begin{bmatrix}
		\begin{array}{lr}
			\frac{\sqrt{2}}{2} & \frac{\sqrt{2}}{2} \\
			\frac{\sqrt{2}}{2} & -\frac{\sqrt{2}}{2}
		\end{array}
	\end{bmatrix} \ .
\end{align}

$F$ is of course one of the famous Pauli matrices, frequently denoted by $\sigma_x$ or $\sigma_1$. In this work we approach the dynamics of the $PQG$ and its extensions by examining the actions available to the players.

\begin{definition} [$PQG$ moves and their composition] \label{def:PQG Moves}
	Let $M_{P} = \{I, F\}$ and $M_{Q} = \{ H \}$ be the sets of permissible moves for Picard and Q, respectively, and let $M = M_{P} \cup M_{Q} = \{I, F, H\}$. The set of all \emph{finite compositions} of moves from $M$, denoted by $M^\star$, is called the \emph{operational space} of the $PQG$.
\end{definition}

Consider for instance the composition $F H$, that can arise in the $PQG$ when Picard replies with $F$ to Q's $H$. A simple matrix multiplication shows that

\begin{align} \label{eq:FH Composition I}
	F H =
	\begin{bmatrix}
		\begin{array}{lr}
			\frac{\sqrt{2}}{2} & - \frac{\sqrt{2}}{2} \\
			\frac{\sqrt{2}}{2} &   \frac{\sqrt{2}}{2}
		\end{array}
	\end{bmatrix} \ .
\end{align}

The operational space $M^\star$ contains not only the above operator~(\ref{eq:FH Composition I}), but also every operator that results from a finite composition of the moves in $M$. Most of them are not realized in the actual $PQG$ because its duration is just 3 rounds; however, this set will provide insight when we consider various extensions of the $PQG$.

Definition \ref{def:PQG Moves} can be generalized as follows: 

\begin{definition} \label{def:V Arbitrary Game Moves}
	Given any game $V$ (e.g., an extension of the original $PQG$ game), for which the set of moves is $M_V$, its operational space is $M_V^\star$.
\end{definition}

\subsection{Dihedral groups}


For the completeness of our presentation we shall recall a few definitions and concepts from group theory. The notation and definitions are based on standard textbooks such as \cite{Artin2011} and \cite{Dummit2004}.

%

\begin{definition}[Group] \label{def:Group Definition}
	A set $G$ equipped with a binary operation $\circ$ is a \emph{group} under $\circ$ if it satisfies the following properties.
	\begin{enumerate}
		\item	There exists an element $\mathds{1} \in G$, the \emph{identity} of $G$, such that $\mathds{1} \circ g = g \circ \mathds{1} = g$ for all $g \in G$.
		\item	For every $g \in G$ there exists an element in $G$, called the \emph{inverse} of $g$ and denoted by $g^{-1}$, such that $g \circ g^{-1} = g^{-1} \circ g = \mathds{1}$.
		\item	For all $f, g, h \in G:$ $(f \circ g) \circ h = f \circ (g \circ h)$, i.e., the \emph{associative} property holds.
	\end{enumerate}
\end{definition}

The number of elements of the group $G$ is called the \emph{order} of $G$ and is denoted by $|G|$. It is customary to employ the following notation regarding powers of an arbitrary element $g$ of a group $G$.

\begin{itemize}
	\item	$g^0 = \mathds{1}$,
	\item	$g^n = \underbrace{ g \circ g \circ \ldots \circ g }_{ \substack{ n \text{ factors } \\ \vspace{0.01 cm} } }$, when $n > 0$, and 
	\item	$g^n = (g^{-1})^{|n|}$, when $n < 0$. 
\end{itemize}

We shall omit the symbol $\circ$ of the binary operation, particularly in view of the fact that in many occasions the group elements will be represented by $2 \times 2$ matrices and the operation $\circ$ will be matrix multiplication. Hence, instead of writing $f \circ g$, we will simply use the juxtaposition of the two elements $f g$.

The groups that capture the symmetries of regular polygons are called dihedral groups. We clarify that by \textit{regular} polygon it is understood that all the sides of the polygon have the same length and all the interior angles are equal. Furthermore, we assume that the center of the regular polygon is located at the origin of the plane.

%

\begin{definition}[Dihedral groups]\label{def:Dihedral Groups Definition}
	The group of symmetries of the \emph{regular} $n$-gon, where $n \geq 3$, is called the dihedral group of order $2n$ and is denoted by $D_n$.\footnote{Many authors denote the dihedral group of order $2n$ by $D_{2n}$ to explicitly indicate its order. However, in this paper we use the notation $D_n$ to emphasize the geometric intuition.}
\end{definition}

Please note that from now on when we refer to an arbitrary dihedral group $D_n$ we shall assume that $n \geq 3$. The group operation is composition of symmetries, i.e., composition of rotations and reflections. The $2n$ symmetries of a regular $n$-gon, where $n \geq 3$, can be categorized as follows.
\begin{itemize}
	\item	There are $n$ rotational symmetries. These are the rotations about the center of the $n$-gon by $\frac{2 \pi k}{n}$, with $k$ taking the values $0, 1, \dots, n-1$. Figures \ref{fig:Rotation Symmetries Regular Heptagon} and \ref{fig:Rotation Symmetries Regular Octagon} show the $7$ and $8$ rotational symmetries of the regular heptagon and octagon, respectively.
	\item	There are also $n$ reflection symmetries.
			\begin{itemize}
				\item	If $n$ is odd these are the reflections in the lines defined by a vertex and the center of the regular $n$-gon. As an example, see Figure \ref{fig:Reflection Symmetries Regular Heptagon} depicting the reflection symmetries of the regular heptagon.
				\item	If $n$ is even, these are $\frac{n}{2}$ reflections in the lines through opposite vertices and $\frac{n}{2}$ reflections in the lines passing through midpoints of opposite faces. An example that will play an important role in the rest of our study is given in Figure \ref{fig:Reflection Symmetries Regular Octagon}, showing the reflection symmetries of the regular octagon.
			\end{itemize}
			By fixing a vertex of the regular $n$-gon to lie on the $x$-axis (such a vertex $1$ in Figures \ref{fig:Reflection Symmetries Regular Heptagon} and \ref{fig:Reflection Symmetries Regular Octagon}) and the center of the $n$-gon at the origin of the plane, we may surmise that the $n$ reflection symmetries correspond to lines through the origin making an angle $\frac{\pi k}{n}$ with the positive $x$-axis, with $k$ taking the values $0, 1, \dots, n-1$.
\end{itemize}

\begin{figure}[H]
	\begin{minipage}[b]{0.475\textwidth}
		\centering
		\begin{tikzpicture}[scale = 2.5]
			\def \angle {360/7}
			\draw [fill, thick, WordBlueDark]
			({cos(0 * \angle)}, {sin(0 * \angle)}) circle (0.75 pt) node [right]{1} --
			({cos(1 * \angle)}, {sin(1 * \angle)}) circle (0.75 pt) node [above right]{2} --
			({cos(2 * \angle)}, {sin(2 * \angle)}) circle (0.75 pt) node [above]{3} --
			({cos(3 * \angle)}, {sin(3 * \angle)}) circle (0.75 pt) node [above left]{4} --
			({cos(4 * \angle)}, {sin(4 * \angle)}) circle (0.75 pt) node [left = 0.1 cm]{5} --
			({cos(5 * \angle)}, {sin(5 * \angle)}) circle (0.75 pt) node [below left]{6} --
			({cos(6 * \angle)}, {sin(6 * \angle)}) circle (0.75 pt) node [below = 0.1 cm]{7} --
			({cos(0 * \angle)}, {sin(0 * \angle)});
			\draw[-, ultra thick, RedPurple] (0, 0) -- ({cos(0 * \angle)}, {sin(0 * \angle)});
			\draw[-, ultra thick, RedPurple] (0, 0) -- ({cos(1 * \angle)}, {sin(1 * \angle)});
			\draw[-, ultra thick, RedPurple] (0, 0) -- ({cos(2 * \angle)}, {sin(2 * \angle)});
			\draw[-, ultra thick, RedPurple] (0, 0) -- ({cos(3 * \angle)}, {sin(3 * \angle)});
			\draw[-, ultra thick, RedPurple] (0, 0) -- ({cos(4 * \angle)}, {sin(4 * \angle)});
			\draw[-, ultra thick, RedPurple] (0, 0) -- ({cos(5 * \angle)}, {sin(5 * \angle)});
			\draw[-, ultra thick, RedPurple] (0, 0) -- ({cos(6 * \angle)}, {sin(6 * \angle)});
			\scoped [on background layer]
			\filldraw [->, RedPurple!20, line width = 0.3 mm] (0, 0) -- (0.85,0) arc (0:\angle:0.85);
			\draw [->, RedPurple!70, line width = 0.3 mm] (0.85,0) arc (0:\angle:0.85);
			\draw [->, RedPurple!70, line width = 0.3 mm] (0.75,0) arc (0:2*\angle:0.75);
			\draw [->, RedPurple!70, line width = 0.3 mm] (0.65,0) arc (0:3*\angle:0.65);
			\draw [->, RedPurple!70, line width = 0.3 mm] (0.55,0) arc (0:4*\angle:0.55);
			\draw [->, RedPurple!70, line width = 0.3 mm] (0.45,0) arc (0:5*\angle:0.45);
			\draw [->, RedPurple!70, line width = 0.3 mm] (0.35,0) arc (0:6*\angle:0.35);
			\draw [->, RedPurple!70, line width = 0.3 mm] (0.25,0) arc (0:7*\angle:0.25);
			\draw [->, RedPurple!70, line width = 0.3 mm] ({1.2 * cos(0.5 * \angle)}, {1.2 * sin(0.5 * \angle)}) node [RedPurple, right] {\Large $\frac{2 \pi}{7}$} -- ({0.85 * cos(0.5 * \angle)}, {0.85 * sin(0.5 * \angle)});
		\end{tikzpicture}
		\caption{The rotational symmetries of the regular heptagon.} \label{fig:Rotation Symmetries Regular Heptagon}
	\end{minipage}
	\hfill
	\begin{minipage}[b]{0.475\textwidth}
		\centering
		\begin{tikzpicture}[scale = 2.5]
			\def \angle {360/7}
			\draw [fill, thick, WordBlueDark]
			({cos(0 * \angle)}, {sin(0 * \angle)}) circle (0.75 pt) node [above right]{1} --
			({cos(1 * \angle)}, {sin(1 * \angle)}) circle (0.75 pt) node [above = 0.1 cm]{2} --
			({cos(2 * \angle)}, {sin(2 * \angle)}) circle (0.75 pt) node [above left]{3} --
			({cos(3 * \angle)}, {sin(3 * \angle)}) circle (0.75 pt) node [below left]{4} --
			({cos(4 * \angle)}, {sin(4 * \angle)}) circle (0.75 pt) node [below = 0.1 cm]{5} --
			({cos(5 * \angle)}, {sin(5 * \angle)}) circle (0.75 pt) node [below right]{6} --
			({cos(6 * \angle)}, {sin(6 * \angle)}) circle (0.75 pt) node [right = 0.1 cm]{7} --
			({cos(0 * \angle)}, {sin(0 * \angle)});
			\draw [-, dashed, thick, GreenTeal] ({1.4 * cos(0 * \angle)}, {1.4 * sin(0 * \angle)}) -- ({1.4 * cos(3.5 * \angle)}, {1.4 * sin(3.5 * \angle)});
			\draw [<->, GreenTeal!70, line width = 0.4 mm] (1.4, 0.125) to [out = 150, in = 210, looseness = 3] (1.4, -0.125);
			\draw [-, dashed, thick, GreenTeal] ({1.4 * cos(0.5 * \angle)}, {1.4 * sin(0.5 * \angle)}) -- ({1.4 * cos(4 * \angle)}, {1.4 * sin(4 * \angle)});
			\draw [<->, GreenTeal!70, line width = 0.4 mm, shift = {(0.0 cm, 0.0 cm)}, rotate = \angle / 2] (1.4, 0.125) to [out = 150, in = 210, looseness = 3] (1.4, -0.125);
			\draw [-, dashed, thick, GreenTeal] ({1.4 * cos(1 * \angle)}, {1.4 * sin(1 * \angle)}) -- ({1.4 * cos(4.5 * \angle)}, {1.4 * sin(4.5 * \angle)});
			\draw [<->, GreenTeal!70, line width = 0.4 mm, shift = {(0.0 cm, 0.0 cm)}, rotate = \angle] (1.4, 0.125) to [out = 150, in = 210, looseness = 3] (1.4, -0.125);
			\draw [-, dashed, thick, GreenTeal] ({1.4 * cos(1.5 * \angle)}, {1.4 * sin(1.5 * \angle)}) -- ({1.4 * cos(5 * \angle)}, {1.4 * sin(5 * \angle)});
			\draw [<->, GreenTeal!70, line width = 0.4 mm, shift = {(0.0 cm, 0.0 cm)}, rotate = 3* (\angle / 2)] (1.4, 0.125) to [out = 150, in = 210, looseness = 3] (1.4, -0.125);
			\draw [-, dashed, thick, GreenTeal] ({1.4 * cos(2 * \angle)}, {1.4 * sin(2 * \angle)}) -- ({1.4 * cos(5.5 * \angle)}, {1.4 * sin(5.5 * \angle)});
			\draw [<->, GreenTeal!70, line width = 0.4 mm, shift = {(0.0 cm, 0.0 cm)}, rotate = 2 * \angle] (1.4, 0.125) to [out = 150, in = 210, looseness = 3] (1.4, -0.125);
			\draw [-, dashed, thick, GreenTeal] ({1.4 * cos(2.5 * \angle)}, {1.4 * sin(2.5 * \angle)}) -- ({1.4 * cos(6 * \angle)}, {1.4 * sin(6 * \angle)});
			\draw [<->, GreenTeal!70, line width = 0.4 mm, shift = {(0.0 cm, 0.0 cm)}, rotate = 5* (\angle / 2)] (1.4, 0.125) to [out = 150, in = 210, looseness = 3] (1.4, -0.125);
			\draw [-, dashed, thick, GreenTeal] ({1.4 * cos(3 * \angle)}, {1.4 * sin(3 * \angle)}) -- ({1.4 * cos(6.5 * \angle)}, {1.4 * sin(6.5 * \angle)});
			\draw [<->, GreenTeal!70, line width = 0.4 mm, shift = {(0.0 cm, 0.0 cm)}, rotate = 3* \angle] (1.4, 0.125) to [out = 150, in = 210, looseness = 3] (1.4, -0.125);
		\end{tikzpicture}
		\caption{The reflection symmetries of the regular heptagon.} \label{fig:Reflection Symmetries Regular Heptagon}
	\end{minipage}
\end{figure}

\begin{figure}[H]
	\begin{minipage}[b]{0.475\textwidth}
		\centering
		\begin{tikzpicture}[scale = 2.5]
			\def \angle {360/8}
			\draw [fill, thick, WordBlueDark]
			({cos(0 * \angle)}, {sin(0 * \angle)}) circle (0.75 pt) node [right]{1} --
			({cos(1 * \angle)}, {sin(1 * \angle)}) circle (0.75 pt) node [above right]{2} --
			({cos(2 * \angle)}, {sin(2 * \angle)}) circle (0.75 pt) node [above]{3} --
			({cos(3 * \angle)}, {sin(3 * \angle)}) circle (0.75 pt) node [above left]{4} --
			({cos(4 * \angle)}, {sin(4 * \angle)}) circle (0.75 pt) node [left]{5} --
			({cos(5 * \angle)}, {sin(5 * \angle)}) circle (0.75 pt) node [below left]{6} --
			({cos(6 * \angle)}, {sin(6 * \angle)}) circle (0.75 pt) node [below]{7} --
			({cos(7 * \angle)}, {sin(7 * \angle)}) circle (0.75 pt) node [below right]{8} --
			({cos(0 * \angle)}, {sin(0 * \angle)});
			\draw[-, ultra thick, RedPurple] (0, 0) -- ({cos(0 * \angle)}, {sin(0 * \angle)});
			\draw[-, ultra thick, RedPurple] (0, 0) -- ({cos(1 * \angle)}, {sin(1 * \angle)});
			\draw[-, ultra thick, RedPurple] (0, 0) -- ({cos(2 * \angle)}, {sin(2 * \angle)});
			\draw[-, ultra thick, RedPurple] (0, 0) -- ({cos(3 * \angle)}, {sin(3 * \angle)});
			\draw[-, ultra thick, RedPurple] (0, 0) -- ({cos(4 * \angle)}, {sin(4 * \angle)});
			\draw[-, ultra thick, RedPurple] (0, 0) -- ({cos(5 * \angle)}, {sin(5 * \angle)});
			\draw[-, ultra thick, RedPurple] (0, 0) -- ({cos(6 * \angle)}, {sin(6 * \angle)});
			\draw[-, ultra thick, RedPurple] (0, 0) -- ({cos(7 * \angle)}, {sin(7 * \angle)});
			\scoped [on background layer]
				\filldraw [->, RedPurple!20, line width = 0.3 mm] (0, 0) -- (0.85,0) arc (0:\angle:0.85);
			\draw [->, RedPurple!70, line width = 0.3 mm] (0.85,0) arc (0:\angle:0.85);
			\draw [->, RedPurple!70, line width = 0.3 mm] (0.75,0) arc (0:2*\angle:0.75);
			\draw [->, RedPurple!70, line width = 0.3 mm] (0.65,0) arc (0:3*\angle:0.65);
			\draw [->, RedPurple!70, line width = 0.3 mm] (0.55,0) arc (0:4*\angle:0.55);
			\draw [->, RedPurple!70, line width = 0.3 mm] (0.45,0) arc (0:5*\angle:0.45);
			\draw [->, RedPurple!70, line width = 0.3 mm] (0.35,0) arc (0:6*\angle:0.35);
			\draw [->, RedPurple!70, line width = 0.3 mm] (0.25,0) arc (0:7*\angle:0.25);
			\draw [->, RedPurple!70, line width = 0.3 mm] (0.15,0) arc (0:8*\angle:0.15);
			\draw [->, RedPurple!70, line width = 0.3 mm] ({1.2 * cos(0.5 * \angle)}, {1.2 * sin(0.5 * \angle)}) node [RedPurple, right] {\Large $\frac{2 \pi}{8}$} -- ({0.85 * cos(0.5 * \angle)}, {0.85 * sin(0.5 * \angle)});
		\end{tikzpicture}
		\caption{The rotational symmetries of the regular octagon.} \label{fig:Rotation Symmetries Regular Octagon}
	\end{minipage}
	\hfill
	\begin{minipage}[b]{0.475\textwidth}
		\centering
		\begin{tikzpicture}[scale = 2.5]
			\def \angle {360/8}
			\draw [fill, thick, WordBlueDark]
			({cos(0 * \angle)}, {sin(0 * \angle)}) circle (0.75 pt) node [xshift = 0.2 cm, yshift = 0.3 cm]{1} --
			({cos(1 * \angle)}, {sin(1 * \angle)}) circle (0.75 pt) node [above = 0.1 cm]{2} --
			({cos(2 * \angle)}, {sin(2 * \angle)}) circle (0.75 pt) node [above left]{3} --
			({cos(3 * \angle)}, {sin(3 * \angle)}) circle (0.75 pt) node [left = 0.1 cm]{4} --
			({cos(4 * \angle)}, {sin(4 * \angle)}) circle (0.75 pt) node [above left]{5} --
			({cos(5 * \angle)}, {sin(5 * \angle)}) circle (0.75 pt) node [above = 0.1 cm,  left = 0.1 cm]{6} --
			({cos(6 * \angle)}, {sin(6 * \angle)}) circle (0.75 pt) node [xshift = 0.2 cm, yshift = -0.3 cm]{7} --
			({cos(7 * \angle)}, {sin(7 * \angle)}) circle (0.75 pt) node [right = 0.1 cm]{8} --
			({cos(0 * \angle)}, {sin(0 * \angle)});
			\draw [fill, thick, WordBlueDark] ({0.92 * cos(0.5 * \angle)}, {0.92 * sin(0.5 * \angle)}) circle (0.75 pt) node [WordBlueDark, xshift = 0.1 cm, yshift = 0.3 cm] {$A$};
			\draw [fill, thick, WordBlueDark] ({0.92 * cos(4.5 * \angle)}, {0.92 * sin(4.5 * \angle)}) circle (0.75 pt) node [WordBlueDark, xshift = -0.35 cm, yshift = 0.1 cm] {$A'$};
			\draw [-, dashed, thick, GreenTeal] ({1.4 * cos(0 * \angle)}, {1.4 * sin(0 * \angle)}) -- ({1.4 * cos(4 * \angle)}, {1.4 * sin(4 * \angle)});
			\draw [<->, GreenTeal!70, line width = 0.4 mm] (1.4, 0.125) to [out = 150, in = 210, looseness = 3] (1.4, -0.125);
			\draw [-, dashed, thick, GreenTeal] ({1.4 * cos(0.5 * \angle)}, {1.4 * sin(0.5 * \angle)}) -- ({1.4 * cos(4.5 * \angle)}, {1.4 * sin(4.5 * \angle)});
			\draw [<->, GreenTeal!70, line width = 0.4 mm, shift = {(0.0 cm, 0.0 cm)}, rotate = 22.5] (1.4, 0.125) to [out = 150, in = 210, looseness = 3] (1.4, -0.125);
			\draw [-, dashed, thick, GreenTeal] ({1.4 * cos(1 * \angle)}, {1.4 * sin(1 * \angle)}) -- ({1.4 * cos(5 * \angle)}, {1.4 * sin(5 * \angle)});
			\draw [<->, GreenTeal!70, line width = 0.4 mm, shift = {(0.0 cm, 0.0 cm)}, rotate = 45] (1.4, 0.125) to [out = 150, in = 210, looseness = 3] (1.4, -0.125);
			\draw [-, dashed, thick, GreenTeal] ({1.4 * cos(1.5 * \angle)}, {1.4 * sin(1.5 * \angle)}) -- ({1.4 * cos(5.5 * \angle)}, {1.4 * sin(5.5 * \angle)});
			\draw [<->, GreenTeal!70, line width = 0.4 mm, shift = {(0.0 cm, 0.0 cm)}, rotate = 67.5] (1.4, 0.125) to [out = 150, in = 210, looseness = 3] (1.4, -0.125);
			\draw [-, dashed, thick, GreenTeal] ({1.4 * cos(2 * \angle)}, {1.4 * sin(2 * \angle)}) -- ({1.4 * cos(6 * \angle)}, {1.4 * sin(6 * \angle)});
			\draw [<->, GreenTeal!70, line width = 0.4 mm, shift = {(0.0 cm, 0.0 cm)}, rotate = 90] (1.4, 0.125) to [out = 150, in = 210, looseness = 3] (1.4, -0.125);
			\draw [-, dashed, thick, GreenTeal] ({1.4 * cos(2.5 * \angle)}, {1.4 * sin(2.5 * \angle)}) -- ({1.4 * cos(6.5 * \angle)}, {1.4 * sin(6.5 * \angle)});
			\draw [<->, GreenTeal!70, line width = 0.4 mm, shift = {(0.0 cm, 0.0 cm)}, rotate = 112.5] (1.4, 0.125) to [out = 150, in = 210, looseness = 3] (1.4, -0.125);
			\draw [-, dashed, thick, GreenTeal] ({1.4 * cos(3 * \angle)}, {1.4 * sin(3 * \angle)}) -- ({1.4 * cos(7 * \angle)}, {1.4 * sin(7 * \angle)});
			\draw [<->, GreenTeal!70, line width = 0.4 mm, shift = {(0.0 cm, 0.0 cm)}, rotate = 135] (1.4, 0.125) to [out = 150, in = 210, looseness = 3] (1.4, -0.125);
			\draw [-, dashed, thick, GreenTeal] ({1.4 * cos(3.5 * \angle)}, {1.4 * sin(3.5 * \angle)}) -- ({1.4 * cos(7.5 * \angle)}, {1.4 * sin(7.5 * \angle)});
			\draw [<->, GreenTeal!70, line width = 0.4 mm, shift = {(0.0 cm, 0.0 cm)}, rotate = 157.5] (1.4, 0.125) to [out = 150, in = 210, looseness = 3] (1.4, -0.125);
		\end{tikzpicture}
		\caption{The reflection symmetries of the regular octagon.} \label{fig:Reflection Symmetries Regular Octagon}
	\end{minipage}
\end{figure}

The general dihedral group contains the following $2n$ elements (for details the interested reader may consult \cite{Artin2011}, \cite{Dummit2004} or \cite{Lovett2015})

\begin{align} \label{def:General Dihedral Group Elements}
	D_n = \{ \mathds{1}, r, r^2, \dots, r^{n - 1}, s, r s, r^2 s, \dots, r^{n - 1} s \} \ ,
\end{align}

where $r$ is the rotation by $\frac{2 \pi}{n}$ and $s$ is \emph{any} reflection. It is evident that each element of $D_n$ can be \emph{uniquely} written as $r^{k} s^{l}$ for some $k, 0 \leq k \leq n - 1,$ and $l$, where $l = 0$ or $1$. Elements $\mathds{1}, r, r^2, \dots, r^{n - 1}$ are rotations, i.e., $r^k$ is the rotation by $\frac{2 \pi k}{n}$, and elements $s, r s, r^2 s, \dots, r^{n - 1} s$ are reflections.


In particular, the dihedral group $D_8$ contains the $16$ elements

\begin{align} \label{def:D8 Group Elements}
	D_8 = \{ \mathds{1}, r, r^2, \dots, r^{7}, s, r s, r^2 s, \dots, r^{7} s \} \ ,
\end{align}

where $r$ is the rotation by $\frac{2 \pi}{8}$ and $s$ is \emph{any} reflection. We remark that, referring to Figure \ref{fig:Reflection Symmetries Regular Octagon}, $s$ can be taken to be the reflection in the line passing through the vertices $1$ and $5$, or the reflection in the line passing through the midpoints $A$ and $A'$, or the reflection in the line passing through the vertices $2$ and $6$, or any of the remaining reflections.

\begin{definition}[Generators] \label{def:Group Generators}
	Given a subset $X$ of a group $G$, the \emph{smallest subgroup} of $G$ that contains $X$ is denoted by $\expval{X}$. The elements of $X$ are called \emph{generators} for $\expval{X}$. 
\end{definition}

When $X$ is finite, i.e., $X = \{ x_1, \dots, x_n \}$, as will be the case in this work, it is customary to simply write $\expval{ x_1, \dots, x_n }$.

A typical way to specify a group is by giving a \emph{presentation} for the group. This amounts to using \emph{generators} and \emph{relations}, with the understanding that all group elements can be constructed as products of powers of the generators, and that the relations are equations involving the generators and the group identity. The following presentation of $D_n$ is especially convenient for our analysis:

\begin{align} \label{eq:Dihedral Group Presentation 1}
	D_n = \langle s, t \ | \ s^2 = t^2 = (s t)^{n} = \mathds{1} \rangle \ . \tag{$P_{1}$}
\end{align}

This presentation demonstrates that $D_n$ can be generated by two reflections $s$ and $t$. It appears in \cite{Artin2011} and \cite{Meier2011}, among others, where it is clarified that $D_n$ can be generated by two reflections $s, t$ in adjacent axes of symmetry passing though the origin and intersecting in an angle $\frac{\pi}{n}$. In this case, the product $s t$ is a rotation through an angle of $\pm \frac{2 \pi}{n}$. We note though that presentations are not unique. For instance, one other widely used presentation for the dihedral groups is $D_n = \langle r, s \ | \ r^n = s^2 = \mathds{1}, r s = s r^{-1} \rangle$.

\section{The connection between $PQG$ and $D_8$} \label{sec:Connecting PQG & $D_8$}

\subsection{Matrix representations of rotations and reflections}

A useful and quite common way to represent rotations and reflections in the plane is to use $2 \times 2$ matrices. Such matrices, which are often called \emph{rotators} and \emph{reflectors}, can be conveniently written in a form that is easy to recognize and manipulate (see \cite{Lovett2015}, \cite{Meyer2000} and \cite{Anton2013} for more details). A rotator representing a counterclockwise rotation through an angle $\varphi$ about the origin is denoted by $R_{\varphi}$ and, similarly, a reflector about a line through the origin that makes an angle $\varphi$ with the positive $x$-axis is denoted by $S_{\varphi}$. $R_{\varphi}$ and $S_{\varphi}$ are given by the formulas shown below. Please note that we use capital $R$ and capital $S$ to designate these $2 \times 2$ matrices, in order to avoid any confusion with the elements of the dihedral group that are denoted by small $r$ and $s$.

\begin{multicols}{2}
	\noindent
	\begin{align} \label{eq:2D Rotator}
		R_{\varphi} =
		\begin{bmatrix}
			\begin{array}{lr}
				\cos \varphi & -\sin \varphi \\
				\sin \varphi & \cos \varphi
			\end{array}
		\end{bmatrix}
	\end{align}
	\begin{align} \label{eq:2D Reflector}
		S_{\varphi} =
		\begin{bmatrix}
			\begin{array}{lr}
				\cos 2\varphi & \sin 2\varphi \\
				\sin 2\varphi & -\cos 2\varphi
			\end{array}
		\end{bmatrix}
	\end{align}
\end{multicols}

It is now quite straightforward to see that the $F$ and $H$ operators can be written as follows.

\begin{multicols}{2}
	\noindent
	\begin{align} \label{eq:The F Reflector}
		F = S_{\frac{2 \pi}{8}} =
		\begin{bmatrix}
			\begin{array}{lr}
				\cos 2 \frac{2 \pi}{8} & \sin 2 \frac{2 \pi}{8} \\
				\sin 2 \frac{2 \pi}{8} & -\cos 2 \frac{2 \pi}{8}
			\end{array}
		\end{bmatrix}
	\end{align}
	\begin{align} \label{eq:The H Reflector}
		H = S_{\frac{\pi}{8}} =
		\begin{bmatrix}
			\begin{array}{lr}
				\cos 2 \frac{\pi}{8} & \sin 2 \frac{\pi}{8} \\
				\sin 2 \frac{\pi}{8} & -\cos 2 \frac{\pi}{8}
			\end{array}
		\end{bmatrix}
	\end{align}
\end{multicols}

This form reveals that both are \emph{reflectors}: $F$ reflects about a line that makes an angle $\frac{\pi}{4}$ with the positive $x$-axis. To be exact this is the line passing through the vertices $2$ and $6$ in Figure \ref{fig:Reflection Symmetries Regular Octagon}. Likewise, $H$ reflects about a line that makes an angle $\frac{\pi}{8}$ with the positive $x$-axis, which is the line passing through the midpoints $A$ and $A'$ in Figure \ref{fig:Reflection Symmetries Regular Octagon}. Hence, their axes of symmetry intersect in an angle $\frac{\pi}{8}$, as shown in Figure \ref{fig:Reflection Symmetries Regular Octagon}. Moreover, their product $F H$, which is given in (\ref{eq:FH Composition I}), is just the rotator $R_{\frac{2 \pi}{8}}$, as can be verified by employing formula (\ref{eq:2D Rotator}). Therefore, by invoking the presentation (\ref{eq:Dihedral Group Presentation 1}), associating $s$ to $F$, and $t$ to $H$, or vice versa, it becomes evident that $F$ and $H$ generate the dihedral group $D_8$. This conclusion is stated as Theorem \ref{thr:PQG Ambient Group}.

Please note that in an effort to enhance the readability of this paper, without worrying about the technical details, we have relocated all the proofs in the Appendix.

\begin{definition}[The ambient group] \label{def:Ambient Group}
	Let $V$ be a game with operational space $M_V^\star$. If $M_V^\star$ is \emph{isomorphic} to the group $G$, then $G$ is called the \emph{ambient group} of the game $V$.
\end{definition}

\begin{theorem}[The ambient group of the $PQG$] \label{thr:PQG Ambient Group}
	The ambient group of the $PQG$ is $D_8$.
\end{theorem}

The above result tells us that Picard and Q's moves generate the group $D_8$. This has important ramifications. As long as the two players are allowed to use only the aforementioned actions, no matter what specific game they play, the game will take place in the $D_8$ group. Every conceivable composition of moves by the players is just an element of $D_8$. Therefore, although the rules of the game can change dramatically, e.g., the players' turn, the number of rounds, etc., the available moves will always be elements of $D_8$.

Actually, it is a well-known fact that every element of the dihedral group $D_n$ can be represented by $2 \times 2$ matrices of the form shown in (\ref{eq:2D Rotator}) and (\ref{eq:2D Reflector}). One such representation is given below. In the literature it is usually referred to as the \emph{standard} representation of $D_n$. We remark that, in more technical terms, this is a faithful irreducible representation of dimension 2. To clear any potential misunderstanding, let us emphasize that in the standard representation $s$ corresponds to the reflection in the line passing through the vertices $1$ and $5$, i.e., the $x$-axis of Figure \ref{fig:Reflection Symmetries Regular Octagon}.

%

\begin{multicols}{2}
	\noindent
	\begin{align} \label{eq:Standard Representation r in D_n}
		r \mapsto R_{\frac{2 \pi}{n}} =
		\begin{bmatrix}
			\begin{array}{lr}
				\cos \frac{2 \pi}{n} & -\sin \frac{2 \pi}{n} \\
				\sin \frac{2 \pi}{n} & \cos \frac{2 \pi}{n}
			\end{array}
		\end{bmatrix} 
	\end{align}
	\begin{align} \label{eq:Standard Representation s in D_n}
		s \mapsto S_0 =
		\begin{bmatrix}
			\begin{array}{lr}
				1 & 0 \\
				0 & -1
			\end{array}
		\end{bmatrix}
	\end{align}
\end{multicols}

The above mapping of $r$ and $s$ uniquely determines the standard representation of the remaining reflections and rotations of $D_n$.

\begin{multicols}{2}
	\noindent
	\begin{align} \label{eq:Standard Representation r^k in D_n}
		r^{k} \mapsto R_{\frac{2 \pi k}{n}} =
		\begin{bmatrix}
			\begin{array}{lr}
				\cos \frac{2 \pi k}{n} & -\sin \frac{2 \pi k}{n} \\
				\sin \frac{2 \pi k}{n} & \cos \frac{2 \pi k}{n}
			\end{array}
		\end{bmatrix}
	\end{align}
	\begin{align} \label{eq:Standard Representation r^k s in D_n}
		r^{k} s \mapsto S_{\frac{\pi k}{n}} =
		\begin{bmatrix}
			\begin{array}{lr}
				\cos \frac{2 \pi k}{n} & \sin \frac{2 \pi k}{n} \\
				\sin \frac{2 \pi k}{n} & -\cos \frac{2 \pi k}{n}
			\end{array}
		\end{bmatrix} \ ,
	\end{align}
\end{multicols}

where $0 \leq k \leq n - 1$.

%

\subsection{Orbits and stabilizers}

\begin{definition}[Group action] \label{def:Group Action}
	Let $G$ be a group and let $X$ be a nonempty set. A \emph{group action} $\star$ of $G$ on $X$ is a function $\star : G \times X \rightarrow X$ that satisfies the following properties.
	\begin{enumerate}
		\setlength{\itemindent}{0.5 em}
		\item[\emph{$(A_1)$}] $\mathds{1} \star x = x$ for every $x \in X$.
		\item[\emph{$(A_2)$}] $g_1 \star \left( g_2 \star x \right) = \left( g_1 g_2 \right) \star x$, for all $g_1, g_2 \in G$ and all $x \in X$.
	\end{enumerate}
\end{definition}

Under the standard representation of $D_n$, its action on a state of the quantum coin is computed by simply multiplying every matrix corresponding to an element of $D_n$ with the ket describing the state of the coin. In what follows, in addition to speaking about an action, we shall occasionally say that $G$ \emph{acts} on $X$. Moreover, we shall just write $g x$ instead of $g \star x$, since the action we study in this paper is that of operators on kets, or, if you prefer, of matrix-vector multiplication.

\begin{definition}[Orbits and stabilizers] \label{def:Orbits & Stabilizers}
	Suppose that a group $G$ of linear operators, or their corresponding matrix representations, acts on a nonempty set of kets $X$. We make the next definitions, always taking into account that all kets of the form $e^{i \theta} \ket{ \psi }$, with $\theta \in \mathbb{R}$, represent ket $\ket{ \psi }$.
	\begin{enumerate}
		\item	Given $x \in X$, the $G$-\emph{orbit} of $x$, denoted by $G \star x$, is the set $\{ g \star x \in X : g \in G\}$.
		\item	Given $S \subset X$, the \emph{$G$-orbit} of $S$, denoted by $G \star S$, is the union of the orbits $G \star x$, for each $x \in S$
		\item	Given $x \in X$, the \emph{stabilizer} of $x$, denoted by $G ( x )$, is the set $\{ g \in G : g \star x = x \}$.
		\item	Given $g \in G$, the \emph{fixed} set of $g$, denoted by $Fix ( g )$, is the set $\{ x \in X : g \star x = x \}$.
		\item	Given $X \subset G$, the \emph{fixed} set of $X$, denoted by $Fix ( X )$, is the intersection of the fixed sets $Fix ( g )$, for each $g \in X$.
	\end{enumerate}
\end{definition}

In the next section we shall employ these tools in the analysis of Q's strategy to gain insight from a group theoretic perspective.

\section{Analyzing Q's strategy in terms of groups} \label{sec:Group Theoretic Analysis of PQG}

We proceed now to interpret the $PQG$ using the aforementioned groups concepts. It will be helpful to utilize the following abbreviations, which are very common in the literature.

\begin{align} \label{eq:State +}
	\ket{+} = \frac{1}{\sqrt{2}} \left( \ket{0} + \ket{1} \right)
\end{align}

\begin{align} \label{eq:State -}
	\ket{-} = \frac{1}{\sqrt{2}} \left( \ket{0} - \ket{1} \right)
\end{align}

Let us first see what is the effect of the action of $D_{8}$ on the computational basis $B$. One easy way to do this is geometrically by consulting Figures \ref{fig:Rotation Symmetries Regular Octagon} and \ref{fig:Reflection Symmetries Regular Octagon} to see where vertices $1$ and $3$ are sent when being acted upon by the elements of $D_{8}$. Alternatively, one can arrive at the same result algebraically simply by multiplying the matrix representation of every member of $D_{8}$ with $\ket{0}$ and $\ket{1}$. The representations of the elements of $D_8$ can be readily found by setting $n = 8$ in the more general formulas (\ref{eq:Standard Representation r^k in D_n}) and (\ref{eq:Standard Representation r^k s in D_n}). In any event, for future reference we summarize the action of $D_{8}$ on the computational basis $B$ in Proposition \ref{prp:The Action of D_8 on B} (recall that $e^{i \theta} \ket{ \psi }$, with $\theta \in \mathbb{R}$, and $\ket{ \psi }$ represent the same state).

\begin{proposition}[The action of $D_{8}$ on $B$] \label{prp:The Action of D_8 on B} \
	\begin{enumerate}
		\item	$\ket{0}$ and $\ket{1}$ have the same orbit:
		\begin{align} \label{eq:Orbits of Ket 0 and Ket 1 in D_8}
			D_{8} \star \ket{0} = D_{8} \star \ket{1} = \{ \ket{0}, \ket{+}, \ket{1}, \ket{-} \} \ .
		\end{align}
		\item	The orbit of $B$ is:
		\begin{align} \label{eq:Orbit of B in D_8}
			D_{8} \star B = \{ \ket{0}, \ket{+}, \ket{1}, \ket{-} \} \ .
		\end{align}
	\end{enumerate}
\end{proposition}

Q's first move aims to drive the coin into the state

\begin{align} \label{eq:Q's First Move}
	H \ket{0} = \ket{+} \ .
\end{align}

Definition \ref{def:Orbits & Stabilizers} is helpful in understanding the advantage of Q's move in terms of group notions. In particular, there are certain elements of $D_8$ whose action on $\ket{+}$ has no effect whatsoever and which constitute the stabilizer of $\ket{+}$. These can be easily found either geometrically or algebraically, and are listed in Proposition \ref{prp:The Stabilizers of 0, +, 1, - in D_8}.

\begin{proposition}[The stabilizers of $\ket{0}, \ket{+}, \ket{1}$ and $\ket{-}$ in $D_8$] \label{prp:The Stabilizers of 0, +, 1, - in D_8} \
	\begin{itemize}
		\item	The stabilizers of $\ket{0}$ and $\ket{1}$ in $D_8$ are
		\begin{align} \label{eq:The Stabilizers of Ket 0 & Ket 1 in D_8}
			D_{8} ( \ket{0} ) = \{ I, R_{\pi}, S_{0}, S_{\frac{4 \pi}{8}} \}
			\quad \text{and} \quad
			D_{8} ( \ket{1} ) = \{ I, R_{\pi}, S_{0}, S_{\frac{4 \pi}{8}} \} \ .
		\end{align}
		\item	The stabilizers of $\ket{+}$ and $\ket{-}$ are
		\begin{align} \label{eq:The Stabilizers of + & - in D_8}
			D_{8} ( \ket{+} ) = \{ I, R_{\pi}, F, S_{\frac{6 \pi}{8}} \}
			\quad \text{and} \quad
			D_{8} ( \ket{-} ) = \{ I, R_{\pi}, F, S_{\frac{6 \pi}{8}} \} \ .
		\end{align}
	\end{itemize}
\end{proposition}

In a complementary manner, we may surmise that Picard's set of moves fixes specific states in $\mathcal{H}_2$, as demonstrated in Proposition \ref{prp:The Fixed Set of M_P in D_8}.

\begin{proposition}[The fixed set of $\{ I, F \}$ in $D_8$] \label{prp:The Fixed Set of M_P in D_8} \
	\begin{enumerate}
	\item	The fixed set of $F$ in $D_8$ is the set
			\begin{align} \label{eq:The Fixed Set of F in D_8}
			Fix ( F ) = \{ \ket{+}, \ket{-} \} \ .
			\end{align}
	\item	The fixed set of $M_P = \{ I, F \}$ in $D_8$ is the set
			\begin{align} \label{eq:The Fixed Set of M_P in D_8}
				Fix ( \{ I, F \} ) = \{ \ket{+}, \ket{-} \} \ .
			\end{align}
	\end{enumerate}
\end{proposition}

Proposition \ref{prp:The Stabilizers of 0, +, 1, - in D_8} tells us that Picard's set of moves is a subset of $D_{8} ( \ket{+} )$ and Proposition \ref{prp:The Fixed Set of M_P in D_8} completes the picture by revealing that ket $\ket{+}$ is among those that are fixed by Picard's moves. Thus, he is completely powerless to change the state $\ket{+}$ of the coin. Under this perspective the progression of the $PQG$ can be abstractly described as by the following ``algorithm.''

\begin{figure}[H]
	\centering
	\begin{algorithm}[H] 
		\caption{Q's Winning strategy in the original $PQG$}
		\label{alg:Q's PQG Winning Strategy}
		Q's first move sends the coin to an intermediate target state (in the actual game it happens to be $\ket{+}$) that satisfies the following property: \emph{this state is fixed by Picard's moves or, equivalently, all of Picard's moves belong to the stabilizer of this state} (in the actual game $I, F \in D_{8} ( \ket{+} )$). \\
		Picard acts on the coin, but no matter which move he makes, the quantum coin remains in the \emph{same} state. \\
		Q's final move sends the coin to the desired state. \\
	\end{algorithm}
	\caption{This simple algorithm captures the essence of Q's strategy in the $PQG$.} 
\end{figure}

Picard symbolizes the classical player and as such it is quite appropriate to assume that his repertoire is the set $M_P = \{ I, F \}$. This set is also a group, in particular the $\mathbb{Z}_2$ group of two elements\footnote{In the literature $\mathbb{Z}_2$ is more often denoted as $\{0, 1\}$ under addition modulo $2$.}. In the rest of this paper we shall always assume that the classical player can only make use of these two actions. In the coming sections we shall employ Algorithm \ref{alg:Q's PQG Winning Strategy} to discover winning strategies for Q in more general situations.

\section{Enlarging the operational space of the game} \label{sec:Enlarging the Operational Space of the PQG}

As we begin this section let us recall that the operational space of the original $PQG$ is indeed a group, and, in particular, the dihedral group $D_8$, as established by Theorem \ref{thr:PQG Ambient Group}. In this section we shall progressively enlarge the ambient group of the $PQG$ and analyze Q's winning strategies. Our analysis is guided by the belief that the essence of the original $PQG$ is the sharp distinction between the classical and the quantum player. From this perspective, our subsequent investigation relies on the following two assumptions.

\begin{enumerate}
	\item	Picard, who embodies the classical player, can flip the coin. If he is deprived of this ability, then the resulting game becomes trivial and meaningless. He should not be able to do more than that, as this would endow him with quantum capabilities. Formally, we express this by specifying:
			\begin{align} \label{eq:The Moves of Picard}
				M_P = \{ I, F \} \ . \tag{$A_{1}$}
			\end{align}
	\item	Q, who stands for the quantum player, must exhibit quantumness. Thus, at least one of his actions must lie outside the classical realm. In more technical terms, his repertoire $M_Q$ must contain at least one operator from $U(2)$ other than $I$ and $F$.
\end{enumerate}

Under the above assumptions, we may state the following properties that are quite general, as they are satisfied by every winning strategy of Q, no matter what the ambient group is. Therefore, we shall invoke these properties when we are examining much larger dihedral groups and the unitary group $U(2)$.


\begin{theorem}[Characteristic properties of winning strategies] \label{thr:Characteristic Properties of Winning Strategies}
	If $(A_1, A_{2})$ is a winning strategy for Q, then:
	\begin{align}
		A_2 I A_1 \ket{0} &= A_2 F A_1 \ket{0} = \ket{0} \ , \label{eq:Characteristic Property I of Winning Strategies}
		\qquad \text{and}
		\\
		A_1 \ket{0} &\in Fix ( \{ F \} ) \ . \label{eq:Characteristic Property II of Winning Strategies}
	\end{align}
\end{theorem}

We introduce the notion of \emph{equivalent strategies} in order to simplify the classification of winning strategies. We consider two strategies to be equivalent if, when acting on the same initial state of the coin, they produce the same sequence of states. In view of the extension of the original game that will be undertaken in Section, the next definition is general enough to deal with strategies for games with more than three number of rounds.

\begin{definition}[Equivalent strategies] \label{def:Equivalent Strategies}
	Let $\sigma = (A_1, \dots, A_{r})$ and $\sigma' = (A'_1, \dots, A'_{r})$ be two strategies of the same player, and let $\ket{q_0}$ be the initial state of the coin. We say that $\sigma$ and $\sigma'$ are \emph{equivalent} with respect to $\ket{q_0}$, denoted by $\sigma \sim \sigma'$, if
		\begin{align} \label{eq:Equivalent Strategies}
			A_{j} \dots A_{1} \ket{q_0} = A'_{j} \dots A'_{1} \ket{q_0}, \text{ for every } j, \ 1 \leq j \leq r \ .
		\end{align}
\end{definition}

For example, Q's strategies $(H, H)$ and $(R_{\frac{2 \pi}{8}}, R_{\frac{14 \pi}{8}})$ are equivalent because they send the coin from state $\ket{0}$ first to $\ket{+}$ and then back to $\ket{0}$. It is obvious that $\sim$ is an equivalence relation that partitions the set of strategies into equivalence classes of strategies.

\begin{definition}[Strategy classes] \label{def:Strategy Classes} \
	\begin{enumerate}
		\item	Given a strategy $\sigma$, we designate by $[ \sigma ]$ the equivalence class that contains $\sigma$. Any member of $[ \sigma ]$ is a \emph{representative} of $[ \sigma ]$.
		\item	To every class $[ \sigma ]$ we associate the \emph{state path} $\tau_{ [ \sigma ] }$ as follows: if $(A_1, \dots, A_{r})$ is any representative of $[ \sigma ]$, we define 
		$\tau_{ [ \sigma ] }$ to be $( \ket{q_{0}}, \ket{q_{1}}, \dots, \ket{q_{r}})$, where
				\begin{align} \label{eq:State Paths}
					\ket{q_{j}} = A_{j} \dots A_{1} \ket{q_0}, \text{ for every } j, 1 \leq j \leq r \ .
				\end{align}
	\end{enumerate}
\end{definition}

Clearly, the state path $\tau_{ [ \sigma ] }$ is well-defined and unique for each class $[ \sigma ]$. The equivalence class $[ (H, H) ]$ contains $16$ strategies, as will be explained in Example \ref{xmp:Q's PQG Winning Strategies in D_8}, and the corresponding state path is $(\ket{0}, \ket{+}, \ket{0})$.

\subsection{Inside $D_8$}

Before delving into other groups, we examine the case where Q can chose his moves from the entire $D_8$ group, i.e.,

\begin{align} \label{eq:Q's Repertoire is D_8}
	M_Q = D_8 \ . \tag{$A_{2}$}
\end{align}

The following Example \ref{xmp:Q's PQG Winning Strategies in D_8} will be instructive.

\begin{example} \label{xmp:Q's PQG Winning Strategies in D_8}
	In this example, we shall apply Algorithm \ref{alg:Q's PQG Winning Strategy} to study \emph{all} winning strategies of Q in the original $PQG$. Let $(A_{1}, A_{2})$ be Q's first and second move in a winning strategy. After Q's first move the quantum coin will in one of the states in the orbit $D_{8} \star B$, where $B$ is the computational basis. From (\ref{eq:Orbit of B in D_8}) we know that $D_{8} \star B = \{ \ket{0}, \ket{+}, \ket{1}, \ket{-} \}$.
	\begin{itemize}
		\item	Let us first establish that if Q leaves the coin at state $\ket{0}$, or sends it to state $\ket{1}$, then he will not be able to win with probability $1.0$. To see this more clearly, let us recall that, by Theorem \ref{thr:Characteristic Properties of Winning Strategies}, $A_2 I A_1 \ket{0} = A_2 F A_1 \ket{0} = \ket{0}$. If $(A_{1}, A_{2})$ leaves the coin at state $\ket{0}$, i.e., $A_1 \ket{0} = \ket{0}$, then $A_2 I \ket{0} = A_2 F \ket{0} = \ket{0} \Rightarrow A_2 \ket{0} = A_2 \ket{1} = \ket{0}$, which is impossible because $A_2$ represents an element of $D_8$. The same reasoning shows that if Q's first move sends the coin to state $\ket{1}$, then he will not be able to win with probability $1.0$.
		\item	In the original $PQG$, Q won by sending the coin to state $\ket{+}$. In $D_{8}$, this can be achieved with $4$ different ways: $H, R_{\frac{2 \pi}{8}}, S_{\frac{5 \pi}{8}}$ and $R_{\frac{10 \pi}{8}}$. State $\ket{+}$ is fixed by $I$ and $F$ according to (\ref{eq:The Fixed Set of M_P in D_8}), which means that no matter what Picard plays, the coin will remain in this state. Finally, Q can send the coin back to the $\ket{0}$ state with $4$ different ways: $H, R_{\frac{14 \pi}{8}}, S_{\frac{5 \pi}{8}}$ and $R_{\frac{6 \pi}{8}}$. This means that Q has $16$ different winning strategies, which, in view of Definition \ref{def:Equivalent Strategies}, are equivalent. Thus, they constitute one equivalence class of winning strategies. Strategy $(H, H)$ is a representative of this class, but any other strategy would also do. For this class the corresponding state path is $(\ket{0}, \ket{+}, \ket{0})$.
		\item	Algorithm \ref{alg:Q's PQG Winning Strategy} enables us to discover one more winning strategy for Q. Q has another option, which is to drive the coin to state $\ket{-}$. This can also be achieved with $4$ different ways: $S_{\frac{7 \pi}{8}}, R_{\frac{14 \pi}{8}}, S_{\frac{3 \pi}{8}}$ or $R_{\frac{6 \pi}{8}}$. Picard cannot change this state either  because $\ket{-}$ is fixed by $I$ and $F$, according to (\ref{eq:The Fixed Set of M_P in D_8}). During the final round Q has the opportunity to send the coin back to the $\ket{0}$ state with $4$ different ways: $S_{\frac{7 \pi}{8}}, R_{\frac{2 \pi}{8}}, S_{\frac{3 \pi}{8}}$ or $R_{\frac{10 \pi}{8}}$. Hence, Q has $16$ more winning strategies, which are equivalent. They make the second equivalence class of winning strategies and any one of them, e.g., $(S_{\frac{7 \pi}{8}}, S_{\frac{7 \pi}{8}})$ can be its representative. For this class, the corresponding state path is $(\ket{0}, \ket{-}, \ket{0})$.
		\item	Picard, unfortunately for him, has no winning strategy.
	\end{itemize}
	According to Definition \ref{def:Winning and Dominant Strategies}, for Q a winning strategy is also a dominant strategy. Hence, Q has precisely two classes of winning and dominant strategies, each containing $16$ individual strategies. These two classes correspond to exactly the $2$ path states $(\ket{0}, \ket{+}, \ket{0})$ and $(\ket{0}, \ket{-}, \ket{0})$.
	\hfill $\triangleleft$
\end{example}

Table \ref{tbl:Q's Equivalence Classes of Winning Strtegies for the PQG in D_8} and Figure \ref{fig:Winning Strategies in the original $PQG$} summarize these results.

\begin{table}[H]
	\centering
	\caption{The two classes of winning and dominant strategies for Q in the original $PQG$.}
	\label{tbl:Q's Equivalence Classes of Winning Strtegies for the PQG in D_8}
	\renewcommand{\arraystretch}{2.0}
	\begin{tabular}{ l !{\vrule width 1.25 pt} c | c | c | c }
		\Xhline{4\arrayrulewidth}
		\multicolumn{5}{c}{The evolution of the $PQG$}
		\\
		\Xhline{\arrayrulewidth}
		&
		Initial state
		&
		Round $1$
		&
		Round $2$
		&
		Round $3$
		\\
		\Xhline{3\arrayrulewidth}
		$(H, H), (R_{\frac{2 \pi}{8}}, R_{\frac{14 \pi}{8}}), (S_{\frac{5 \pi}{8}}, S_{\frac{5 \pi}{8}}), \dots$
		&
		$\ket{0}$
		&
		\cellcolor{GreenLighter2!20} $\ket{+}$
		&
		\cellcolor{GreenLighter2!20} $\ket{+}$
		&
		$\ket{0}$
		\\
		\Xhline{\arrayrulewidth}
		$(S_{\frac{7 \pi}{8}}, S_{\frac{7 \pi}{8}}), (R_{\frac{14 \pi}{8}}, R_{\frac{2 \pi}{8}}), (S_{\frac{3 \pi}{8}}, S_{\frac{3 \pi}{8}}), \dots$
		&
		$\ket{0}$
		&
		\cellcolor{RedPurple!15} $\ket{-}$
		&
		\cellcolor{RedPurple!15} $\ket{-}$
		&
		$\ket{0}$
		\\
		\Xhline{4\arrayrulewidth}
	\end{tabular}
	\renewcommand{\arraystretch}{1.0}
\end{table}

\vspace{0.5 cm}

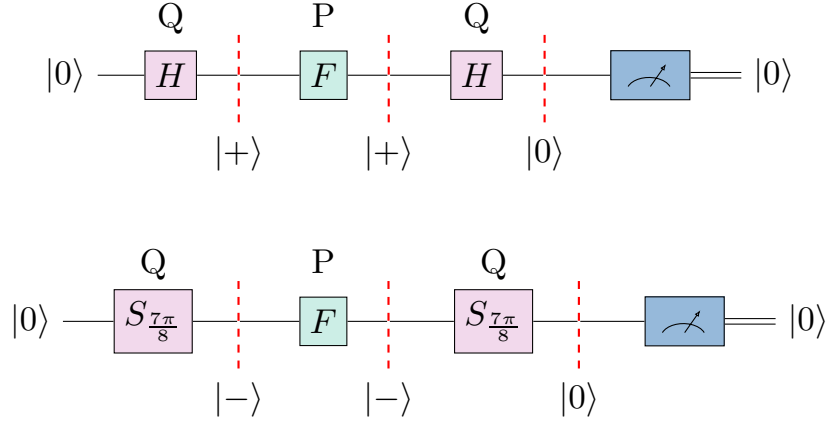
\begin{figure}[H]
	\centering
	\begin{tikzpicture} [scale = 1.3]
		\begin{yquant}[operators/every barrier/.append style={red, thick, shorten <= -3mm, shorten >= -3mm}]
			qubit {$\ket{0}$} GAME;
			hspace {0.3 cm} GAME;
			[ name = Q1 ]
			[ fill = RedPurple!15 ]
			box {$H$} GAME;
			[ name = S1 ]
			barrier GAME;
			hspace {0.2 cm} GAME;
			[ name = P1 ]
			[ fill = GreenLighter2!20 ]
			box {$F$} GAME;
			[ name = S2 ]
			barrier GAME;
			hspace {0.2 cm} GAME;
			[ name = Q2 ]
			[ fill = RedPurple!15 ]
			box {$H$} GAME;
			[ name = S3 ]
			barrier GAME;
			hspace {0.2 cm} GAME;
			[ fill = WordBlueDark!50 ]
			measure GAME;
			hspace {0.3 cm} GAME;
			output {$\ket{0}$} GAME;
			\node [ above = 0.3 cm ] at (Q1) {Q};
			\node [ below = 0.5 cm ] at (S1) {$\ket{+}$};
			\node [ above = 0.37 cm ] at (P1) {P};
			\node [ below = 0.5 cm ] at (S2) {$\ket{+}$};
			\node [ above = 0.3 cm ] at (Q2) {Q};
			\node [ below = 0.5 cm ] at (S3) {$\ket{0}$};
		\end{yquant}
	\end{tikzpicture}
	\\
	\vspace{0.7 cm}
	\begin{tikzpicture} [scale = 1.3]
		\begin{yquant}[operators/every barrier/.append style={red, thick, shorten <= -3mm, shorten >= -3mm}]
			qubit {$\ket{0}$} GAME;
			hspace {0.3 cm} GAME;
			[ name = Q1 ]
			[ fill = RedPurple!15 ]
			box {$S_{\frac{7 \pi}{8}}$} GAME;
			[ name = S1 ]
			barrier GAME;
			hspace {0.2 cm} GAME;
			[ name = P1 ]
			[ fill = GreenLighter2!20 ]
			box {$F$} GAME;
			[ name = S2 ]
			barrier GAME;
			hspace {0.2 cm} GAME;
			[ name = Q2 ]
			[ fill = RedPurple!15 ]
			box {$S_{\frac{7 \pi}{8}}$} GAME;
			[ name = S3 ]
			barrier GAME;
			hspace {0.2 cm} GAME;
			[ fill = WordBlueDark!50 ]
			measure GAME;
			hspace {0.3 cm} GAME;
			output {$\ket{0}$} GAME;
			\node [ above = 0.3 cm ] at (Q1) {Q};
			\node [ below = 0.5 cm ] at (S1) {$\ket{-}$};
			\node [ above = 0.37 cm ] at (P1) {P};
			\node [ below = 0.5 cm ] at (S2) {$\ket{-}$};
			\node [ above = 0.3 cm ] at (Q2) {Q};
			\node [ below = 0.5 cm ] at (S3) {$\ket{0}$};
		\end{yquant}
	\end{tikzpicture}

	\caption{This figure depicts two different winning strategies for Q that represent the two winning strategy classes, as well as the corresponding path states.}
	\label{fig:Winning Strategies in the original $PQG$}
\end{figure}

\begin{theorem}[The ambient group of the $PQG$ is $D_8$] \label{thr:Q's PQG Winning Strategies in D_8}
	If we assume that $M_P = \{ I, F \}$ and $M_Q = D_8$, i.e., the ambient group of the $PQG$ is $D_8$, then the following hold.
	\begin{enumerate}
		\item	Q has exactly two classes of winning and dominant strategies, each containing $16$ equivalent strategies:
				\begin{align} \label{eq:D_8 Winning Strategy Classes}
					\mathcal{C}_{+} = [ ( A_{1}, A_{2} ) ] \quad \text{and} \quad \mathcal{C}_{-} = [ ( B_{1}, B_{2} ) ] \ ,
				\end{align}
				where
				\begin{itemize}
					\item	$A_{1}$ is one of $H, R_{\frac{2 \pi}{8}}, S_{\frac{5 \pi}{8}}$ or $R_{\frac{10 \pi}{8}}$,
					\item	$A_{2}$ is one of $H, R_{\frac{14 \pi}{8}}, S_{\frac{5 \pi}{8}}$ or $R_{\frac{6 \pi}{8}}$,
					\item	$B_{1}$ is one of $S_{\frac{7 \pi}{8}}, R_{\frac{14 \pi}{8}}, S_{\frac{3 \pi}{8}}$ or $R_{\frac{6 \pi}{8}}$, and
					\item	$B_{2}$ is one of $S_{\frac{7 \pi}{8}}, R_{\frac{2 \pi}{8}}, S_{\frac{3 \pi}{8}}$ or $R_{\frac{10 \pi}{8}}$.
				\end{itemize}
		\item	The winning state paths corresponding to $\mathcal{C}_{+}$ and $\mathcal{C}_{-}$ are
				\begin{align} \label{eq:D_8 Winning State Paths}
					\tau_{\mathcal{C}_{+}} = (\ket{0}, \ket{+}, \ket{0}) \quad \text{and} \quad \tau_{\mathcal{C}_{-}} = (\ket{0}, \ket{-}, \ket{0}) \ .
				\end{align}
		\item	Picard has no winning strategy.
	\end{enumerate}
\end{theorem}

\subsection{The smaller dihedral groups $D_3, D_4, D_5, D_6$ and $D_7$}

We may ask whether any of the smaller dihedral groups $D_3, D_4, D_5, D_6$ and $D_7$ can be an appropriate operational space for the $PQG$. The answer is no for the reasons outlined below.

\begin{itemize}
	\item	$D_3, D_5, D_6$ and $D_7$ do not contain the reflection $F$. This can be verified by comparing formula (\ref{eq:The F Reflector}) with formula (\ref{eq:Standard Representation r^k s in D_n}) for $n = 3, 5, 6$ and $7$ and $k = 1, \dots, n - 1$. We have assumed that Picard, the classical player, must be able to flip the coin, as emphasized in (\ref{eq:The Moves of Picard}).
	\item	$D_4$ does contain the reflection $F$. However, the orbit $D_{4} \star B$ is $\{ \ket{0}, \ket{1} \}$. This means that Q can only flip the coin from heads to tails or vice versa. If $M_Q = D_4$, then the $PQG$ degenerates to the classical coin tossing game. Q is no longer a quantum entity and, as explained in Example \ref{xmp:Q's PQG Winning Strategies in D_8}, no longer possesses a winning strategy. From this perspective, it becomes meaningless to play the $PQG$ in $D_4$.
\end{itemize}


These conclusions are contained in Table \ref{tbl:D_8 is the smallest Ambient Group in which Q has a winning strategy} for easy reference.

\begin{table}[H]
	\centering
	\caption{In smaller dihedral groups, it is either impossible to play the $PQG$, or, in the event that it is possible (such as in $D_4$), Q lacks a winning strategy.}
	\label{tbl:D_8 is the smallest Ambient Group in which Q has a winning strategy}
	\renewcommand{\arraystretch}{1.5}
	\begin{tabular}{c !{\vrule width 1.25 pt} c|c|c}
		\Xhline{4 \arrayrulewidth}
		Ambient group & Is $PQG$ playable & Winning strategy for Picard & Winning strategy for Q
		\\
		\Xhline{3 \arrayrulewidth}
		$D_3$ & No ($F \not \in M_P$) & --- & ---
		\\
		\hline
		$D_4$ & Yes (classical coin tossing) & No & No
		\\
		\hline
		$D_5$ & No ($F \not \in M_P$) & --- & ---
		\\
		\hline
		$D_6$ & No ($F \not \in M_P$) & --- & ---
		\\
		\hline
		$D_7$ & No ($F \not \in M_P$) & --- & ---
		\\
		\Xhline{4 \arrayrulewidth}
	\end{tabular}
	\renewcommand{\arraystretch}{1.0}
\end{table}

Therefore, if we accept that the classical player should, at the very least, be able to flip the coin in order to have a nontrivial game, and that the quantum player must exhibit quantumness, then the smallest dihedral group for the $PQG$ is $D_8$. This fact is stated as Theorem \ref{thr:The smallest dihedral group of the PQG is D_8}.

\begin{theorem}[The smallest dihedral group for the $PQG$ is $D_8$] \label{thr:The smallest dihedral group of the PQG is D_8}
	 $D_8$ is the smallest of the dihedral groups such that $PQG$ can be meaningful played and in which Q has a quantum winning strategy.
\end{theorem}

\subsection{The dihedral groups $D_{8 n}$, $n \geq 1$}

The previous subsection demonstrated that the smallest meaningful group for the $PQG$ is $D_{8}$. This subsection examines what happens if we allow Q to choose from a larger repertoire, and, more specifically, if we assume that

\begin{align} \label{eq:The Moves of Q in D_n}
	M_Q = D_n \ , \ n \geq 8 \ . \tag{$A_{3}$}
\end{align}

Let us as first make the helpful observation that when $n$ is odd, then $D_{n}$ does not contain $F$, which is stated as Proposition \ref{prp:Absence of F in D_n for n Odd}.

\begin{proposition}[$D_{n}$ does not contain $F$ when $n$ odd] \label{prp:Absence of F in D_n for n Odd}
	If $n$ is odd, then the dihedral group $D_{n}$ does not contain $F$.
\end{proposition}

This result enables us to exclude these groups from now on when considering larger groups where the $PQG$ can be successfully played.

Another useful result about the orbits of $B$ in general dihedral groups is contained in Theorem \ref{thr:The Action of D_n on B}.

\begin{theorem}[The action of $D_{n}$ on $B$] \label{thr:The Action of D_n on B}
	The action of the general dihedral group $D_{n}, n \geq 3,$ on the computational basis $B$ depends on whether $n$ is a multiple of $4$ or $n$ is even but not a multiple of $4$. Specifically,
	\begin{enumerate}
		\item	if $n$ is a multiple of $4$, then the action of the dihedral group $D_{n}$ on the computational basis $B$ is
				\begin{align} \label{eq:Orbit of B in D_n for n 4-Multiple}
					D_{n} \star \ket{0} = D_{n} \star \ket{1} = D_{n} \star B =
					\{ \cos \frac{2 \pi k}{n} \ket{0} + \sin \frac{2 \pi k}{n} \ket{1} : 0 \leq k < \frac{n}{2} \} \ ,
				\end{align}
		\item	if $n$ is even but not a multiple of $4$, then the action of the dihedral group $D_{n}$ on the computational basis $B$ is
				\begin{align}
					D_{n} \star \ket{0}
					&=
					\{ \cos \frac{2 \pi k}{n} \ket{0} + \sin \frac{2 \pi k}{n} \ket{1} : 0 \leq k < \frac{n}{2} \} \ , \label{eq:Orbit of Ket 0 in D_n for n Even but Not 4-Multiple}
					\\
					D_{n} \star \ket{1}
					&=
					\{ -\sin \frac{2 \pi k}{n} \ket{0} + \cos \frac{2 \pi k}{n} \ket{1} : 0 \leq k < \frac{n}{2} \} \ , \label{eq:Orbit of Ket 1 in D_n for n Even but Not 4-Multiple}
					\\
					D_{n} \star B
					&= \{ \cos \frac{2 \pi k}{n} \ket{0} + \sin \frac{2 \pi k}{n} \ket{1} : 0 \leq k < \frac{n}{2} \}
					\cup \{ -\sin \frac{2 \pi k}{n} \ket{0} + \cos \frac{2 \pi k}{n} \ket{1} : 0 \leq k < \frac{n}{2} \} \ . \label{eq:Orbit of B in D_n for n Even but Not 4-Multiple}
				\end{align}
	\end{enumerate}
\end{theorem}

Figures \ref{fig:Orbit of B in D_n for n 4-Multiple}, \ref{fig:Orbit of B in D_n for n Even}, and \ref{fig:Orbit of B in D_n for n Odd} provide intuitive visualizations of Theorem \ref{thr:The Action of D_n on B} and Proposition \ref{prp:Absence of F in D_n for n Odd}.

\begin{figure}[H]
	\begin{minipage}[t]{0.32\textwidth}
		\centering
		\begin{tikzpicture}[scale = 1.7]
			\def \angle {360/8}
			\draw (-1.5, 0) -- (1.5, 0);
			\draw (0, -1.5) -- (0, 1.5);
			\draw [fill, thick, WordBlueDark]
			({cos(0 * \angle)}, {sin(0 * \angle)}) circle (0.75 pt) node [below right] {1};
			\draw [fill, thick, WordBlueDark]
			({cos(1 * \angle)}, {sin(1 * \angle)}) circle (0.75 pt) node [above right] {A};
			\draw [fill, thick, WordBlueDark]
			({cos(2 * \angle)}, {sin(2 * \angle)}) circle (0.75 pt) node [above left] {1};
			\draw [fill, thick, WordBlueDark]
			({cos(3 * \angle)}, {sin(3 * \angle)}) circle (0.75 pt) node [above left] {B};
			\draw [thick, WordBlueDark]
			({cos(4 * \angle)}, {sin(4 * \angle)}) circle (0.75 pt) node [below left] {-1};
			\draw [thick, WordBlueDark]
			({cos(5 * \angle)}, {sin(5 * \angle)}) circle (0.75 pt) node [below left] {$A^\prime$};
			\draw [thick, WordBlueDark]
			({cos(6 * \angle)}, {sin(6 * \angle)}) circle (0.75 pt) node [below left] {-1};
			\draw [thick, WordBlueDark]
			({cos(7 * \angle)}, {sin(7 * \angle)}) circle (0.75 pt) node [below right] {$B^\prime$};
			\draw [thin, dashed, WordBlueDark] ({cos(1 * \angle)}, {sin(1 * \angle)}) -- ({cos(5 * \angle)}, {sin(5 * \angle)});
			\draw [thin, dashed, WordBlueDark] ({cos(3 * \angle)}, {sin(3 * \angle)}) -- ({cos(7 * \angle)}, {sin(7 * \angle)});
			\draw [thick, WordBlueDark] (1cm, 0cm) arc [start angle = 0, end angle = 180, radius = 1cm];
			\draw [thin, dashed, WordBlueDark] (1cm, 0cm) arc [start angle = 0, end angle = -180, radius = 1cm];
			\scoped [on background layer]
			\filldraw [->, MyVeryLightBlue, line width = 0.3 mm] (0, 0) -- (0.5,0) arc (0:\angle:0.5);
			\draw [->, WordBlueDark, line width = 0.3 mm] (0.5,0) arc (0:\angle:0.5);
			\draw [->, WordBlueDark, line width = 0.3 mm] ({1.2 * cos(0.5 * \angle)}, {1.2 * sin(0.5 * \angle)}) node [WordBlueDark, right] {\Large $\frac{2 \pi}{n}$} -- ({0.3 * cos(0.5 * \angle)}, {0.3 * sin(0.5 * \angle)});
		\end{tikzpicture}
		\caption{Kets $\ket{0}$ and $\ket{1}$ have the same orbit in case $n$ is a 4-multiple. The antipodal points that arise represent the same state.}
		\label{fig:Orbit of B in D_n for n 4-Multiple}
	\end{minipage}
	\hfill
	\begin{minipage}[t]{0.32\textwidth}
		\centering
		\begin{tikzpicture}[scale = 1.7]
			\def \angle {360/6}
			\draw (-1.5, 0) -- (1.5, 0);
			\draw (0, -1.5) -- (0, 1.5);
			\draw [thin, WordBlueDark] (0, 0) circle [ radius = 1 cm ];
			\draw [fill, thick, RedPurple]
			({cos(0 * \angle)}, {sin(0 * \angle)}) circle (0.75 pt) node [below right] {1};
			\draw [fill, thick, RedPurple]
			({cos(1 * \angle)}, {sin(1 * \angle)}) circle (0.75 pt) node [above right] {A};
			\draw [fill, thick, RedPurple]
			({cos(2 * \angle)}, {sin(2 * \angle)}) circle (0.75 pt) node [above left] {B};
			\draw [thick, RedPurple]
			({cos(3 * \angle)}, {sin(3 * \angle)}) circle (0.75 pt) node [below left] {-1};
			\draw [thick, RedPurple]
			({cos(4 * \angle)}, {sin(4 * \angle)}) circle (0.75 pt) node [below left] {$A^\prime$};
			\draw [thick, RedPurple]
			({cos(5 * \angle)}, {sin(5 * \angle)}) circle (0.75 pt) node [below right] {$B^\prime$};
			\draw [thin, dashed, RedPurple] ({cos(1 * \angle)}, {sin(1 * \angle)}) -- ({cos(4 * \angle)}, {sin(4 * \angle)});
			\draw [thin, dashed, RedPurple] ({cos(2 * \angle)}, {sin(2 * \angle)}) -- ({cos(5 * \angle)}, {sin(5 * \angle)});
			\draw [fill, thick, GreenLighter2]
			({cos(90 + 0 * \angle)}, {sin(90 + 0 * \angle)}) circle (0.75 pt) node [above left] {1};
			\draw [fill, thick, GreenLighter2]
			({cos(90 - 1 * \angle)}, {sin(90 - 1 * \angle)}) circle (0.75 pt) node [above right] {C};
			\draw [thick, GreenLighter2]
			({cos(90 + 2 * \angle)}, {sin(90 + 2 * \angle)}) circle (0.75 pt) node [below left] {$C^\prime$};
			\draw [thick, GreenLighter2]
			({cos(90 + 3 * \angle)}, {sin(90 + 3 * \angle)}) circle (0.75 pt) node [below left] {-1};
			\draw [fill, thick, GreenLighter2]
			({cos(90 + 1 * \angle)}, {sin(90 + 1 * \angle)}) circle (0.75 pt) node [above left] {D};
			\draw [thick, GreenLighter2]
			({cos(90 - 2 * \angle)}, {sin(90 - 2 * \angle)}) circle (0.75 pt) node [below right] {$D^\prime$};
			\draw [thin, dashed, GreenLighter2] ({cos(90 - 1 * \angle)}, {sin(90 - 1 * \angle)}) -- ({cos(90 + 2 * \angle)}, {sin(90 + 2 * \angle)});
			\draw [thin, dashed, GreenLighter2] ({cos(90 + 1 * \angle)}, {sin(90 + 1 * \angle)}) -- ({cos(90 - 2 * \angle)}, {sin(90 - 2 * \angle)});
			\scoped [on background layer]
			\filldraw [->, MyLightRed, line width = 0.3 mm] (0, 0) -- (0.5, 0) arc (0:\angle:0.5);
			\draw [->, RedPurple, line width = 0.3 mm] (0.5, 0) arc (0:\angle:0.5);
			\draw [->, RedPurple, line width = 0.3 mm] ({1.2 * cos(0.25 * \angle)}, {1.2 * sin(0.25 * \angle)}) node [RedPurple, right] {\Large $\frac{2 \pi}{n}$} -- ({0.3 * cos(0.5 * \angle)}, {0.3 * sin(0.5 * \angle)});
			\scoped [on background layer]
			\filldraw [->, WordLightGreen, line width = 0.3 mm] (0, 0) -- (0, 0.5) arc (90:90+\angle:0.5);
			\draw [->, GreenLighter2, line width = 0.3 mm] (0, 0.5) arc (90:90+\angle:0.5);
			\draw [->, GreenLighter2, line width = 0.3 mm] ({1.2 * cos(1.75 * \angle)}, {1.2 * sin(1.75 * \angle)}) node [GreenLighter2, above] {\Large $\frac{2 \pi}{n}$} -- ({0.3 * cos(2.0 * \angle)}, {0.3 * sin(2.0 * \angle)});
		\end{tikzpicture}
		\caption{Kets $\ket{0}$ and $\ket{1}$ have different orbits in case in case $n$ is even, but not a 4-multiple.}
		\label{fig:Orbit of B in D_n for n Even}
	\end{minipage}
	\hfill
	\begin{minipage}[t]{0.32\textwidth}
		\centering
		\begin{tikzpicture}[scale = 1.7]
			\def \angle {360/5}
			\draw (-1.5, 0) -- (1.5, 0);
			\draw (0, -1.5) -- (0, 1.5);
			\draw [thin, WordBlueDark] (0, 0) circle [ radius = 1 cm ];
			\draw [fill, thick, RedPurple]
			({cos(0 * \angle)}, {sin(0 * \angle)}) circle (0.75 pt) node [below right] {1};
			\draw [fill, thick, RedPurple]
			({cos(1 * \angle)}, {sin(1 * \angle)}) circle (0.75 pt) node [above right] {A};
			\draw [fill, thick, RedPurple]
			({cos(2 * \angle)}, {sin(2 * \angle)}) circle (0.75 pt) node [above left] {B};
			\draw [fill, thick, RedPurple]
			({cos(3 * \angle)}, {sin(3 * \angle)}) circle (0.75 pt) node [below left] {C};
			\draw [fill, thick, RedPurple]
			({cos(4 * \angle)}, {sin(4 * \angle)}) circle (0.75 pt) node [below right] {D};
			\draw [fill, thick, GreenLighter2]
			({cos(90 + 0 * \angle)}, {sin(90 + 0 * \angle)}) circle (0.75 pt) node [above left] {1};
			\draw [fill, thick, GreenLighter2]
			({cos(90 + 1 * \angle)}, {sin(90 + 1 * \angle)}) circle (0.75 pt) node [above left] {E};
			\draw [fill, thick, GreenLighter2]
			({cos(90 + 2 * \angle)}, {sin(90 + 2 * \angle)}) circle (0.75 pt) node [below left] {F};
			\draw [fill, thick, GreenLighter2]
			({cos(90 + 3 * \angle)}, {sin(90 + 3 * \angle)}) circle (0.75 pt) node [below right] {G};
			\draw [fill, thick, GreenLighter2]
			({cos(90 + 4 * \angle)}, {sin(90 + 4 * \angle)}) circle (0.75 pt) node [above right] {H};
			\draw [thin, dashed, RedPurple] (0, 0) -- ({cos(1 * \angle)}, {sin(1 * \angle)});
			\scoped [on background layer]
			\filldraw [->, MyLightRed, line width = 0.3 mm] (0, 0) -- (0.5, 0) arc (0:\angle:0.5);
			\draw [->, RedPurple, line width = 0.3 mm] (0.5, 0) arc (0:\angle:0.5);
			\draw [->, RedPurple, line width = 0.3 mm] ({1.2 * cos(0.5 * \angle)}, {1.2 * sin(0.5 * \angle)}) node [RedPurple, right] {\Large $\frac{2 \pi}{n}$} -- ({0.3 * cos(0.5 * \angle)}, {0.3 * sin(0.5 * \angle)});
			\draw [thin, dashed, GreenLighter2] (0, 0) -- ({cos(90 + 1 * \angle)}, {sin(90 + 1 * \angle)});
			\scoped [on background layer]
			\filldraw [->, WordLightGreen, line width = 0.3 mm] (0, 0) -- (0, 0.5) arc (90:90+\angle:0.5);
			\draw [->, GreenLighter2, line width = 0.3 mm] (0, 0.5) arc (90:90+\angle:0.5);
			\draw [->, GreenLighter2, line width = 0.3 mm] ({1.2 * cos(1.75 * \angle)}, {1.2 * sin(1.75 * \angle)}) node [GreenLighter2, above] {\Large $\frac{2 \pi}{n}$} -- ({0.3 * cos(1.75 * \angle)}, {0.3 * sin(1.75 * \angle)});
		\end{tikzpicture}
		\caption{Kets $\ket{0}$ and $\ket{1}$ have different orbits in case in case $n$ is odd. No antipodal points arise in this case.}
		\label{fig:Orbit of B in D_n for n Odd}
	\end{minipage}
\end{figure}

In this more complex setting we may resort to Algorithm \ref{alg:Q's PQG Winning Strategy} to establish under what conditions Q still possesses winning strategies and, if so, which are these. This is facilitated by the next Theorem \ref{thr:The Fixed Set of M_P in D_n}, which explains what happens to the fixed set of $\{ I, F \}$ in $D_n$.

\begin{theorem}[The fixed set of $\{ I, F \}$ in $D_n$] \label{thr:The Fixed Set of M_P in D_n}
	When the general dihedral group $D_{n}, n \geq 3,$ acts on the computational basis $B$, the fixed set of $M_P = \{ I, F \}$ depends on whether $n$ is a multiple of $8$ or not.
	\begin{enumerate}
		\item	If $n$ is a multiple of $8$, then:
				\begin{align} \label{eq:The Fixed Set of F in D_n for n 8-Multiple}
					Fix ( \{ I, F \} ) = Fix ( F ) = \{ \ket{+}, \ket{-} \} \ .
				\end{align}
		\item	In every other case:
				\begin{align} \label{eq:The Fixed Set of F in D_n for n Not 8-Multiple}
					Fix ( \{ I, F \} ) = Fix ( F ) = \emptyset \ .
				\end{align}
	\end{enumerate}
\end{theorem}

The significance of Theorem \ref{thr:The Fixed Set of M_P in D_n} is twofold. First, from a somewhat negative perspective, disqualifies most of the dihedral groups as potential ambient groups for the $PQG$. Simultaneously, in a positive note, ascertains that the $PQG$ can be meaningfully played in every dihedral group $D_n$ such that $n$ is a multiple of 8. The next Theorem \ref{thr:Q's PQG Winning Strategies in D_{8 n}} explains what exactly happens in terms of winning strategies when the $PQG$ is played in the aforementioned groups.

\begin{theorem}[The ambient group of the $PQG$ is $D_{8 n}$] \label{thr:Q's PQG Winning Strategies in D_{8 n}}
	If $M_P = \{ I, F \}$ and $M_Q = D_{8 n}$, i.e., the ambient group of the $PQG$ is $D_{8 n}$, where $n \geq 1$, then the following hold.
	\begin{enumerate}
		\item	Q has exactly two classes of winning and dominant strategies, each containing $16$ equivalent strategies:
				\begin{align} \label{eq:D_{8 n} Winning Strategy Classes}
					\mathcal{C}_{+} = [ ( A_{1}, A_{2} ) ] \quad \text{and} \quad \mathcal{C}_{-} = [ ( B_{1}, B_{2} ) ] \ ,
				\end{align}
				where
				\begin{itemize}
					\item	$A_{1}$ is one of $H, R_{\frac{2 \pi}{8}}, S_{\frac{5 \pi}{8}}$ or $R_{\frac{10 \pi}{8}}$,
					\item	$A_{2}$ is one of $H, R_{\frac{14 \pi}{8}}, S_{\frac{5 \pi}{8}}$ or $R_{\frac{6 \pi}{8}}$,
					\item	$B_{1}$ is one of $S_{\frac{7 \pi}{8}}, R_{\frac{14 \pi}{8}}, S_{\frac{3 \pi}{8}}$ or $R_{\frac{6 \pi}{8}}$, and
					\item	$B_{2}$ is one of $S_{\frac{7 \pi}{8}}, R_{\frac{2 \pi}{8}}, S_{\frac{3 \pi}{8}}$ or $R_{\frac{10 \pi}{8}}$.
				\end{itemize}
		\item	The winning state paths corresponding to $\mathcal{C}_{+}$ and $\mathcal{C}_{-}$ are
				\begin{align} \label{eq:D_{8 n} Winning State Paths}
					\tau_{\mathcal{C}_{+}} = (\ket{0}, \ket{+}, \ket{0}) \quad \text{and} \quad \tau_{\mathcal{C}_{-}} = (\ket{0}, \ket{-}, \ket{0}) \ .
				\end{align}
		\item	Picard has no winning strategy.
	\end{enumerate}
\end{theorem}

We are led to a very important conclusion: nothing substantial will change if the game takes place in much larger groups than $D_{8}$; the winning strategies remain precisely the same. This realization of course begs the question whether things we will turn to be different when Q has at his disposal the largest group possible, $U(2)$, which is examined in the next subsection.

\subsection{The entire $U(2)$}

In this section we shall examine the situation when Q is free to choose from all of $U(2)$, which is the largest possible group that Q can draw his moves from. Therefore, without further ado we state our final assumption regarding Q's set of actions.

\begin{align} \label{eq:The Moves of Q in U(2)}
	M_Q = U(2) \ . \tag{$A_{4}$}
\end{align}

The major difference now compared to the previous cases is that $U(2)$ contains infinitely many elements, whereas the previous groups were finite. Although, superficially, this might be expected to drastically enhance Q's capabilities, in turns out that in a certain sense everything remains the same. This can be attributed to the following very simple fact. Kets $\ket{\psi}$ and $e^{i \theta} \ket{\psi}$, where $\theta \in \mathbb{R}$, physically represent the same state. In turn, this implies that the action of the operator $A \in U(2)$ on a ket $\ket{\psi}$ is the same as the action of $e^{i \theta} A \in U(2)$ on $\ket{\psi}$ (see  \cite{Hall2013} for details). If we view $e^{i \theta} A$ as denoting a parametric family of operators, it is clear that all these operators can be considered equivalent and any of them, e.g., $A$, can be taken as the representative of the corresponding equivalence class. In order to simplify the notation, we make the following Definition \ref{def:One-Parameter Operator Families}.

\begin{definition} [Families of unitary operators] \label{def:One-Parameter Operator Families} \

	If $A \in U(2)$, then we define the one-parameter family of unitary operators
	\begin{align} \label{eq:One-Parameter Unitary Operator Family}
		A (\theta) = e^{i \theta} A \ , \ \theta \in \mathbb{R} \ .
	\end{align}
	Similarly, if $R_{\varphi}$ and $S_{\varphi}$ are rotators and reflectors, as given by (\ref{eq:2D Rotator}) and (\ref{eq:2D Reflector}), respectively, we define the one-parameter families of operators:
	\begin{multicols}{2}
		\noindent
		\begin{align} \label{eq:One-Parameter Operator Family of 2D Rotator}
			R_{\varphi} (\theta) = e^{i \theta} R_{\varphi}
		\end{align}
		\begin{align} \label{eq:eq:One-Parameter Operator Family of 2D Reflector}
			S_{\varphi} (\theta) = e^{i \theta} S_{\varphi} \ ,
		\end{align}
	\end{multicols}
	where $\theta \in \mathbb{R}$. Analogously, if $H$ is the Hadamard transform and $R_{\frac{2 \pi k}{n}}$ and $S_{\frac{\pi k}{n}}$ are the matrix representations given by (\ref{eq:Standard Representation r^k in D_n}) and (\ref{eq:Standard Representation r^k s in D_n}), we define the following collections of operators:
	\begin{multicols}{3}
		\noindent
		\begin{align} \label{eq:One-Parameter Operator Family of the Hadamard Transform}
			H (\theta) = e^{i \theta} H
		\end{align}
		\begin{align} \label{eq:One-Parameter Operator Family of the Standard Representation r^k in D_n}
			R_{\frac{2 \pi k}{n}} (\theta) = e^{i \theta} R_{\frac{2 \pi k}{n}}
		\end{align}
		\begin{align} \label{eq:One-Parameter Operator Family of the Standard Representation r^k s in D_n}
			S_{\frac{\pi k}{n}} (\theta) = e^{i \theta} S_{\frac{\pi k}{n}} \ .
		\end{align}
	\end{multicols}
\end{definition}

Again it is fruitful to turn to Algorithm \ref{alg:Q's PQG Winning Strategy} to establish under what conditions Q possesses winning strategies and which are these. This approach is guided by the next Theorem \ref{thr:The Fixed Set of M_P in U(2)}, which establishes the fixed set of $\{ I, F \}$ in $U(2)$.

\begin{theorem}[The fixed set of $\{ I, F \}$ in $U(2)$] \label{thr:The Fixed Set of M_P in U(2)}
	Under the action of $U(2)$ on the computational basis $B$, the fixed set of $M_P = \{ I, F \}$ is
	\begin{align} \label{eq:The Fixed Set of F in U(2)}
		Fix ( \{ I, F \} ) = Fix ( F ) = \{ \ket{+}, \ket{-} \} \ .
	\end{align}
\end{theorem}

This result is crucial in discovering and enumerating the winning strategies of Q in $U(2)$. Although, we might have hoped for more variety in discovering winning strategies, the result is not unexpected because the flip operator $F$ cannot fix more that two states. The next Theorem \ref{thr:Q's PQG Winning Strategies in U(2)} explains what exactly happens in terms of winning strategies when the $PQG$ takes place in $U(2)$.

\begin{theorem}[The ambient group of the $PQG$ is $U(2)$] \label{thr:Q's PQG Winning Strategies in U(2)}
	If $M_P = \{ I, F \}$ and $M_Q = U(2)$, i.e., the ambient group of the $PQG$ is $U(2)$, then the following hold.
	\begin{enumerate}
		\item	Q has exactly two classes of winning and dominant strategies, each containing infinite equivalent strategies:
				\begin{align} \label{eq:U(2) Winning Strategy Classes}
					\mathcal{C}_{+} = [ ( A_{1} (\theta_{1}), A_{2} (\theta_{2}) ) ] \quad \text{and} \quad \mathcal{C}_{-} = [ ( B_{1} (\theta_{3}), B_{2} (\theta_{4}) ) ] \ ,
				\end{align}
				where
				\begin{itemize}
					\item	$A_{1} (\theta_{1})$ is one of $H (\theta_{1}), R_{\frac{2 \pi}{8}} (\theta_{1}), S_{\frac{5 \pi}{8}} (\theta_{1})$ or $R_{\frac{10 \pi}{8}} (\theta_{1})$,
					\item	$A_{2} (\theta_{2})$ is one of $H (\theta_{2}), R_{\frac{14 \pi}{8}} (\theta_{2}), S_{\frac{5 \pi}{8}} (\theta_{2})$ or $R_{\frac{6 \pi}{8}} (\theta_{2})$,
					\item	$B_{1} (\theta_{3})$ is one of $S_{\frac{7 \pi}{8}} (\theta_{3}), R_{\frac{14 \pi}{8}} (\theta_{3}), S_{\frac{3 \pi}{8}} (\theta_{3})$ or $R_{\frac{6 \pi}{8}} (\theta_{3})$,
					\item	$B_{2} (\theta_{4})$ is one of $S_{\frac{7 \pi}{8}} (\theta_{4}), R_{\frac{2 \pi}{8}} (\theta_{4}), S_{\frac{3 \pi}{8}} (\theta_{4})$ or $R_{\frac{10 \pi}{8}} (\theta_{4})$, and
					\item	$\theta_{1}, \theta_{2}, \theta_{3}, \theta_{4}$ are possibly different real parameters.
				\end{itemize}
		\item	The winning state paths corresponding to $\mathcal{C}_{+}$ and $\mathcal{C}_{-}$ are
				\begin{align} \label{eq:U(2) Winning State Paths}
					\tau_{\mathcal{C}_{+}} = (\ket{0}, \ket{+}, \ket{0}) \quad \text{and} \quad \tau_{\mathcal{C}_{-}} = (\ket{0}, \ket{-}, \ket{0}) \ .
				\end{align}
		\item	Picard has no winning strategy.
	\end{enumerate}
\end{theorem}

This final conclusion is illuminating. Although there are infinite winning strategies, they are equivalent to the strategies we determined when we investigated what happens in $D_{8}$. In this perspective nothing is really gained by enabling Q to pick moves from $U(2)$. In a certain sense the spirit of the game is completely captured when it is realized in $D_{8}$. The next Table \ref{tbl:Q's Equivalence Classes of Winning Strtegies for the PQG in U_2} contains the complete results about Q's classes of winning strategies whether the ambient group belongs to the family of dihedral groups $D_{8 n}$ or is the entire $U(2)$.

\begin{table}[H]
	\centering
	\caption{The two classes of winning strategies for Q in the $PQG$ when th game is played in a dihedral group $D_{8 n}, n \geq 1$ and when the game is played in the largest possible group $U(2)$.}
	\label{tbl:Q's Equivalence Classes of Winning Strtegies for the PQG in U_2}
	\renewcommand{\arraystretch}{2.0}
	\begin{tabular}{ l !{\vrule width 1.25 pt} c | c | c | c }
		\Xhline{4\arrayrulewidth}
		\multicolumn{5}{c}{The ambient group is $D_{8 n}, n \geq 1$}
		\\
		\Xhline{\arrayrulewidth}
		&
		Initial state
		&
		Round $1$
		&
		Round $2$
		&
		Round $3$
		\\
		\Xhline{3\arrayrulewidth}
		$(H, H), (R_{\frac{2 \pi}{8}}, R_{\frac{14 \pi}{8}}), (S_{\frac{5 \pi}{8}}, S_{\frac{5 \pi}{8}}), \dots$
		&
		$\ket{0}$
		&
		\cellcolor{GreenLighter2!20} $\ket{+}$
		&
		\cellcolor{GreenLighter2!20} $\ket{+}$
		&
		$\ket{0}$
		\\
		\Xhline{\arrayrulewidth}
		$(S_{\frac{7 \pi}{8}}, S_{\frac{7 \pi}{8}}), (R_{\frac{14 \pi}{8}}, R_{\frac{2 \pi}{8}}), (S_{\frac{3 \pi}{8}}, S_{\frac{3 \pi}{8}}), \dots$
		&
		$\ket{0}$
		&
		\cellcolor{RedPurple!15} $\ket{-}$
		&
		\cellcolor{RedPurple!15} $\ket{-}$
		&
		$\ket{0}$
		\\
		\Xhline{4\arrayrulewidth}
		\multicolumn{5}{c}{The ambient group is $U(2)$ ($\theta \in \mathbb{R}$)}
		
		\\
		\Xhline{\arrayrulewidth}
		&
		Initial state
		&
		Round $1$
		&
		Round $2$
		&
		Round $3$
		\\
		\Xhline{3\arrayrulewidth}
		$(H (\theta_{1}), H (\theta_{2})), (R_{\frac{2 \pi}{8}} (\theta_{1}), R_{\frac{14 \pi}{8}} (\theta_{2})), \dots$
		&
		$\ket{0}$
		&
		\cellcolor{GreenLighter2!20} $\ket{+}$
		&
		\cellcolor{GreenLighter2!20} $\ket{+}$
		&
		$\ket{0}$
		\\
		\Xhline{\arrayrulewidth}
		$(S_{\frac{7 \pi}{8}} (\theta_{3}), S_{\frac{7 \pi}{8}} (\theta_{4})), (R_{\frac{14 \pi}{8}} (\theta_{3}), R_{\frac{2 \pi}{8}} (\theta_{4})), \dots$
		&
		$\ket{0}$
		&
		\cellcolor{RedPurple!15} $\ket{-}$
		&
		\cellcolor{RedPurple!15} $\ket{-}$
		&
		$\ket{0}$
		\\
		\Xhline{4\arrayrulewidth}
	\end{tabular}
	\renewcommand{\arraystretch}{1.0}
\end{table}

In $U(2)$ the strategy classes contain infinite many strategies, but all these strategies are equivalent to the strategies encountered before.

\section{Extending the game} \label{sec:Extending the PQG}

The original $PQG$ can be extended in numerous ways. In each conceivable extension, the precise formulation of the rules of the game is of paramount importance. By drastically changing the rules it is even possible for Picard to win the game. This was accomplished in \cite{Anand2015} where the authors exploited entanglement in a clever way, so that whether the system ends up in a maximally entangled or separable state determines the outcome. In \cite{Andronikos2018}, it was shown that all possible finite extensions of the $PQG$ can be expressed in terms of simple finite automata, provided that the allowable moves of Picard are either $I$ or $F$ and Q always uses the Hadamard transform $H$. In this paper, our investigation focused on the enlargement of the operational space of the game. Therefore, as we consider possible extensions of the original $PQG$, we adhere to the assumptions (\ref{eq:The Moves of Picard}) and (\ref{eq:The Moves of Q in U(2)}), i.e., $M_P = \{ I, F \}$ and $M_Q = U(2)$. Additionally, we suppose that:

\begin{itemize}
	\item	at the start of the game the coin is in a predefined basis state, which we call the \emph{initial} state and designate by $\ket{q_{0}}$,
	\item	Picard and Q alternate turns acting on the coin following a specified order, and
	\item	when the game ends, the coin is measured in the computational basis; if it is found in state $\ket{q_{P}}$ Picard wins, whereas if it is found in $\ket{q_{Q}}$ then Q wins.
\end{itemize}

It is convenient to refer to $\ket{q_{P}}$ as Picard's \emph{target} state and to $\ket{q_{Q}}$ as Q's \emph{target} state. Obviously, in a zero-sum game the target states $\ket{q_{P}}$ and $\ket{q_{Q}}$ are different. Furthermore, since both Picard and Q draw their moves from groups, nothing is lost in terms of generality if we agree that neither of them is allowed to make consecutive moves. Two or more successive moves by Q can be composed to give just one equivalent move, and the same holds for Picard.

\begin{definition} [Extended games between Picard \& Q] \label{def:Extended PQGs}
	An $n$-round game, $n \geq 2$, is a function that associates one of the players, i.e., Picard or Q, to every round of the game. An $n$-round game is conveniently represented as a sequence of length $n$ from the alphabet $\{ P, Q \}$, where the letters P and Q stand for Picard and Q, respectively. According to our previous remark, if $( K_{1}, K_{2}, \dots, K_{n} )$ is an $n$-round game, then $K_{i} \neq K_{i + 1}$, $1 \leq i < n$.
\end{definition}

Definition \ref{def:Winning and Dominant Strategies} is general enough to also hold for extended games.

We state the next important Theorem \ref{thr:Picard Lacks a Winning Strategy}, which confirms our suspicion that Picard cannot win any such game with probability $1.0$. This negative result implies that for the type of extended games we consider, Picard is at a permanent disadvantage.

\begin{theorem}[Picard lacks a winning strategy] \label{thr:Picard Lacks a Winning Strategy}
	Picard does not have a winning strategy in any $n$-round game, $n \geq 2$, as long as Q makes at least one move.
\end{theorem}

The next couple of theorems explain in which games Q is unable to formulate a winning strategy. First, rather predictably, if Picard is given the opportunity to act last on the coin, then Q cannot surely win.

\begin{theorem}[Q lacks a winning strategy when Picard plays last] \label{thr:Q Lacks a Winning Strategy when Picard Plays Last}
	Q does not have a winning strategy in any $n$-round game, $n \geq 2$, in which Picard makes the last move.
\end{theorem}

Symmetrically, it is also true that Q is unable to devise a winning strategy when Picard playes first, as the next Theorem \ref{thr:Q Lacks a Winning Strategy when Picard Plays First} asserts.

\begin{theorem}[Q lacks a winning strategy when Picard plays first] \label{thr:Q Lacks a Winning Strategy when Picard Plays First}
	Q does not have a winning strategy in any $n$-round game, $n \geq 2$, in which Picard makes the first move.
\end{theorem}

Theorems \ref{thr:Q Lacks a Winning Strategy when Picard Plays Last} and \ref{thr:Q Lacks a Winning Strategy when Picard Plays First} establish that Q cannot surely win in any game where Picard makes the first or the last move. Figures \ref{fig:No Winning Strategy if Picard Has the Last Move} and \ref{fig:No Winning Strategy if Picard Has the First Move} provide a visual explanation of the validity of these two theorems.

\begin{figure}[H]
	\centering
	\begin{tikzpicture} [scale = 1.4]
		\begin{yquant}[operators/every barrier/.append style={red, thick, shorten <= -3mm, shorten >= -3mm}]
			qubit {$\ket{q_{0}}$} GAME;
			hspace {0.3 cm} GAME;
			[ name = Q1 ]
			[ fill = RedPurple!15 ]
			box {$A_{1}$} GAME;
			%
			%
			hspace {0.3 cm} GAME;
			[ name = P1 ]
			[ fill = GreenLighter2!20 ]
			box {$I$} GAME;
			%
			%
			hspace {0.3 cm} GAME;
			[draw=none]
			box {$\dots$} GAME;
			hspace {0.3 cm} GAME;
			[ name = Q2 ]
			[ fill = RedPurple!15 ]
			box {$A_{\frac{n}{2}}$} GAME;
			%
			%
			hspace {0.3 cm} GAME;
			[ name = P2 ]
			[ fill = GreenLighter2!20 ]
			box {$I$} GAME;
			%
			%
			hspace {0.3 cm} GAME;
			[ fill = WordBlueDark!50 ]
			measure GAME;
			hspace {0.3 cm} GAME;
			output {$\ket{?}$} GAME;
			\node [ above = 0.3 cm ] at (Q1) {Q};
			\node [ above = 0.37 cm ] at (P1) {P};
			\node [ above = 0.3 cm ] at (Q2) {Q};
			\node [ above = 0.37 cm ] at (P2) {P};
		\end{yquant}
	\end{tikzpicture}
	\\
	\vspace{0.7 cm}
	\begin{tikzpicture} [scale = 1.4]
		\begin{yquant}[operators/every barrier/.append style={red, thick, shorten <= -3mm, shorten >= -3mm}]
			qubit {$\ket{q_{0}}$} GAME;
			hspace {0.3 cm} GAME;
			[ name = Q1 ]
			[ fill = RedPurple!15 ]
			box {$A_{1}$} GAME;
			%
			%
			hspace {0.3 cm} GAME;
			[ name = P1 ]
			[ fill = GreenLighter2!20 ]
			box {$I$} GAME;
			%
			%
			hspace {0.3 cm} GAME;
			[draw=none]
			box {$\dots$} GAME;
			hspace {0.3 cm} GAME;
			[ name = Q2 ]
			[ fill = RedPurple!15 ]
			box {$A_{\frac{n}{2}}$} GAME;
			%
			%
			hspace {0.3 cm} GAME;
			[ name = P2 ]
			[ fill = GreenLighter2!20 ]
			box {$F$} GAME;
			%
			%
			hspace {0.3 cm} GAME;
			[ fill = WordBlueDark!50 ]
			measure GAME;
			hspace {0.3 cm} GAME;
			output {$\ket{?}$} GAME;
			\node [ above = 0.3 cm ] at (Q1) {Q};
			\node [ above = 0.37 cm ] at (P1) {P};
			\node [ above = 0.3 cm ] at (Q2) {Q};
			\node [ above = 0.37 cm ] at (P2) {P};
		\end{yquant}
	\end{tikzpicture}
	\caption{To intuitively understand why Q cannot have a winning strategy if Picard plays last, it suffices to consider two of Picard's strategies: $\sigma_{P} = ( I, \dots, I, I )$ and $\sigma'_{P} = ( I, \dots, I, F )$. It is impossible for any single strategy $\sigma_{Q} = ( A_{1}, \dots, A_{\frac{n}{2}} )$ of Q to win with probability $1.0$ against both of them.}
	\label{fig:No Winning Strategy if Picard Has the Last Move}
\end{figure}
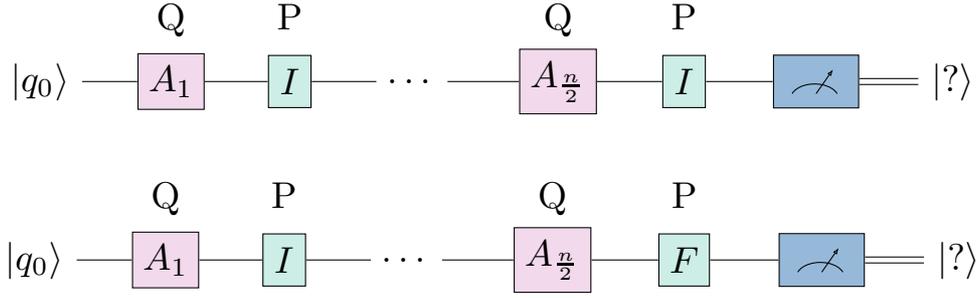

\begin{figure}[H]
	\centering
	\begin{tikzpicture} [scale = 1.4]
		\begin{yquant}[operators/every barrier/.append style={red, thick, shorten <= -3mm, shorten >= -3mm}]
			qubit {$\ket{q_{0}}$} GAME;
			hspace {0.3 cm} GAME;
			[ name = P1 ]
			[ fill = GreenLighter2!20 ]
			box {$I$} GAME;
			%
			%
			hspace {0.3 cm} GAME;
			[ name = Q1 ]
			[ fill = RedPurple!15 ]
			box {$A_{1}$} GAME;
			%
			%
			hspace {0.3 cm} GAME;
			[draw=none]
			box {$\dots$} GAME;
			hspace {0.3 cm} GAME;
			[ name = P2 ]
			[ fill = GreenLighter2!20 ]
			box {$I$} GAME;
			%
			%
			hspace {0.3 cm} GAME;
			[ name = Q2 ]
			[ fill = RedPurple!15 ]
			box {$A_{\frac{n}{2}}$} GAME;
			%
			%
			hspace {0.3 cm} GAME;
			[ fill = WordBlueDark!50 ]
			measure GAME;
			hspace {0.3 cm} GAME;
			output {$\ket{?}$} GAME;
			\node [ above = 0.3 cm ] at (Q1) {Q};
			\node [ above = 0.37 cm ] at (P1) {P};
			\node [ above = 0.3 cm ] at (Q2) {Q};
			\node [ above = 0.37 cm ] at (P2) {P};
		\end{yquant}
	\end{tikzpicture}
	\\
	\vspace{0.7 cm}
	\begin{tikzpicture} [scale = 1.4]
		\begin{yquant}[operators/every barrier/.append style={red, thick, shorten <= -3mm, shorten >= -3mm}]
			qubit {$\ket{q_{0}}$} GAME;
			hspace {0.3 cm} GAME;
			[ name = P1 ]
			[ fill = GreenLighter2!20 ]
			box {$F$} GAME;
			%
			%
			hspace {0.3 cm} GAME;
			[ name = Q1 ]
			[ fill = RedPurple!15 ]
			box {$A_{1}$} GAME;
			%
			%
			hspace {0.3 cm} GAME;
			[draw=none]
			box {$\dots$} GAME;
			hspace {0.3 cm} GAME;
			[ name = P2 ]
			[ fill = GreenLighter2!20 ]
			box {$I$} GAME;
			%
			%
			hspace {0.3 cm} GAME;
			[ name = Q2 ]
			[ fill = RedPurple!15 ]
			box {$A_{\frac{n}{2}}$} GAME;
			%
			%
			hspace {0.3 cm} GAME;
			[ fill = WordBlueDark!50 ]
			measure GAME;
			hspace {0.3 cm} GAME;
			output {$\ket{?}$} GAME;
			\node [ above = 0.3 cm ] at (Q1) {Q};
			\node [ above = 0.37 cm ] at (P1) {P};
			\node [ above = 0.3 cm ] at (Q2) {Q};
			\node [ above = 0.37 cm ] at (P2) {P};
		\end{yquant}
	\end{tikzpicture}
	\caption{To see why Q does not have a winning strategy if Picard plays first, it suffices to consider two of Picard's strategies: $\sigma_{P} = ( I, \dots, I, I )$ and $\sigma'_{P} = ( F, \dots, I, I )$. It is impossible for any single strategy $\sigma_{Q} = ( A_{1}, \dots, A_{\frac{n}{2}} )$ of Q to win with probability $1.0$ against both of them.}
	\label{fig:No Winning Strategy if Picard Has the First Move}
\end{figure}
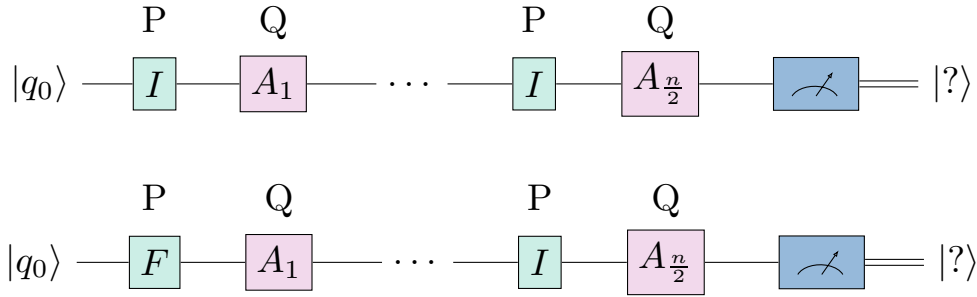

The picture is completed by the next Theorem \ref{thr:When Q Has a Winning Strategy} which asserts that Q has a winning strategy if and only if Q makes the first and the last move. This result clarifies that the overwhelmingly larger repertoire of moves of the quantum player by itself is not enough. It has to be combined with the advantage of making both the first and the last move in order to guarantee that the quantum player will surely win.

\begin{theorem}[When Q possesses a winning strategy] \label{thr:When Q Has a Winning Strategy}
	In any $n$-round game, $n \geq 2$, Q has a winning strategy iff Q makes the first and the last move.
\end{theorem}

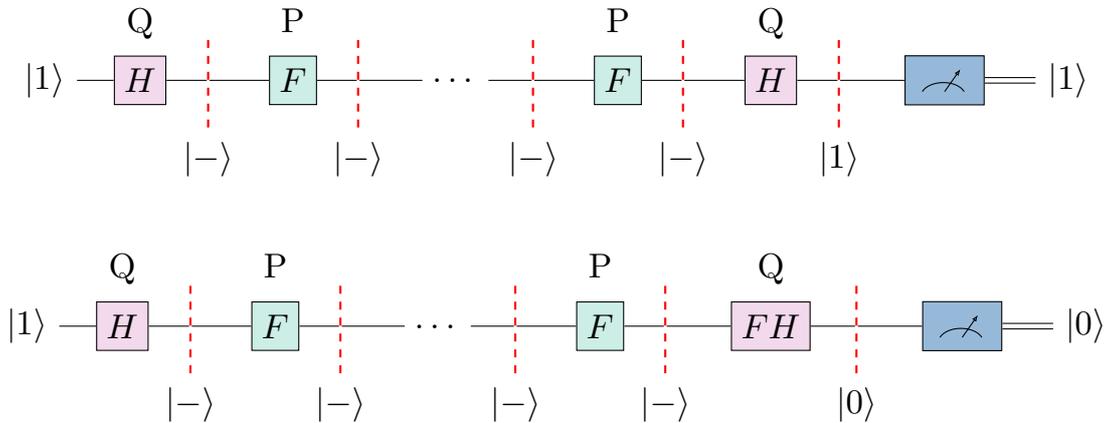
\begin{figure}[H]
	\centering
	\begin{tikzpicture} [scale = 1.3]
		\begin{yquant}[operators/every barrier/.append style={red, thick, shorten <= -3mm, shorten >= -3mm}]
			qubit {$\ket{1}$} GAME;
			hspace {0.2 cm} GAME;
			[ name = Q1 ]
			[ fill = RedPurple!15 ]
			box {$H$} GAME;
			[ name = B1 ]
			barrier GAME;
			hspace {0.2 cm} GAME;
			[ name = P1 ]
			[ fill = GreenLighter2!20 ]
			box {$F$} GAME;
			[ name = B2 ]
			barrier GAME;
			hspace {0.2 cm} GAME;
			[ draw = none ]
			box {$\dots$} GAME;
			[ name = B3 ]
			barrier GAME;
			hspace {0.2 cm} GAME;
			[ name = P2 ]
			[ fill = GreenLighter2!20 ]
			box {$F$} GAME;
			[ name = B4 ]
			barrier GAME;
			hspace {0.2 cm} GAME;
			[ name = Q2 ]
			[ fill = RedPurple!15 ]
			box {$H$} GAME;
			[ name = B5 ]
			barrier GAME;
			hspace {0.2 cm} GAME;
			[ fill = WordBlueDark!50 ]
			measure GAME;
			hspace {0.3 cm} GAME;
			output {$\ket{1}$} GAME;
			\node [ above = 0.3 cm ] at (Q1) {Q};
			\node [ below = 0.5 cm ] at (B1) {$\ket{-}$};
			\node [ above = 0.37 cm ] at (P1) {P};
			\node [ below = 0.5 cm ] at (B2) {$\ket{-}$};
			\node [ below = 0.5 cm ] at (B3) {$\ket{-}$};
			\node [ above = 0.37 cm ] at (P2) {P};
			\node [ below = 0.5 cm ] at (B4) {$\ket{-}$};
			\node [ above = 0.3 cm ] at (Q2) {Q};
			\node [ below = 0.5 cm ] at (B5) {$\ket{1}$};
		\end{yquant}
	\end{tikzpicture}
	\\
	\vspace{0.7 cm}
	\begin{tikzpicture} [scale = 1.3]
		\begin{yquant}[operators/every barrier/.append style={red, thick, shorten <= -3mm, shorten >= -3mm}]
			qubit {$\ket{1}$} GAME;
			hspace {0.2 cm} GAME;
			[ name = Q1 ]
			[ fill = RedPurple!15 ]
			box {$H$} GAME;
			[ name = B1 ]
			barrier GAME;
			hspace {0.2 cm} GAME;
			[ name = P1 ]
			[ fill = GreenLighter2!20 ]
			box {$F$} GAME;
			[ name = B2 ]
			barrier GAME;
			hspace {0.2 cm} GAME;
			[ draw = none ]
			box {$\dots$} GAME;
			[ name = B3 ]
			barrier GAME;
			hspace {0.2 cm} GAME;
			[ name = P2 ]
			[ fill = GreenLighter2!20 ]
			box {$F$} GAME;
			[ name = B4 ]
			barrier GAME;
			hspace {0.2 cm} GAME;
			[ name = Q2 ]
			[ fill = RedPurple!15 ]
			box {$F H$} GAME;
			[ name = B5 ]
			barrier GAME;
			hspace {0.2 cm} GAME;
			[ fill = WordBlueDark!50 ]
			measure GAME;
			hspace {0.3 cm} GAME;
			output {$\ket{0}$} GAME;
			\node [ above = 0.3 cm ] at (Q1) {Q};
			\node [ below = 0.5 cm ] at (B1) {$\ket{-}$};
			\node [ above = 0.37 cm ] at (P1) {P};
			\node [ below = 0.5 cm ] at (B2) {$\ket{-}$};
			\node [ below = 0.5 cm ] at (B3) {$\ket{-}$};
			\node [ above = 0.37 cm ] at (P2) {P};
			\node [ below = 0.5 cm ] at (B4) {$\ket{-}$};
			\node [ above = 0.3 cm ] at (Q2) {Q};
			\node [ below = 0.5 cm ] at (B5) {$\ket{0}$};
		\end{yquant}
	\end{tikzpicture}
	\caption{This figure depicts two winning strategies for Q two an $n$-round games where Q makes the first and the last move. In the first game the initial state of the coin and the target state for Q are both the same, i.e., $\ket{1}$. In the second game the initial state of the coin is also $\ket{1}$, but the target state for Q now is $\ket{0}$.}
	\label{fig:Winning Strategies for Q}
\end{figure}

One last conclusion we may draw from the above theorems is about the initial state of the coin and the target states of the players. In the original $PQG$ the both the initial state and Q's target state were the same, namely $\ket{0}$. Clearly, the particular choice of the initial state and the target states is of no importance. If there exists a winning strategy for Q with respect to specific initial and target states, then there exists a winning strategy for every combination of initial and target states.

\begin{corollary}[The impact of initial and target states] \label{crl:The Impact of Initial and Target States is Negligible}
	In any $n$-round game, $n \geq 2$, if Q has a winning strategy, then he has a winning strategy for every combination of initial and target states.
\end{corollary}

\section{Conclusions} \label{sec:Conclusions}

Quantum games not only pose many interesting questions, but also motivate research that can have important applications in other related fields such quantum algorithms and quantum key distribution. This work was inspired by the iconic PQ penny flip game that, undoubtedly, helped create the field. We have approached this game using concepts from group theory. This allowed us to uncover the interesting connection of the original game with the dihedral group $D_{8}$. Interpreting the game in terms of stabilizers and fixed sets enabled us to easily explain and replicate Q's strategy. This in turn allowed to prove that there exist precisely two classes of winning strategies for Q. Each class contains many different strategies, but all these strategies are equivalent in the sense that they drive the coin through the same sequence of states. We established that there are exactly two different sequences of states that can guaranteed Q's win with probability $1.0$. What is noteworthy is the realization that even when the game takes place in larger dihedral groups or even in the entire $U(2)$, this fact remains true. The essence of the game can be succinctly summarized by saying that there are precisely two paths that lead to Q's win and, of course, no path that leads to Picard's win. When we examined extensions of the game without any restriction, we discovered a very important fact, namely that for the quantum player to surely win against the classical player the tremendous advantage of quantum actions is not enough. Q must also make the first and the last move, or else he is not certain to win.

There still many question to be answered, a lot of important issues have to be addressed. These results were based on the assumption that the initial state of the coin and the target states of the players were one the the basis states. If this is not the case, and, moreover, entanglement comes into play, how does that affect the progression of the game? As this is, in our view, a particularly interesting topic, we expect it to be the subject of a future work.

\appendix \label{Appendix A}

\appendixpage

\section{Proofs of the main results}

\subsection{Proofs for Section \ref{sec:Connecting PQG & $D_8$}}

We begin this Appendix by first giving the rather easy proof of Theorem \ref{thr:PQG Ambient Group}.

\begin{ManualTheorem}{3.1}[The ambient group of the $PQG$] \label{thr:PQG Ambient Group Appendix}
	The ambient group of the $PQG$ is $D_8$.
\end{ManualTheorem}
\begin{proof}[Proof of Theorem \ref{thr:PQG Ambient Group Appendix}]
{\small
	Let us recall the presentation (\ref{eq:Dihedral Group Presentation 1}) for the general dihedral group $D_n$.
	\begin{align}
		D_n = \langle s, t \ | \ s^2 = t^2 = (s t)^{n} = \mathds{1} \rangle \ \tag{$P_{1}$}
	\end{align}
	By making the concrete associations
	\begin{align} \label{eq:Associating 1, s, t To I, F, H}
		\mathds{1} \mapsto I
		\overset{(\ref{eq:PQ Matrices})}{=}
		\begin{bmatrix}
			1 & 0 \\
			0 & 1
		\end{bmatrix}, \quad
		s \mapsto F
		\overset{(\ref{eq:PQ Matrices})}{=}
		\begin{bmatrix}
			0 & 1 \\
			1 & 0
		\end{bmatrix}, \quad \text{and} \quad
		t \mapsto H
		\overset{(\ref{eq:PQ Matrices})}{=}
		\begin{bmatrix}
			\begin{array}{lr}
				\frac{\sqrt{2}}{2} & \frac{\sqrt{2}}{2} \\
				\frac{\sqrt{2}}{2} & -\frac{\sqrt{2}}{2}
			\end{array}
		\end{bmatrix} \ ,
	\end{align}
	we can readily verify the following facts:
	\begin{enumerate}
		\item	$
				F^{2} =
				\begin{bmatrix}
					1 & 0 \\
					0 & 1
				\end{bmatrix}
				= I,
				$
		\item	$
				H^{2} =
				\begin{bmatrix}
					1 & 0 \\
					0 & 1
				\end{bmatrix}
				= I,
				$ and
		\item	$
				FH =
				\begin{bmatrix}
					\begin{array}{lr}
						\cos \frac{2 \pi}{8} & - \sin \frac{2 \pi}{8} \\
						\sin \frac{2 \pi}{8} &   \cos \frac{2 \pi}{8}
					\end{array}
				\end{bmatrix}
				\overset{(\ref{eq:2D Rotator})}{=} R_{\frac{2 \pi}{8}},
				$
		\item	$
				(FH)^{k} =
				\begin{bmatrix}
					\begin{array}{lr}
						\cos \frac{2 \pi k}{8} & - \sin \frac{2 \pi k}{8} \\
						\sin \frac{2 \pi k}{8} &   \cos \frac{2 \pi k}{8}
					\end{array}
				\end{bmatrix}
				\overset{(\ref{eq:2D Rotator})}{=} R_{\frac{2 \pi k}{8}}
				$,
				for $0 \leq k \leq 7$, and
		\item	$
				(FH)^{8} =
				\begin{bmatrix}
					1 & 0 \\
					0 & 1
				\end{bmatrix}
				= I,
				$
	\end{enumerate}
	Hence, presentation (\ref{eq:Dihedral Group Presentation 1}) is satisfied by $F$ and $H$ for $n = 8$, meaning that $F$ and $H$ generate the dihedral group $D_8$: $D_8 = \langle F, H \rangle$.
}
\end{proof}

For some of the remaining proofs, it will be convenient to explicitly give the \emph{standard} matrix representation of $D_8$ by listing for every one of its elements the corresponding $2 \times 2$ matrix in the following Tables \ref{tbl:D8 Standard Representation of Rotations} and \ref{tbl:D8 Standard Representation of Reflections}.

\begin{table}[H]
	\centering
	\caption{The standard representation for each of the 8 rotations of the dihedral group $D_8$.}
	\label{tbl:D8 Standard Representation of Rotations}
	\footnotesize
	{
		\begin{tabular}{ !{\vrule width 1.0 pt} l || l !{\vrule width 1.0 pt} }
			\Xhline{2\arrayrulewidth}
			\cellcolor{RedPurple!15}
			&
			\cellcolor{RedPurple!15}
			\\
			\cellcolor{RedPurple!15}
			$
			\mathds{1} \mapsto R_{0} =
			\begin{bmatrix}
				\begin{array}{lr}
					\cos 0 & - \sin 0 \\
					\sin 0 &   \cos 0
				\end{array}
			\end{bmatrix}
			=
			\begin{bmatrix}
				\begin{array}{lr}
					1 & 0 \\
					0 & 1
				\end{array}
			\end{bmatrix} = I
			$
			&
			\cellcolor{RedPurple!15}
			$
			r \mapsto R_{\frac{2 \pi}{8}} =
			\begin{bmatrix}
				\begin{array}{lr}
					\cos \frac{2 \pi}{8} & - \sin \frac{2 \pi}{8} \\
					\sin \frac{2 \pi}{8} &   \cos \frac{2 \pi}{8}
				\end{array}
			\end{bmatrix}
			=
			\begin{bmatrix}
				\begin{array}{lr}
					\frac{\sqrt{2}}{2} & - \frac{\sqrt{2}}{2} \\
					\frac{\sqrt{2}}{2} &   \frac{\sqrt{2}}{2}
				\end{array}
			\end{bmatrix}
			$
			\\ [4 ex]
			\cellcolor{RedPurple!15}
			$
			r^{2} \mapsto R_{\frac{4 \pi}{8}} =
			\begin{bmatrix}
				\begin{array}{lr}
					\cos \frac{4 \pi}{8} & - \sin \frac{4 \pi}{8} \\
					\sin \frac{4 \pi}{8} &   \cos \frac{4 \pi}{8}
				\end{array}
			\end{bmatrix}
			=
			\begin{bmatrix}
				\begin{array}{lr}
					0 & - 1 \\
					1 & 0
				\end{array}
			\end{bmatrix}
			$
			&
			\cellcolor{RedPurple!15}
			$
			r^{3} \mapsto R_{\frac{6 \pi}{8}} =
			\begin{bmatrix}
				\begin{array}{lr}
					\cos \frac{6 \pi}{8} & - \sin \frac{6 \pi}{8} \\
					\sin \frac{6 \pi}{8} &   \cos \frac{6 \pi}{8}
				\end{array}
			\end{bmatrix}
			=
			\begin{bmatrix}
				\begin{array}{rr}
						- \frac{\sqrt{2}}{2} & - \frac{\sqrt{2}}{2} \\
						\frac{\sqrt{2}}{2} & - \frac{\sqrt{2}}{2}
				\end{array}
			\end{bmatrix}
			$
			\\ [4 ex]
			\cellcolor{RedPurple!15}
			$
			r^{4} \mapsto R_{\frac{8 \pi}{8}} =
			\begin{bmatrix}
				\begin{array}{lr}
					\cos \frac{8 \pi}{8} & - \sin \frac{8 \pi}{8} \\
					\sin \frac{8 \pi}{8} &   \cos \frac{8 \pi}{8}
				\end{array}
			\end{bmatrix}
			=
			\begin{bmatrix}
				\begin{array}{rr}
					-1 & 0 \\
					0 & -1
				\end{array}
			\end{bmatrix}
			$
			&
			\cellcolor{RedPurple!15}
			$
			r^{5} \mapsto R_{\frac{10 \pi}{8}} =
			\begin{bmatrix}
				\begin{array}{lr}
					\cos \frac{10 \pi}{8} & - \sin \frac{10 \pi}{8} \\
					\sin \frac{10 \pi}{8} &   \cos \frac{10 \pi}{8}
				\end{array}
			\end{bmatrix}
			=
			\begin{bmatrix}
				\begin{array}{lr}
						- \frac{\sqrt{2}}{2} &   \frac{\sqrt{2}}{2} \\
						- \frac{\sqrt{2}}{2} & - \frac{\sqrt{2}}{2}
				\end{array}
			\end{bmatrix}
			$
			\\ [4 ex]
			\cellcolor{RedPurple!15}
			$
			r^{6} \mapsto R_{\frac{12 \pi}{8}} =
			\begin{bmatrix}
				\begin{array}{lr}
					\cos \frac{12 \pi}{8} & - \sin \frac{12 \pi}{8} \\
					\sin \frac{12 \pi}{8} &   \cos \frac{12 \pi}{8}
				\end{array}
			\end{bmatrix}
			=
			\begin{bmatrix}
				\begin{array}{rr}
						0 & 1 \\
						-1 & 0
				\end{array}
			\end{bmatrix}
			$
			&
			\cellcolor{RedPurple!15}
			$
			r^{7} \mapsto R_{\frac{14 \pi}{8}} =
			\begin{bmatrix}
				\begin{array}{lr}
					\cos \frac{14 \pi}{8} & - \sin \frac{14 \pi}{8} \\
					\sin \frac{14 \pi}{8} &   \cos \frac{14 \pi}{8}
				\end{array}
			\end{bmatrix}
			=
			\begin{bmatrix}
				\begin{array}{rr}
						\frac{\sqrt{2}}{2} & \frac{\sqrt{2}}{2} \\ 
						- \frac{\sqrt{2}}{2} & \frac{\sqrt{2}}{2} 
				\end{array}
			\end{bmatrix}
			$
			\\
			\cellcolor{RedPurple!15}
			&
			\cellcolor{RedPurple!15}
			\\
			\Xhline{2\arrayrulewidth}
		\end{tabular}
	}
\end{table}

\begin{table}[H]
	\centering
	\caption{The standard representation for each of the 8 reflections of the dihedral group $D_8$.}
	\label{tbl:D8 Standard Representation of Reflections}
	\footnotesize
	{
		\begin{tabular}{ !{\vrule width 1.0 pt} l || l !{\vrule width 1.0 pt} }
			\Xhline{2\arrayrulewidth}
			\cellcolor{GreenLighter2!20}
			&
			\cellcolor{GreenLighter2!20}
			\\
			\cellcolor{GreenLighter2!20}
			$
			s \mapsto S_{0} =
			\begin{bmatrix}
				\begin{array}{lr}
					\cos 0 &   \sin 0 \\
					\sin 0 & - \cos 0
				\end{array}
			\end{bmatrix}
			=
			\begin{bmatrix}
				\begin{array}{lr}
					1 &   0 \\
					0 & - 1
				\end{array}
			\end{bmatrix}
			$
			&
			\cellcolor{GreenLighter2!20}
			$
			rs \mapsto S_{\frac{\pi}{8}} =
			\begin{bmatrix}
				\begin{array}{lr}
					\cos \frac{2 \pi}{8} &   \sin \frac{2 \pi}{8} \\
					\sin \frac{2 \pi}{8} & - \cos \frac{2 \pi}{8}
				\end{array}
			\end{bmatrix}
			=
			\begin{bmatrix}
				\begin{array}{lr}
					\frac{\sqrt{2}}{2} &   \frac{\sqrt{2}}{2} \\
					\frac{\sqrt{2}}{2} & - \frac{\sqrt{2}}{2}
				\end{array}
			\end{bmatrix} = H
			$
			\\ [4 ex]
			\cellcolor{GreenLighter2!20}
			$
			r^{2} s \mapsto S_{\frac{2 \pi}{8}} =
			\begin{bmatrix}
				\begin{array}{lr}
					\cos \frac{4 \pi}{8} &   \sin \frac{4 \pi}{8} \\
					\sin \frac{4 \pi}{8} & - \cos \frac{4 \pi}{8}
				\end{array}
			\end{bmatrix}
			=
			\begin{bmatrix}
				\begin{array}{lr}
					0 & 1 \\
					1 & 0
				\end{array}
			\end{bmatrix} = F
			$
			&
			\cellcolor{GreenLighter2!20}
			$
			r^{3} s \mapsto S_{\frac{3 \pi}{8}} =
			\begin{bmatrix}
				\begin{array}{lr}
					\cos \frac{6 \pi}{8} &   \sin \frac{6 \pi}{8} \\
					\sin \frac{6 \pi}{8} & - \cos \frac{6 \pi}{8}
				\end{array}
			\end{bmatrix}
			=
			\begin{bmatrix}
				\begin{array}{rr}
					- \frac{\sqrt{2}}{2} & \frac{\sqrt{2}}{2} \\
					\frac{\sqrt{2}}{2} & \frac{\sqrt{2}}{2}
				\end{array}
			\end{bmatrix}
			$
			\\ [4 ex]
			\cellcolor{GreenLighter2!20}
			$
			r^{4} s \mapsto S_{\frac{4 \pi}{8}} =
			\begin{bmatrix}
				\begin{array}{lr}
					\cos \frac{8 \pi}{8} &   \sin \frac{8 \pi}{8} \\
					\sin \frac{8 \pi}{8} & - \cos \frac{8 \pi}{8}
				\end{array}
			\end{bmatrix}
			=
			\begin{bmatrix}
				\begin{array}{rr}
					- 1 & 0 \\
					0 & 1
				\end{array}
			\end{bmatrix}
			$
			&
			\cellcolor{GreenLighter2!20}
			$
			r^{5} s \mapsto S_{\frac{5 \pi}{8}} =
			\begin{bmatrix}
				\begin{array}{lr}
					\cos \frac{10 \pi}{8} & \sin \frac{10 \pi}{8} \\
					\sin \frac{10 \pi}{8} & - \cos \frac{10 \pi}{8}
				\end{array}
			\end{bmatrix}
			=
			\begin{bmatrix}
				\begin{array}{lr}
					- \frac{\sqrt{2}}{2} & - \frac{\sqrt{2}}{2} \\
					- \frac{\sqrt{2}}{2} &   \frac{\sqrt{2}}{2}
				\end{array}
			\end{bmatrix}
			$
			\\ [4 ex]
			\cellcolor{GreenLighter2!20}
			$
			r^{6} s \mapsto S_{\frac{6 \pi}{8}} =
			\begin{bmatrix}
				\begin{array}{lr}
					\cos \frac{12 \pi}{8} &   \sin \frac{12 \pi}{8} \\
					\sin \frac{12 \pi}{8} & - \cos \frac{12 \pi}{8}
				\end{array}
			\end{bmatrix}
			=
			\begin{bmatrix}
				\begin{array}{rr}
					0 & - 1 \\
					- 1 &   0
				\end{array}
			\end{bmatrix}
			$
			&
			\cellcolor{GreenLighter2!20}
			$
			r^{7} s \mapsto S_{\frac{7 \pi}{8}} =
			\begin{bmatrix}
				\begin{array}{lr}
					\cos \frac{14 \pi}{8} &   \sin \frac{14 \pi}{8} \\
					\sin \frac{14 \pi}{8} & - \cos \frac{14 \pi}{8}
				\end{array}
			\end{bmatrix}
			=
			\begin{bmatrix}
				\begin{array}{rr}
					\frac{\sqrt{2}}{2} & - \frac{\sqrt{2}}{2} \\ 
					- \frac{\sqrt{2}}{2} & - \frac{\sqrt{2}}{2} 
				\end{array}
			\end{bmatrix}
			$
			\\
			\cellcolor{GreenLighter2!20}
			&
			\cellcolor{GreenLighter2!20}
			\\
			\Xhline{2\arrayrulewidth}
		\end{tabular}
	}
\end{table}

\subsection{Proofs for Section \ref{sec:Group Theoretic Analysis of PQG}}

It is quite straightforward to follow and verify the proofs given below, by keeping in mind that:

\begin{itemize}
	\item	under the matrix representation of the dihedral groups, the action of a dihedral group on any state of the quantum coin can be determined by simply multiplying the matrices representing the elements of the group with the ket corresponding to the state, and
	\item	a ket of the form $e^{i \theta} \ket{ \psi }$, with $\theta \in \mathbb{R}$, represents the same state as the ket $\ket{ \psi }$.
\end{itemize}

\begin{ManualProposition}{4.1}[The action of $D_{8}$ on $B$] \label{prp:The Action of D_8 on B Appendix} \
	\begin{enumerate}
		\item	$\ket{0}$ and $\ket{1}$ have the same orbit:
		\begin{align} \label{eq:Orbits of Ket 0 and Ket 1 in D_8 Appendix}
			D_{8} \star \ket{0} = D_{8} \star \ket{1} = \{ \ket{0}, \ket{+}, \ket{1}, \ket{-} \} \ .
		\end{align}
		\item	The orbit of $B$ is:
		\begin{align} \label{eq:Orbit of B in D_8 Appendix}
			D_{8} \star B = \{ \ket{0}, \ket{+}, \ket{1}, \ket{-} \} \ .
		\end{align}
	\end{enumerate}
\end{ManualProposition}
\begin{proof}[Proof of Proposition \ref{prp:The Action of D_8 on B Appendix}]
{\small
	We will make use of the standard matrix representation of the rotations and reflections of $D_8$ as given in Tables \ref{tbl:D8 Standard Representation of Rotations} and \ref{tbl:D8 Standard Representation of Reflections}.
	\begin{enumerate}
		\item	By systematically multiplying all the matrices in Tables \ref{tbl:D8 Standard Representation of Rotations} and \ref{tbl:D8 Standard Representation of Reflections} with $\ket{0}$ we get $\ket{0}, \ket{+}, \ket{1}, -\ket{-},$ $-\ket{0}, -\ket{+}, -\ket{1}$ and $\ket{-}$. Of course, $\ket{0}$ and $-\ket{0}$ represent the same state. This also applies to the pairs $\ket{1}$ and $-\ket{1}$, $\ket{+}$ and $-\ket{+}$, $\ket{-}$ and $-\ket{-}$. Thus, $D_{8} \star \ket{0} =\{ \ket{0}, \ket{+}, \ket{1}, \ket{-} \}$. In a symmetrical fashion, we may compute $D_{8} \star \ket{1}$ and verify that (\ref{eq:Orbits of Ket 0 and Ket 1 in D_8 Appendix}) holds.
		\item	Simply taking the union of the orbits $D_{8} \star \ket{0}$ and $D_{8} \star \ket{1}$ gives the desired result.
	\end{enumerate}
}
\end{proof}

\begin{ManualProposition}{4.2}[The stabilizers of $\ket{0}, \ket{+}, \ket{1}$ and $\ket{-}$ in $D_8$] \label{prp:The Stabilizers of 0, +, 1, - in D_8 Appendix} \
	\begin{itemize}
	\item	The stabilizers of $\ket{0}$ and $\ket{1}$ in $D_8$ are
	\begin{align} \label{eq:The Stabilizers of Ket 0 & Ket 1 in D_8 Appendix}
		D_{8} ( \ket{0} ) = \{ I, R_{\pi}, S_{0}, S_{\frac{4 \pi}{8}} \}
		\quad \text{and} \quad
		D_{8} ( \ket{1} ) = \{ I, R_{\pi}, S_{0}, S_{\frac{4 \pi}{8}} \} \ .
	\end{align}
	\item	The stabilizers of $\ket{+}$ and $\ket{-}$ are
	\begin{align} \label{eq:The Stabilizers of + & - in D_8 Appendix}
		D_{8} ( \ket{+} ) = \{ I, R_{\pi}, F, S_{\frac{6 \pi}{8}} \}
		\quad \text{and} \quad
		D_{8} ( \ket{-} ) = \{ I, R_{\pi}, F, S_{\frac{6 \pi}{8}} \} \ .
	\end{align}
\end{itemize}
\end{ManualProposition}
\begin{proof}[Proof of Proposition \ref{prp:The Stabilizers of 0, +, 1, - in D_8 Appendix}]
{\small
	We will only show how to find the stabilizer of $\ket{+}$, since the proofs regarding the states $\ket{0}, \ket{1}$ and $\ket{-}$ are completely analogous. It suffices to exhaustively multiply every matrix appearing in Tables \ref{tbl:D8 Standard Representation of Rotations} and \ref{tbl:D8 Standard Representation of Reflections} with $\ket{+}$ and note for which matrices the outcome is again $\ket{+}$. The stabilizer of $\ket{+}$ will contain precisely these matrices. These are the two rotations $I$ and $R_{\pi}$, through angles zero and $\pi$ (see Figure \ref{fig:Rotation Symmetries Regular Octagon}), and the two reflections $F$, about the line passing through vertices $2$ and $6$, and $S_{\frac{6 \pi}{8}}$, about the line passing through vertices $4$ and $8$ (see Figure \ref{fig:Reflection Symmetries Regular Octagon}). Therefore, we conclude that $D_{8} ( \ket{+} ) = \{ I, R_{\pi}, F, S_{\frac{6 \pi}{8}} \}$.
}
\end{proof}

\begin{ManualProposition}{4.3}[The fixed set of $\{ I, F \}$ in $D_8$] \label{prp:The Fixed Set of M_P in D_8 Appendix} \
	\begin{enumerate}
		\item	The fixed set of $F$ in $D_8$ is the set
		\begin{align} \label{eq:The Fixed Set of F in D_8 Appendix}
			Fix ( F ) = \{ \ket{+}, \ket{-} \} \ .
		\end{align}
		\item	The fixed set of $M_P = \{ I, F \}$ in $D_8$ is the set
		\begin{align} \label{eq:The Fixed Set of M_P in D_8 Appendix}
			Fix ( \{ I, F \} ) = \{ \ket{+}, \ket{-} \} \ .
		\end{align}
	\end{enumerate}
\end{ManualProposition}
\begin{proof}[Proof of Proposition \ref{prp:The Fixed Set of M_P in D_8 Appendix}] \
{\small
	\begin{enumerate}
		\item	We know from (\ref{eq:Orbit of B in D_8 Appendix}) that in $D_8$ the coin can be in one the states contained in $D_{8} \star B = \{ \ket{0}, \ket{+},$ $\ket{1},$ $\ket{-} \}$. By successively multiplying $F$ with these states, we find that: $F \ket{0} = \ket{1}$, $F \ket{+} = \ket{+}$, $F \ket{1} = \ket{0}$, and $F \ket{-} = \ket{-}$. These results show that the action of $F$ on the states $\ket{+}$ and $\ket{-}$ does not change the state of the coin. Hence, $Fix ( F ) = \{ \ket{+}, \ket{-} \}$ and (\ref{eq:The Fixed Set of F in D_8 Appendix}) holds.
		\item	The identity $I$ fixes every state in the orbit, so the intersection of $Fix ( I )$ with $Fix ( F )$ is just $Fix ( F )$, which verifies (\ref{eq:The Fixed Set of M_P in D_8 Appendix}).
	\end{enumerate}
}
\end{proof}

\subsection{Proofs for Section \ref{sec:Enlarging the Operational Space of the PQG}}

\begin{ManualTheorem}{5.1}[Characteristic properties of winning strategies] \label{thr:Characteristic Properties of Winning Strategies Appendix}
	If $(A_1, A_{2})$ is a winning strategy for Q, then:
	\begin{align}
		A_2 I A_1 \ket{0} &= A_2 F A_1 \ket{0} = \ket{0} \ , \label{eq:Characteristic Property I of Winning Strategies Appendix}
		\qquad \text{and}
		\\
		A_1 \ket{0} &\in Fix ( \{ F \} ) \ . \label{eq:Characteristic Property II of Winning Strategies Appendix}
	\end{align}
\end{ManualTheorem}
\begin{proof}[Proof of Theorem \ref{thr:Characteristic Properties of Winning Strategies Appendix}]
{\small
	By Definition \ref{def:Winning and Dominant Strategies} $(A_1, A_{2})$ is a winning strategy for Q if for every strategy of Picard, Q wins the game with probability $1.0$. If the coin, just prior to measurement, is in a state $a \ket{0} + b \ket{1}$, with $b \neq 0$, then the probability that Q will win the game is strictly less than $1.0$. Therefore, every winning strategy must eventually drive the coin to the state $\ket{0}$, no matter what Picard plays. This implies that $A_2 I A_1 \ket{0} = A_2 F A_1 \ket{0} = \ket{0}$.

	Let us assume in order to reach a contradiction that $\ket{\psi} = A_1 \ket{0}$ is not fixed by $F$. Then $F \ket{\psi} = \ket{\psi'}$, where state $\ket{\psi'}$ is different from state $\ket{\psi}$. However, according to (\ref{eq:Characteristic Property I of Winning Strategies Appendix}), $A_2 I A_1 \ket{0} = A_2 F A_1 \ket{0}$ $\Rightarrow$ $A_2 I \ket{\psi} = A_2 F \ket{\psi}$ $\Rightarrow$ $A_2 \ket{\psi} = A_2 \ket{\psi'}$ $\Rightarrow$ $\ket{\psi} = \ket{\psi'}$, which  contradicts our assumption that $\ket{\psi}$ and $\ket{\psi'}$ are different states. The last implication is valid because $A_2$, as a group element, has a unique inverse.
}
\end{proof}

\begin{ManualTheorem}{5.2}[The ambient group of the $PQG$ is $D_8$] \label{thr:Q's PQG Winning Strategies in D_8 Appendix}
	If we assume that $M_P = \{ I, F \}$ and $M_Q = D_8$, i.e., the ambient group of the $PQG$ is $D_8$, then the following hold.
	\begin{enumerate}
		\item	Q has exactly two classes of winning and dominant strategies
				\begin{align} \label{eq:D_8 Winning Strategy Classes Appendix}
					\mathcal{C}_{+} = [ (H, H) ] \quad \text{and} \quad \mathcal{C}_{-} = [ (S_{\frac{7 \pi}{8}}, S_{\frac{7 \pi}{8}}) ] \ ,
				\end{align}
				each containing $16$ equivalent strategies.
		\item	The winning state paths corresponding to $\mathcal{C}_{+}$ and $\mathcal{C}_{-}$ are
				\begin{align} \label{eq:D_8 Winning State Paths Appendix}
					\tau_{\mathcal{C}_{+}} = (\ket{0}, \ket{+}, \ket{0}) \quad \text{and} \quad \tau_{\mathcal{C}_{-}} = (\ket{0}, \ket{-}, \ket{0}) \ .
				\end{align}
		\item	Picard has no winning strategy.
	\end{enumerate}
\end{ManualTheorem}
\begin{proof}[Proof of Theorem \ref{thr:Q's PQG Winning Strategies in D_8 Appendix}]
{\small
	If the ambient group is $D_8$, the quantum coin will be in one of the states of the orbit $D_{8} \star B$, where $B$ is the computational basis. From (\ref{eq:Orbit of B in D_8 Appendix}) we know that $D_{8} \star B = \{ \ket{0}, \ket{+}, \ket{1}, \ket{-} \}$. Let $\sigma = (A_{1}, A_{2})$ be a winning strategy for Q. According to (\ref{eq:Characteristic Property I of Winning Strategies Appendix}), $A_2 I A_1 \ket{0} = A_2 F A_1 \ket{0} = \ket{0}$.
	\begin{enumerate}
		\item	We may distinguish $4$ cases, depending on the state of the coin after Q's first move $A_{1}$.
				\begin{enumerate}
				\item[(i)]	If Q leaves the coin at state $\ket{0}$, i.e., $A_1 \ket{0} = \ket{0}$, then (\ref{eq:Characteristic Property I of Winning Strategies Appendix}) implies that $A_2 I \ket{0} = A_2 F \ket{0} = \ket{0} \Rightarrow A_2 \ket{0} = A_2 \ket{1} = \ket{0} \Rightarrow \ket{0} = \ket{1}$, which is absurd. To arrive at this contradiction, we have used the fact that $A_2$, as a group element, has a unique inverse. This result shows that the first move of every winning strategy for Q must drive the coin to a state other than $\ket{0}$.
				\item[(ii)]	If Q sends the coin to state $\ket{1}$, i.e., $A_1 \ket{0} = \ket{1}$, then (\ref{eq:Characteristic Property I of Winning Strategies Appendix}) implies that $A_2 I \ket{1} = A_2 F \ket{1} = \ket{0} \Rightarrow A_2 \ket{1} = A_2 \ket{0} = \ket{0} \Rightarrow \ket{1} = \ket{0}$, which is also absurd for the same reason as in the previous case. Hence, the first move of every winning strategy for Q cannot send the coin to state $\ket{1}$.
				\item[(iii)]	If Q sends the coin to state $\ket{+}$, which can can be achieved through $4$ different ways: $H, R_{\frac{2 \pi}{8}}, S_{\frac{5 \pi}{8}}$ and $R_{\frac{10 \pi}{8}}$, then, no matter what Picard plays, the coin will remain in this state because $\ket{+}$ is fixed by $I$ and $F$, according to (\ref{eq:The Fixed Set of M_P in D_8 Appendix}). Finally, Q can send the coin back to the $\ket{0}$ state with $4$ different ways: $H, R_{\frac{14 \pi}{8}}, S_{\frac{5 \pi}{8}}$ and $R_{\frac{6 \pi}{8}}$. This means that Q has $16$ different winning strategies, which, in view of Definition \ref{def:Equivalent Strategies}, are equivalent. Thus, they constitute one equivalence class of $16$ winning strategies, which we designate by $\mathcal{C}_{+}$. Any one of them, e.g., $(H, H)$ can be taken as a representative of this class, so we may write $\mathcal{C}_{+} = [ (H, H) ]$.
				\item[(iv)]	In an analogous way, Q can send the coin to state $\ket{-}$ using $4$ different moves: $S_{\frac{7 \pi}{8}}, R_{\frac{14 \pi}{8}}, S_{\frac{3 \pi}{8}}$ or $R_{\frac{6 \pi}{8}}$. Picard is unable to change this state because $\ket{-}$ is also fixed by $I$ and $F$, according to (\ref{eq:The Fixed Set of M_P in D_8 Appendix}). This enables Q to send the coin back to $\ket{0}$ with $4$ different ways: $S_{\frac{7 \pi}{8}}, R_{\frac{2 \pi}{8}}, S_{\frac{3 \pi}{8}}$ or $R_{\frac{10 \pi}{8}}$. Once again Q has $16$ different winning strategies, which, in view of Definition \ref{def:Equivalent Strategies}, are equivalent. They make the second equivalence class of $16$ winning strategies, which is denoted by $\mathcal{C}_{-}$. Any one of them, for instance $(S_{\frac{7 \pi}{8}}, S_{\frac{7 \pi}{8}})$, can be taken as a representative of this class, so we may write $\mathcal{C}_{-} = [ (S_{\frac{7 \pi}{8}}, S_{\frac{7 \pi}{8}}) ]$.
		\end{enumerate}
				This concludes the proof of (\ref{eq:D_8 Winning Strategy Classes Appendix}).
		\item	Based on the above analysis of cases $(iii)$ and $(iv)$ it is straightforward to verify (\ref{eq:D_8 Winning State Paths Appendix}).
		\item	By Definition \ref{def:Winning and Dominant Strategies}, Picard has no winning strategy because if Q employs one of his winning strategies, Picard has $0.0$ probability to win the game.
	\end{enumerate}
}
\end{proof}

\begin{ManualTheorem}{5.3}[The smallest dihedral group for the $PQG$ is $D_8$] \label{thr:The smallest dihedral group of the PQG is D_8 Appendix}
	$D_8$ is the smallest of the dihedral groups such that $PQG$ can be meaningful played and in which Q has a quantum winning strategy.
\end{ManualTheorem}
\begin{proof}[Proof of Theorem \ref{thr:The smallest dihedral group of the PQG is D_8 Appendix}]
{\small
	Let us first clearly state the two assumptions on which this result is based:
	\begin{enumerate}
		\item	Picard's set of moves $M_P$ is $\{ I, F \}$, according to assumption (\ref{eq:The Moves of Picard}). As a classical player, Picard must certainly be able to flip the coin, or else the game will be meaningless. On the other hand, he should not be able to employ a true quantum move.
		\item	Q's actions $M_Q$ should contain at least one unitary operator other than the classical $I$ and $F$ operators, in order to exhibit quantumness.
	\end{enumerate}
	With the above clarifications in mind, let us examine whether any of the smaller dihedral groups $D_3, D_4, D_5, D_6$ and $D_7$ can serve as the operational space for a meaningful, or at least nontrivial, realization of the $PQG$.
	\begin{itemize}
		\item	The dihedral group $D_3$ does not contain the reflection $F$. One can verify this by comparing formula (\ref{eq:The F Reflector}) with formula (\ref{eq:Standard Representation r^k s in D_n}) for $n = 3$ and $k = 0, 1, 2$. This shows that $D_3$ does not satisfy assumption (\ref{eq:The Moves of Picard}) and, hence, is an inappropriate stage for the $PQG$.
		\item	$D_4$ contains the reflection $F$. However, the orbit $D_{4} \star B$ is $\{ \ket{0}, \ket{1} \}$. This means that Q can only flip the coin from heads to tails or vice versa. If $M_Q = D_4$, then the $PQG$ degenerates to the classical coin tossing game. Q is unable to employ a truly quantum strategy, something that contradicts the second assumption at the beginning of Section \ref{sec:Enlarging the Operational Space of the PQG} and goes against the spirit of the $PQG$. Moreover, in $D_4$ Q no longer possesses a winning strategy. For these reasons, it is meaningless to play the $PQG$ in $D_4$.
		\item	The dihedral groups $D_5, D_6$ and $D_7$ do not contain the reflection $F$ either. Once again the formulas (\ref{eq:The F Reflector}) and (\ref{eq:Standard Representation r^k s in D_n}) for $n = 5, 6$ and $7$ and $k = 0, 1, \dots, n - 1$ can be used to verify this fact. These groups do not satisfy assumption (\ref{eq:The Moves of Picard}) and are also inadmissible for the $PQG$.
	\end{itemize}
}
\end{proof}


\begin{ManualProposition}{5.4}[$D_{n}$ does not contain $F$ when $n$ odd] \label{prp:Absence of F in D_n for n Odd Appendix}
	If $n$ is odd, then the dihedral group $D_{n}$ does not contain $F$.
\end{ManualProposition}
\begin{proof}[Proof of Proposition \ref{prp:Absence of F in D_n for n Odd Appendix}]
{\small
	Let us recall the formulas (\ref{eq:Standard Representation r^k in D_n}) and (\ref{eq:Standard Representation r^k s in D_n}). For convenience, we repeat them below, noting that they are valid for every $n \geq 3$ and every $k, 0 \leq k \leq n - 1$.
	\begin{multicols}{2}
		\noindent
		\begin{align} \label{eq:Standard Representation r^k in D_n Appendix Proposition}
			r^{k} \mapsto R_{\frac{2 \pi k}{n}} =
			\begin{bmatrix}
				\begin{array}{lr}
					\cos \frac{2 \pi k}{n} & -\sin \frac{2 \pi k}{n} \\
					\sin \frac{2 \pi k}{n} & \cos \frac{2 \pi k}{n}
				\end{array}
			\end{bmatrix}
			\tag{\ref{eq:Standard Representation r^k in D_n}}
		\end{align}
		\begin{align} \label{eq:Standard Representation r^k s in D_n Appendix Proposition}
			r^{k} s \mapsto S_{\frac{\pi k}{n}} =
			\begin{bmatrix}
				\begin{array}{lr}
					\cos \frac{2 \pi k}{n} & \sin \frac{2 \pi k}{n} \\
					\sin \frac{2 \pi k}{n} & -\cos \frac{2 \pi k}{n}
				\end{array}
			\end{bmatrix}
			\tag{\ref{eq:Standard Representation r^k s in D_n}}
		\end{align}
	\end{multicols}
	Let us assume to the contrary that there is an odd $n$ such that $D_{n}$ does contain $F$. Then there must be a $k$, $0 \leq k < n$, such that
	\begin{align} \label{eq:Absence of F in $D_{n}$ for n Odd Appendix}
		F
		=
		\begin{bmatrix}
			0 & 1
			\\
			1 & 0
		\end{bmatrix}
		=
		\pm
		\begin{bmatrix}
			\begin{array}{lr}
				\cos \frac{2 \pi k}{n} & \sin \frac{2 \pi k}{n} \\
				\sin \frac{2 \pi k}{n} & -\cos \frac{2 \pi k}{n}
			\end{array}
		\end{bmatrix}
		\Rightarrow
		\left \{
		\begin{matrix*}[l]
			\cos \frac{2 \pi k}{n} = 0 \\
			\sin \frac{2 \pi k}{n} = 1
		\end{matrix*}
		\right \}
		&\text{ or }
		\left \{
		\begin{matrix*}[l]
			\cos \frac{2 \pi k}{n} = 0 \\
			\sin \frac{2 \pi k}{n} = -1
		\end{matrix*}
		\right \}
		\ .
		\tag{ \ref{prp:Absence of F in D_n for n Odd Appendix}.i }
	\end{align}
	The fact that $0 \leq k < n$, implies that $0 \leq \frac{2 \pi k}{n} < 2 \pi$. Hence, either $\frac{2 \pi k}{n} = \frac{\pi}{2}$ or $\frac{2 \pi k}{n} = \frac{3 \pi}{2}$. The former equation leads to $k = \frac{n}{4}$ and the latter to $k = \frac{3 n}{4}$, which are both impossible because $n$ is odd. Thus, we have arrived at a contradiction, which proves that $F$ does not exist in $D_{n}$ when $n$ is odd.
}
\end{proof}

For completeness of the exposition, we remind the reader of some very familiar notions, that will be invoked in our forthcoming proofs.

\begin{definition}[The unit circle] \label{def:The Unit Circle}
	The circle $S^{1}$ of unit radius centered at the origin, which will be henceforth called the \emph{unit circle}, is defined as
	\begin{align} \label{eq:The Unit Circle}
		S^{1} = \{ (x, y) \in \mathbb{R}^{2} \ : \ x^{2} + y^{2} = 1 \} \ .
	\end{align}
	The \emph{upper} semicircle $S^{1}_{y \geq 0}$ of the unit circle is
	\begin{align} \label{eq:The Upper semicircle}
		S^{1}_{y \geq 0} = \{ (x, y) \in S^{1} \ : \ y \geq 0 \} \ .
	\end{align}
	Symmetrically, the \emph{lower} semicircle $S^{1}_{y \leq 0}$ of the unit circle is
	\begin{align} \label{eq:The Lower semicircle}
		S^{1}_{y \leq 0} = \{ (x, y) \in S^{1} \ : \ y \leq 0 \} \ .
	\end{align}
	Given a point $\mathbf{x} = \begin{bmatrix} x \\ y \end{bmatrix} \in S^{1}$, its \emph{antipodal} point is $-\mathbf{x} = \begin{bmatrix} -x \\ -y \end{bmatrix} \in S^{1}$.
\end{definition}

\begin{figure}[H]
	\begin{minipage}[t]{0.32\textwidth}
		\centering
		\begin{tikzpicture}[scale = 1.7]
			\def \angle {360/8}
			\draw (-1.5, 0) -- (1.5, 0);
			\draw (0, -1.5) -- (0, 1.5);
			\draw [fill, thick, WordBlueDark]
				({cos(0 * \angle)}, {sin(0 * \angle)}) circle (0.75 pt) node [below right] {1};
			\draw [fill, thick, WordBlueDark]
				({cos(1 * \angle)}, {sin(1 * \angle)}) circle (0.75 pt) node [above right] {A};
			\draw [fill, thick, WordBlueDark]
				({cos(2 * \angle)}, {sin(2 * \angle)}) circle (0.75 pt) node [above left] {1};
			\draw [fill, thick, WordBlueDark]
				({cos(3 * \angle)}, {sin(3 * \angle)}) circle (0.75 pt) node [above left] {B};
			\draw [fill, thick, WordBlueDark]
				({cos(4 * \angle)}, {sin(4 * \angle)}) circle (0.75 pt) node [below left] {-1};
			\draw [fill, thick, WordBlueDark]
				({cos(5 * \angle)}, {sin(5 * \angle)}) circle (0.75 pt) node [below left] {$A^\prime$};
			\draw [fill, thick, WordBlueDark]
				({cos(6 * \angle)}, {sin(6 * \angle)}) circle (0.75 pt) node [below left] {-1};
			\draw [fill, thick, WordBlueDark]
				({cos(7 * \angle)}, {sin(7 * \angle)}) circle (0.75 pt) node [below right] {$B^\prime$};
			\draw [thin, dashed] ({cos(1 * \angle)}, {sin(1 * \angle)}) -- ({cos(5 * \angle)}, {sin(5 * \angle)});
			\draw [thin, dashed] ({cos(3 * \angle)}, {sin(3 * \angle)}) -- ({cos(7 * \angle)}, {sin(7 * \angle)});
			\draw [thick, WordBlueDark] (0, 0) circle [ radius = 1 cm ];
		\end{tikzpicture}
		\caption{The unit circle $S^{1}$ and the antipodal pairs $A, A^\prime$ and $B, B^\prime$.}
		\label{fig:The Unit Circle S^1}
	\end{minipage}
	\hfill
	\begin{minipage}[t]{0.32\textwidth}
		\centering
		\begin{tikzpicture}[scale = 1.7]
			\def \angle {360/8}
			\draw (-1.5, 0) -- (1.5, 0);
			\draw (0, -1.5) -- (0, 1.5);
			\draw [fill, thick, RedPurple]
				({cos(0 * \angle)}, {sin(0 * \angle)}) circle (0.75 pt) node [below right] {1};
			\draw [fill, thick, RedPurple]
				({cos(1 * \angle)}, {sin(1 * \angle)}) circle (0.75 pt) node [above right] {A};
			\draw [fill, thick, RedPurple]
				({cos(2 * \angle)}, {sin(2 * \angle)}) circle (0.75 pt) node [above left] {1};
			\draw [fill, thick, RedPurple]
				({cos(3 * \angle)}, {sin(3 * \angle)}) circle (0.75 pt) node [above left] {B};
			\draw [fill, thick, RedPurple]
				({cos(4 * \angle)}, {sin(4 * \angle)}) circle (0.75 pt) node [below left] {-1};
			\draw [thick, RedPurple] (1cm, 0cm) arc [start angle = 0, end angle = 180, radius = 1cm];
			\draw [thin, dashed] (1cm, 0cm) arc [start angle = 0, end angle = -180, radius = 1cm];
			\draw [thin, dashed, RedPurple] (0, 0) -- ({cos(1 * \angle)}, {sin(1 * \angle)});
			\scoped [on background layer]
				\filldraw [->, MyLightRed, line width = 0.3 mm] (0, 0) -- (0.5, 0) arc (0:\angle:0.5);
			\draw [->, RedPurple, line width = 0.3 mm] (0.5, 0) arc (0:\angle:0.5);
			\draw [->, RedPurple, line width = 0.3 mm] ({1.2 * cos(0.5 * \angle)}, {1.2 * sin(0.5 * \angle)}) node [RedPurple, right] {$\frac{2 \pi}{n}$} -- ({0.3 * cos(0.5 * \angle)}, {0.3 * sin(0.5 * \angle)});
			\draw [thin, dashed, RedPurple] (0, 0) -- ({cos(2 * \angle)}, {sin(2 * \angle)});
			\draw [->, RedPurple, line width = 0.3 mm] (0.6, 0) arc (0:2 * \angle:0.6);
			\draw [->, RedPurple, line width = 0.3 mm] ({1.2 * cos(1.5 * \angle)}, {1.2 * sin(1.5 * \angle)}) node [RedPurple, above] {$\frac{4 \pi}{n}$} -- ({0.3 * cos(1.5 * \angle)}, {0.3 * sin(1.5 * \angle)});
			\draw [thin, dashed, RedPurple] (0, 0) -- ({cos(3 * \angle)}, {sin(3 * \angle)});
			\draw [->, RedPurple, line width = 0.3 mm] (0.7, 0) arc (0:3 * \angle:0.7);
			\draw [->, RedPurple, line width = 0.3 mm] ({1.2 * cos(2.5 * \angle)}, {1.2 * sin(2.5 * \angle)}) node [RedPurple, above] {$\frac{6 \pi}{n}$} -- ({0.3 * cos(2.5 * \angle)}, {0.3 * sin(2.5 * \angle)});
		\end{tikzpicture}
		\caption{The upper semicircle of the unit circle, its end points and some intermediate points.}
		\label{fig:The Upper Semicircle of the Unit Circle}
	\end{minipage}
	\hfill
	\begin{minipage}[t]{0.32\textwidth}
		\centering
		\begin{tikzpicture}[scale = 1.7]
			\def \angle {360/8}
			\draw (-1.5, 0) -- (1.5, 0);
			\draw (0, -1.5) -- (0, 1.5);
			\draw [fill, thick, GreenLighter2]
			({cos(0 * \angle)}, {sin(0 * \angle)}) circle (0.75 pt) node [below right] {1};
			\draw [fill, thick, GreenLighter2]
			({cos(4 * \angle)}, {sin(4 * \angle)}) circle (0.75 pt) node [below left] {-1};
			\draw [fill, thick, GreenLighter2]
			({cos(5 * \angle)}, {sin(5 * \angle)}) circle (0.75 pt) node [below left] {$A^\prime$};
			\draw [fill, thick, GreenLighter2]
			({cos(6 * \angle)}, {sin(6 * \angle)}) circle (0.75 pt) node [below left] {-1};
			\draw [fill, thick, GreenLighter2]
			({cos(7 * \angle)}, {sin(7 * \angle)}) circle (0.75 pt) node [below right] {$B^\prime$};
			\draw [thick, GreenLighter2] (1cm, 0cm) arc [start angle = 0, end angle = -180, radius = 1cm];
			\draw [thin, dashed] (1cm, 0cm) arc [start angle = 0, end angle = 180, radius = 1cm];
		\end{tikzpicture}
		\caption{The lower semicircle of the unit circle, its end points and some intermediate points.}
		\label{fig:The Lower Semicircle of the Unit Circle}
	\end{minipage}
\end{figure}

It will also be helpful to recall some well-known trigonometric identities (see \cite{Beecher2016}):

{\small
\begin{align}
	\cos ( \theta + \frac{\pi}{2} ) &= - \sin \theta
	&\sin ( \theta + \frac{\pi}{2} ) &= \cos \theta \label{eq:Cos + Sin Theta + Pi/2 Identity}
	\\
	\cos ( \theta + \pi ) &= - \cos \theta
	&\sin ( \theta + \pi ) &= - \sin \theta \label{eq:Cos + Sin Theta + Pi Identity}
	\\
	\sin \theta + \sin \varphi &= 2 \sin ( \frac{ \theta + \varphi }{2} ) \cos ( \frac{ \theta - \varphi }{2} )
	&\sin \theta - \sin \varphi &= 2 \cos ( \frac{ \theta + \varphi }{2} ) \sin ( \frac{ \theta - \varphi }{2} ) \label{eq:Sin Theta +- Sin Varphi Identity}
	\\
	\cos ( \theta + \varphi ) &= \cos \theta \cos \varphi - \sin \theta \sin \varphi
	&\cos ( \theta - \varphi ) &= \cos \theta \cos \varphi + \sin \theta \sin \varphi \label{eq:Cos Addition Subtraction Identity}
	\\
	\sin ( \theta + \varphi ) &= \sin \theta \cos \varphi + \cos \theta \sin \varphi
	&\sin ( \theta - \varphi ) &= \sin \theta \cos \varphi - \cos \theta \sin \varphi \label{eq:Sin Addition Subtraction Identity}
\end{align}
}

\begin{lemma} \label{lem:Criterion for Orbit Coincidence} \
	\begin{enumerate}
		\item	If $\ket{1} \in G \star \ket{0}$, then $G \star \ket{0} = G \star \ket{1}$, where $G$ is any group of linear operators.
		\item	If $\ket{0} \in G \star \ket{1}$, then $G \star \ket{0} = G \star \ket{1}$, where $G$ is any group of linear operators.
	\end{enumerate}
\end{lemma}
\begin{proof}[Proof of Lemma \ref{lem:Criterion for Orbit Coincidence}] \
{\small
	\begin{enumerate}
		\item	By Definition \ref{def:Orbits & Stabilizers}, we know that if $\ket{1} \in G \star \ket{0}$, then there exists an element $g_2 \in G$ such that $\ket{1} = g_2 \star \ket{0}$ (\ref{lem:Criterion for Orbit Coincidence}.i). Every member of $\ket{x_1} \in G \star \ket{1}$ has the form $\ket{x_1} = g_1 \star \ket{1}$ (\ref{lem:Criterion for Orbit Coincidence}.ii) for some $g_1 \in G$. If we combine Definition \ref{def:Group Action} with (\ref{lem:Criterion for Orbit Coincidence}.i) and (\ref{lem:Criterion for Orbit Coincidence}.ii), we deduce that $\ket{x_1} = \left( g_1 g_2 \right) \star \ket{0}$, that is $\ket{x_1} \in G \star \ket{0}$ too. Hence, $G \star \ket{1} \subset G \star \ket{0}$ (\ref{lem:Criterion for Orbit Coincidence}.iii).

		At the same time, Definition \ref{def:Group Action} together with (\ref{lem:Criterion for Orbit Coincidence}.i), imply that $\ket{0} = g^{-1}_{2} \star \ket{1}$ (\ref{lem:Criterion for Orbit Coincidence}.iv). Every member of $\ket{x_2} \in G \star \ket{0}$ has the form $\ket{x_2} = g_3 \star \ket{0}$ (\ref{lem:Criterion for Orbit Coincidence}.v) for some $g_3 \in G$. If we combine Definition \ref{def:Group Action} with (\ref{lem:Criterion for Orbit Coincidence}.iv) and (\ref{lem:Criterion for Orbit Coincidence}.v), we deduce that $\ket{x_2} = \left( g_3 g^{-1}_{2} \right) \star \ket{1}$, that is $\ket{x_2} \in G \star \ket{1}$ too. Therefore, $G \star \ket{0} \subset G \star \ket{1}$ (\ref{lem:Criterion for Orbit Coincidence}.vi). Together (\ref{lem:Criterion for Orbit Coincidence}.iii) and (\ref{lem:Criterion for Orbit Coincidence}.vi) establish that $G \star \ket{0} = G \star \ket{1}$.
		\item	The proof is completely symmetrical.
	\end{enumerate}
}
\end{proof}

\begin{lemma} \label{lem:The Action of D_n on B}
	The action of $D_n$ on the basis kets $\ket{0}$ and $\ket{1}$ gives rise to the following two sequences of kets $\ket{\varphi_k}$ and $\ket{\chi_k}$, where $0 \leq k \leq n - 1$:
	\begin{multicols}{2}
		\noindent
		\begin{align} \label{eq:Orbit of ket 0 in D_n Appendix}
			\ket{\varphi_k}
			=
			\begin{bmatrix}
				\cos \frac{2 \pi k}{n}
				\\
				\sin \frac{2 \pi k}{n}
			\end{bmatrix}
		\end{align}
		\begin{align} \label{eq:Orbit of ket 1 in D_n Appendix}
			\ket{\chi_k}
			=
			\begin{bmatrix*}[r]
				-\sin \frac{2 \pi k}{n}
				\\
				\cos \frac{2 \pi k}{n}
			\end{bmatrix*}
		\end{align}
	\end{multicols}
\end{lemma}
\begin{proof}[Proof of Lemma \ref{lem:The Action of D_n on B}]
{\small
	Let us first recall the formulas (\ref{eq:Standard Representation r^k in D_n}) and (\ref{eq:Standard Representation r^k s in D_n}). For convenience, we repeat them below, noting that they are valid for every $n \geq 3$ and every $k, 0 \leq k \leq n - 1$.
	\begin{multicols}{2}
		\noindent
		\begin{align} \label{eq:Standard Representation r^k in D_n Appendix}
			r^{k} \mapsto R_{\frac{2 \pi k}{n}} =
			\begin{bmatrix}
				\begin{array}{lr}
					\cos \frac{2 \pi k}{n} & -\sin \frac{2 \pi k}{n} \\
					\sin \frac{2 \pi k}{n} & \cos \frac{2 \pi k}{n}
				\end{array}
			\end{bmatrix}
			\tag{\ref{eq:Standard Representation r^k in D_n}}
		\end{align}
		\begin{align} \label{eq:Standard Representation r^k s in D_n Appendix}
			r^{k} s \mapsto S_{\frac{\pi k}{n}} =
			\begin{bmatrix}
				\begin{array}{lr}
					\cos \frac{2 \pi k}{n} & \sin \frac{2 \pi k}{n} \\
					\sin \frac{2 \pi k}{n} & -\cos \frac{2 \pi k}{n}
				\end{array}
			\end{bmatrix}
			\tag{\ref{eq:Standard Representation r^k s in D_n}}
		\end{align}
	\end{multicols}
	The action of the standard matrix representation of $D_n$ on the computational basis $B$ is given by the matrix-vector multiplication of the matrices (\ref{eq:Standard Representation r^k in D_n Appendix}) and (\ref{eq:Standard Representation r^k s in D_n Appendix}) with the kets $\ket{0} = \begin{bmatrix} 1 \\ 0 \end{bmatrix}$ and $\ket{1} = \begin{bmatrix} 0 \\ 1 \end{bmatrix}$. The resulting products are the following.
	\begin{multicols}{2}
		\noindent
		\begin{align} \label{eq:Rotational Orbit of ket 0 in D_n Appendix}
			\begin{bmatrix}
				\begin{array}{lr}
					\cos \frac{2 \pi k}{n} & -\sin \frac{2 \pi k}{n} \\
					\sin \frac{2 \pi k}{n} & \cos \frac{2 \pi k}{n}
				\end{array}
			\end{bmatrix}
			\begin{bmatrix}
				1
				\\
				0
			\end{bmatrix}
			=
			\begin{bmatrix}
				\cos \frac{2 \pi k}{n}
				\\
				\sin \frac{2 \pi k}{n}
			\end{bmatrix}
			\tag{ \ref{lem:The Action of D_n on B}.i }
		\end{align}
		\begin{align} \label{eq:Reflection Orbit of ket 0 in D_n Appendix}
			\begin{bmatrix}
				\begin{array}{lr}
					\cos \frac{2 \pi k}{n} & \sin \frac{2 \pi k}{n} \\
					\sin \frac{2 \pi k}{n} & -\cos \frac{2 \pi k}{n}
				\end{array}
			\end{bmatrix}
			\begin{bmatrix}
				1
				\\
				0
			\end{bmatrix}
			=
			\begin{bmatrix}
				\cos \frac{2 \pi k}{n}
				\\
				\sin \frac{2 \pi k}{n}
			\end{bmatrix}
			\tag{ \ref{lem:The Action of D_n on B}.ii }
		\end{align}
	\end{multicols}
	\begin{multicols}{2}
		\noindent
		\begin{align} \label{eq:Rotational Orbit of ket 1 in D_n Appendix}
			\begin{bmatrix}
				\begin{array}{lr}
					\cos \frac{2 \pi k}{n} & -\sin \frac{2 \pi k}{n} \\
					\sin \frac{2 \pi k}{n} & \cos \frac{2 \pi k}{n}
				\end{array}
			\end{bmatrix}
			\begin{bmatrix}
				0
				\\
				1
			\end{bmatrix}
			=
			\begin{bmatrix}
				-\sin \frac{2 \pi k}{n}
				\\
				\cos \frac{2 \pi k}{n}
			\end{bmatrix}
			\tag{ \ref{lem:The Action of D_n on B}.iii }
		\end{align}
		\begin{align} \label{eq:Reflection Orbit of ket 1 in D_n Appendix}
			\begin{bmatrix}
				\begin{array}{lr}
					\cos \frac{2 \pi k}{n} & \sin \frac{2 \pi k}{n} \\
					\sin \frac{2 \pi k}{n} & -\cos \frac{2 \pi k}{n}
				\end{array}
			\end{bmatrix}
			\begin{bmatrix}
				0
				\\
				1
			\end{bmatrix}
			=
			\begin{bmatrix}
				\sin \frac{2 \pi k}{n}
				\\
				-\cos \frac{2 \pi k}{n}
			\end{bmatrix}
			\tag{ \ref{lem:The Action of D_n on B}.iv }
		\end{align}
	\end{multicols}
	By comparing (\ref{eq:Rotational Orbit of ket 0 in D_n Appendix}) and (\ref{eq:Reflection Orbit of ket 0 in D_n Appendix}) we see that the action of the rotations and the reflections of $D_n$ on the basis ket $\ket{0}$ gives rise to precisely the same kets, specifically those that have the form
	\begin{align}
		\begin{bmatrix}
			\cos \frac{2 \pi k}{n}
			\\
			\sin \frac{2 \pi k}{n}
		\end{bmatrix}
		\ , \ 0 \leq k \leq n - 1 \ . \tag{\ref{eq:Orbit of ket 0 in D_n Appendix}}
	\end{align}
	Symmetrically, (\ref{eq:Rotational Orbit of ket 1 in D_n Appendix}) and (\ref{eq:Reflection Orbit of ket 1 in D_n Appendix}), together with the fact that $\ket{\psi}$ and $- \ket{\psi}$ stand for the same state, reveal that the action of the rotations and the reflections of $D_n$ on the basis ket $\ket{1}$ leads to the same kets, namely those shown below.
	\begin{align}
		\begin{bmatrix*}[r]
			-\sin \frac{2 \pi k}{n}
			\\
			\cos \frac{2 \pi k}{n}
		\end{bmatrix*}
		\ , \ 0 \leq k \leq n - 1 \ . \tag{\ref{eq:Orbit of ket 1 in D_n Appendix}}
	\end{align}
}
\end{proof}

\begin{lemma}[The action of $D_{n}$ on $B$ when $n = 4 m$] \label{lem:The Action of D_n on B for n 4-Multiple Appendix}
	If $n \geq 3$ is a multiple of $4$, then the action of the dihedral group $D_{n}$ on the computational basis $B$ is
	\begin{align} \label{eq:Orbit of B in D_n for n 4-Multiple Appendix Lemma}
		D_{n} \star \ket{0} = D_{n} \star \ket{1} = D_{n} \star B =
		\{ \cos \frac{2 \pi k}{n} \ket{0} + \sin \frac{2 \pi k}{n} \ket{1} : 0 \leq k < \frac{n}{2} \} \ .
	\end{align}
\end{lemma}
\begin{proof}[Proof of Lemma \ref{lem:The Action of D_n on B for n 4-Multiple Appendix}]
{\small
	In this case we assume that
	\begin{align} \label{eq:Special Case for n 4-Multiple}
		n = 4 m \ , \ m \geq 1 \ . \tag{ \ref{lem:The Action of D_n on B for n 4-Multiple Appendix}.i }
	\end{align}
	Consequently, (\ref{eq:Orbit of ket 0 in D_n Appendix}) and (\ref{eq:Orbit of ket 1 in D_n Appendix}) become:
	\begin{align} \label{eq:Orbit of B in D_n for n 4-Multiple Extended Form Appendix}
		\ket{\varphi_k}
		=
		\begin{bmatrix}
			\cos \frac{\pi k}{2 m}
			\\
			\sin \frac{\pi k}{2 m}
		\end{bmatrix}
		\quad {\rm and} \quad
		\ket{\chi_k}
		=
		\begin{bmatrix*}[r]
			- \sin \frac{\pi k}{2 m}
			\\
			\cos \frac{\pi k}{2 m}
		\end{bmatrix*}
		\ , \ 0 \leq k \leq 4m - 1 \ . \tag{ \ref{lem:The Action of D_n on B for n 4-Multiple Appendix}.ii }
	\end{align}
	These kets are not all different. To understand this let us first observe that
	\begin{align}
		\ket{\chi_{k}}
		\overset
		{ ( \ref{eq:Orbit of B in D_n for n 4-Multiple Extended Form Appendix} ) }
		{ = }
		\begin{bmatrix*}[r]
			- \sin \frac{\pi k}{2 m}
			\\
			\cos \frac{\pi k}{2 m}
		\end{bmatrix*}
		\overset
		{ ( \ref{eq:Cos + Sin Theta + Pi/2 Identity} ) }
		{ = }
		\begin{bmatrix}
			\cos ( \frac{\pi k}{2 m} + \frac{\pi}{2} )
			\\
			\sin ( \frac{\pi k}{2 m} + \frac{\pi}{2} )
		\end{bmatrix}
		=
		\begin{bmatrix}
			\cos ( \frac{\pi (k + m)}{2 m} )
			\\
			\sin ( \frac{\pi (k + m)}{2 m} )
		\end{bmatrix}
		\ .
		\tag{ \ref{lem:The Action of D_n on B for n 4-Multiple Appendix}.iii }
	\end{align}
	When $k$ ranges from $0$ to $3 m - 1$, equations (\ref{lem:The Action of D_n on B for n 4-Multiple Appendix}.ii) and (\ref{lem:The Action of D_n on B for n 4-Multiple Appendix}.iii) immediately give that
	\begin{align}
		\ket{\chi_{k}}
		=
		\ket{\varphi_{ k + m } }
		\ , \ 0 \leq k \leq 3 m - 1 \ .
		\tag{ \ref{lem:The Action of D_n on B for n 4-Multiple Appendix}.iv }
	\end{align}
	It remains to ascertain what happens when $k$ ranges from $3 m$ to $4 m - 1$. Then, $k + m$ ranges from $4 m$ to $4 m + ( m - 1 )$ and, according to (\ref{lem:The Action of D_n on B for n 4-Multiple Appendix}.iii), the kets $\ket{\chi_k}$ assume the values
	\begin{align}
		\begin{bmatrix}
			\cos ( 2 \pi + \frac{\pi 0}{2 m} )
			\\
			\sin ( 2 \pi + \frac{\pi 0}{2 m} )
		\end{bmatrix}
		=
		\begin{bmatrix}
			\cos \frac{\pi 0}{2 m}
			\\
			\sin \frac{\pi 0}{2 m}
		\end{bmatrix}
		\overset
		{ ( \ref{eq:Orbit of B in D_n for n 4-Multiple Extended Form Appendix} ) }
		{ = }
		\ket{\varphi_{ 0 } },
		\dots,
		\begin{bmatrix}
			\cos ( 2 \pi + \frac{\pi (m - 1)}{2 m} )
			\\
			\sin ( 2 \pi + \frac{\pi (m - 1)}{2 m} )
		\end{bmatrix}
		=
		\begin{bmatrix}
			\cos \frac{\pi (m - 1)}{2 m}
			\\
			\sin \frac{\pi (m - 1)}{2 m}
		\end{bmatrix}
		\overset
		{ ( \ref{eq:Orbit of B in D_n for n 4-Multiple Extended Form Appendix} ) }
		{ = }
		\ket{\varphi_{ m - 1 } }
		\ .
		\tag{ \ref{lem:The Action of D_n on B for n 4-Multiple Appendix}.v }
	\end{align}
	If we combine equations (\ref{lem:The Action of D_n on B for n 4-Multiple Appendix}.iv) and (\ref{lem:The Action of D_n on B for n 4-Multiple Appendix}.v) we derive that 
	\begin{align} \label{eq:Relation I Among kets for n 4-Multiple Appendix}
		\ket{\chi_{k}} = \ket{\varphi_{ (k + m) \bmod n} }
		\ , \ 0 \leq k \leq 4m - 1 \ ,
		\tag{ \ref{lem:The Action of D_n on B for n 4-Multiple Appendix}.vi }
	\end{align}
	which shows that all the kets of the $\ket{\chi_k}$ sequence also appear in the $\ket{\varphi_k}$ sequence.

	Another way to arrive at this conclusion is to observe that the $D_n$-orbits of $\ket{0}$ and $\ket{1}$ consist of kets appearing in the sequences $\ket{\varphi_k}$ and $\ket{\chi_k}$, respectively. In view of the fact that $\ket{1} = \ket{\chi_0}$ appears in the $\ket{\varphi_k}$ sequence as $\ket{\varphi_m}$, Lemma \ref{lem:Criterion for Orbit Coincidence} asserts that $D_{n} \star \ket{0} = D_{n} \star \ket{1} = D_{n} \star B$.

	Furthermore, it also happens that only $2 m$ of the kets in the $\ket{\varphi_k}$ sequence are distinct ($\ket{\psi}$ and $- \ket{\psi}$ represent the same state). In particular, it holds that
	\begin{align} \label{eq:Relation II Among kets for n 4-Multiple Appendix}
		\ket{\varphi_k} = - \ket{\varphi_{ k + 2 m }}
		\ , \ 0 \leq k \leq 2 m - 1 \ ,
		\tag{ \ref{lem:The Action of D_n on B for n 4-Multiple Appendix}.vii }
	\end{align}
	that is kets $\ket{\varphi_k}$ and $\ket{\varphi_{ k + 2 m }}$ correspond to antipodal points in the unit circle (Figure \ref{fig:Orbit of B in D_n for n 4-Multiple} gives a geometric depiction of the situation). The latter is easily proved as follows:
	\begin{align}
		\ket{\varphi_k}
		\overset
		{ ( \ref{eq:Orbit of B in D_n for n 4-Multiple Extended Form Appendix} ) }
		{ = }
		\begin{bmatrix}
			\cos \frac{\pi k}{2 m}
			\\
			\sin \frac{\pi k}{2 m}
		\end{bmatrix}
		\overset
		{ ( \ref{eq:Cos + Sin Theta + Pi Identity} ) }
		{ = }
		\begin{bmatrix}
			- \cos ( \frac{\pi k}{2 m} + \pi )
			\\
			- \sin ( \frac{\pi k}{2 m} + \pi )
		\end{bmatrix}
		=
		\begin{bmatrix}
			- \cos \frac{\pi k + 2 \pi m}{2 m}
			\\
			- \sin \frac{\pi k + 2 \pi m}{2 m}
		\end{bmatrix}
		\overset
		{ ( \ref{eq:Orbit of B in D_n for n 4-Multiple Extended Form Appendix} ) }
		{ = }
		- \ket{\varphi_{ k + 2 m }}
		\ , \ 0 \leq k \leq 2 m - 1 \ .
		\tag{ \ref{lem:The Action of D_n on B for n 4-Multiple Appendix}.viii }
	\end{align}
	When $k$ ranges from $0$ to $2 m - 1$, formula (\ref{eq:Orbit of B in D_n for n 4-Multiple Extended Form Appendix}) gives the  first $2 m$ kets in the $\ket{\varphi_k}$ sequence
	\begin{align} \label{eq:Complete Orbit of ket 0 in D_n for n 4-Multiple Appendix}
		\begin{bmatrix}
			1
			\\
			0
		\end{bmatrix},
		\begin{bmatrix}
			\cos \frac{\pi}{2 m}
			\\
			\sin \frac{\pi}{2 m}
		\end{bmatrix},
		\dots,
		\begin{bmatrix}
			\cos \frac{\pi (2 m - 1)}{2 m}
			\\
			\sin \frac{\pi (2 m - 1)}{2 m}
		\end{bmatrix} \ .
		\tag{ \ref{lem:The Action of D_n on B for n 4-Multiple Appendix}.ix }
	\end{align}
	These are all distinct because each one of them corresponds to a unique different point that lies on the upper semicircle of the unit circle and makes an angle $\frac{\pi k}{2 m}$, where $0 \leq k \leq 2 m - 1$, with the positive $x$-axis, as shown in Figure \ref{fig:The Upper Semicircle of the Unit Circle}. Finally, by noting that $2 m - 1 < 2m = \frac{n}{2}$, we verify that (\ref{eq:Orbit of B in D_n for n 4-Multiple Appendix Lemma}) holds.
}
\end{proof}

\begin{lemma}[The action of $D_{n}$ on $B$ when $n = 2 m$] \label{lem:The Action of D_n on B for n Even Appendix}
	If $n \geq 3$ is even, but not a multiple of $4$, then the action of the dihedral group $D_{n}$ on the computational basis $B$ is
	\begin{align}
		D_{n} \star \ket{0}
		&=
		\{ \cos \frac{2 \pi k}{n} \ket{0} + \sin \frac{2 \pi k}{n} \ket{1} : 0 \leq k < \frac{n}{2} \} \ , \label{eq:Orbit of Ket 0 in D_n for n Even but Not 4-Multiple Appendix Lemma}
		\\
		D_{n} \star \ket{1}
		&=
		\{ -\sin \frac{2 \pi k}{n} \ket{0} + \cos \frac{2 \pi k}{n} \ket{1} : 0 \leq k < \frac{n}{2} \} \ , \label{eq:Orbit of Ket 1 in D_n for n Even but Not 4-Multiple Appendix Lemma}
		\\
		D_{n} \star B
		&= \{ \cos \frac{2 \pi k}{n} \ket{0} + \sin \frac{2 \pi k}{n} \ket{1} : 0 \leq k < \frac{n}{2} \}
		\cup \{ -\sin \frac{2 \pi k}{n} \ket{0} + \cos \frac{2 \pi k}{n} \ket{1} : 0 \leq k < \frac{n}{2} \} \ . \label{eq:Orbit of B in D_n for n Even but Not 4-Multiple Appendix Lemma}
	\end{align}
\end{lemma}
\begin{proof}[Proof of Lemma \ref{lem:The Action of D_n on B for n Even Appendix}]
{\small
	In this case we know that
	\begin{align} \label{eq:Special Case for n Even but Not 4-Multiple}
		n = 2 m \ , \ \text{where} \ m \ \text{is odd and} \ m \geq 3 \ . \tag{ \ref{lem:The Action of D_n on B for n Even Appendix}.i }
	\end{align}
	As a result now (\ref{eq:Orbit of ket 0 in D_n Appendix}) and (\ref{eq:Orbit of ket 1 in D_n Appendix}) give:
	\begin{align} \label{eq:Orbit of B in D_n for n Even but Not 4-Multiple Extended Form Appendix}
		\ket{\varphi_k}
		=
		\begin{bmatrix}
			\cos \frac{\pi k}{m}
			\\
			\sin \frac{\pi k}{m}
		\end{bmatrix}
		\quad {\rm and} \quad
		\ket{\chi_k}
		=
		\begin{bmatrix*}[r]
			- \sin \frac{\pi k}{m}
			\\
			\cos \frac{\pi k}{m}
		\end{bmatrix*}
		\ , \ 0 \leq k \leq 2m - 1 \ \text{ and } m \text{ odd}. \tag{ \ref{lem:The Action of D_n on B for n Even Appendix}.ii }
	\end{align}
	Once again we encounter the phenomenon that the kets in the above sequences are not all different. Only $m$ of the kets in the $\ket{\varphi_k}$ sequence and only $m$ of the kets in the $\ket{\chi_k}$ sequence are distinct (as always, we keep in mind that $\ket{\psi}$ and $- \ket{\psi}$ represent the same state). In particular, it holds that
	\begin{align} \label{eq:Relation I Among kets for n Even Appendix}
		\ket{\varphi_k} = - \ket{\varphi_{ k + m }}
		\ , \ 0 \leq k \leq m - 1 \ ,
		\tag{ \ref{lem:The Action of D_n on B for n Even Appendix}.iii }
	\end{align}
	that is kets $\ket{\varphi_k}$ and $\ket{\varphi_{ k + m }}$ correspond to antipodal points in the unit circle (Figure \ref{fig:Orbit of B in D_n for n Even} gives a geometric depiction of the situation). This can be shown as follows:
	\begin{align}
		\ket{\varphi_k}
		\overset
		{ ( \ref{eq:Orbit of B in D_n for n Even but Not 4-Multiple Extended Form Appendix} ) }
		{ = }
		\begin{bmatrix}
			\cos \frac{\pi k}{m}
			\\
			\sin \frac{\pi k}{m}
		\end{bmatrix}
		\overset
		{ ( \ref{eq:Cos + Sin Theta + Pi Identity} ) }
		{ = }
		\begin{bmatrix}
			- \cos ( \frac{\pi k}{m} + \pi )
			\\
			- \sin ( \frac{\pi k}{m} + \pi )
		\end{bmatrix}
		=
		\begin{bmatrix}
			- \cos \frac{\pi k + \pi m}{m}
			\\
			- \sin \frac{\pi k + \pi m}{m}
		\end{bmatrix}
		\overset
		{ ( \ref{eq:Orbit of B in D_n for n Even but Not 4-Multiple Extended Form Appendix} ) }
		{ = }
		- \ket{\varphi_{ k + m }}
		\ , \ 0 \leq k \leq m - 1 \ .
		\tag{ \ref{lem:The Action of D_n on B for n Even Appendix}.iv }
	\end{align}
	When $k$ ranges from $0$ to $m - 1$, formula (\ref{eq:Orbit of B in D_n for n Even but Not 4-Multiple Extended Form Appendix}) gives the first $m$ kets in the $\ket{\varphi_k}$ sequence
	\begin{align} \label{eq:Complete Orbit of ket 0 in D_n for n Even Appendix}
		\begin{bmatrix}
			1
			\\
			0
		\end{bmatrix},
		\begin{bmatrix}
			\cos \frac{\pi}{m}
			\\
			\sin \frac{\pi}{m}
		\end{bmatrix},
		\dots,
		\begin{bmatrix}
			\cos \frac{\pi (m - 1)}{m}
			\\
			\sin \frac{\pi (m - 1)}{m}
		\end{bmatrix} \ .
		\tag{ \ref{lem:The Action of D_n on B for n Even Appendix}.v }
	\end{align}
	The above kets correspond to the $m$ points $p_0, p_1, \dots, p_{m - 1}$ that lie on the upper semicircle of the unit circle and make angles $0 < \frac{\pi}{m} < \frac{2 \pi}{m} < \dots < \frac{\pi (m - 1)}{m}$, respectively, with the positive $x$-axis, as shown in Figure \ref{fig:The Upper Semicircle of the Unit Circle}. The associated angles lie in the interval $[0, \pi)$ because $\frac{\pi (m - 1)}{m} < \pi$ and, therefore, the points $p_0, p_1, \dots, p_{m - 1}$ are all distinct. Finally, by noting that $m - 1 < m = \frac{n}{2}$, we verify that (\ref{eq:Orbit of Ket 0 in D_n for n Even but Not 4-Multiple Appendix Lemma}) holds.

	Analogously, it also holds that
	\begin{align} \label{eq:Relation II Among kets for n Even Appendix}
		\ket{\chi_k} = - \ket{\chi_{ k + m }}
		\ , \ 0 \leq k \leq m - 1 \ ,
		\tag{ \ref{lem:The Action of D_n on B for n Even Appendix}.vi }
	\end{align}
	that is kets $\ket{\chi_k}$ and $\ket{\chi_{ k + m }}$ too correspond to antipodal points in the unit circle (again consult Figure \ref{fig:Orbit of B in D_n for n Even}). This is also shown as follows:
	\begin{align}
		\ket{\chi_k}
		\overset
		{ ( \ref{eq:Orbit of B in D_n for n Even but Not 4-Multiple Extended Form Appendix} ) }
		{ = }
		\begin{bmatrix*}[r]
			- \sin \frac{\pi k}{m}
			\\
			\cos \frac{\pi k}{m}
		\end{bmatrix*}
		\overset
		{ ( \ref{eq:Cos + Sin Theta + Pi Identity} ) }
		{ = }
		\begin{bmatrix*}[r]
			\sin ( \frac{\pi k}{m} + \pi )
			\\
			- \cos ( \frac{\pi k}{m} + \pi )
		\end{bmatrix*}
		=
		\begin{bmatrix*}[r]
			\sin \frac{\pi k + \pi m}{m}
			\\
			- \cos \frac{\pi k + \pi m}{m}
		\end{bmatrix*}
		\overset
		{ ( \ref{eq:Orbit of B in D_n for n Even but Not 4-Multiple Extended Form Appendix} ) }
		{ = }
		- \ket{\chi_{ k + m }}
		\ , \ 0 \leq k \leq m - 1 \ .
		\tag{ \ref{lem:The Action of D_n on B for n Even Appendix}.vii }
	\end{align}
	When $k$ ranges from $0$ to $m - 1$, formula (\ref{eq:Orbit of B in D_n for n Even but Not 4-Multiple Extended Form Appendix}) gives the first $m$ kets in the $\ket{\chi_k}$ sequence
	\begin{align} \label{eq:Complete Orbit of ket 1 in D_n for n Even Appendix}
		\begin{bmatrix}
			0
			\\
			1
		\end{bmatrix},
		\begin{bmatrix*}[r]
			- \sin \frac{\pi}{m}
			\\
			\cos \frac{\pi}{m}
		\end{bmatrix*},
		\dots,
		\begin{bmatrix*}[r]
			- \sin \frac{\pi (m - 1)}{m}
			\\
			\cos \frac{\pi (m - 1)}{m}
		\end{bmatrix*} \ .
		\tag{ \ref{lem:The Action of D_n on B for n Even Appendix}.viii }
	\end{align}
	These kets correspond to the $m$ points $q_0, q_1, \dots, q_{m - 1}$ that lie on the unit circle and make angles $\frac{\pi}{2} < \frac{\pi}{2} + \frac{\pi}{m} < \frac{\pi}{2} + \frac{2 \pi}{m} < \dots < \frac{\pi}{2} + \frac{\pi (m - 1)}{m}$, respectively, with the positive $x$-axis. The associated angles lie in the interval $[\frac{\pi}{2}, \frac{\pi}{2} + \pi)$ because $\frac{\pi (m - 1)}{m} < \pi$, i.e., the points $q_0, q_1, \dots, q_{m - 1}$ are all distinct. Taking into account that $m - 1 < m = \frac{n}{2}$, we have established that (\ref{eq:Orbit of Ket 1 in D_n for n Even but Not 4-Multiple Appendix Lemma}) holds.

	The important observation in this case is that
	\begin{itemize}
		\item	no ket (or its opposite) from the sequence (\ref{eq:Complete Orbit of ket 0 in D_n for n Even Appendix}) appears in the sequence (\ref{eq:Complete Orbit of ket 1 in D_n for n Even Appendix}), that is $\ket{\chi_{k_2}} \neq \pm \ket{\varphi_{k_1}}$, $\forall k_1, k_2$, where $0 \leq k_1, k_2 \leq m - 1$, and
		\item	no ket (or its opposite) from the sequence (\ref{eq:Complete Orbit of ket 1 in D_n for n Even Appendix}) appears in the sequence (\ref{eq:Complete Orbit of ket 0 in D_n for n Even Appendix}), i.e., $\ket{\varphi_{k_1}} \neq \pm \ket{\chi_{k_2}}$, $\forall k_1, k_2$, where $0 \leq k_1, k_2 \leq m - 1$.
	\end{itemize}
	To verify these claims, let us suppose to the contrary that there exist $k_1, k_2$, where $0 \leq k_1, k_2 \leq m - 1$, such that
	\begin{align} \label{eq:No Coincidence Among kets for n Even Appendix}
		\begin{bmatrix}
			\cos \frac{\pi k_1}{m}
			\\
			\sin \frac{\pi k_1}{m}
		\end{bmatrix}
		=
		\pm
		\begin{bmatrix*}[r]
			- \sin \frac{\pi k_2}{m}
			\\
			\cos \frac{\pi k_2}{m}
		\end{bmatrix*}
		\Rightarrow
		\left \{
		\begin{matrix*}[l]
			\cos \frac{\pi k_1}{m} = - \sin \frac{\pi k_2}{m} \\
			\sin \frac{\pi k_1}{m} =   \cos \frac{\pi k_2}{m}
		\end{matrix*}
		\right \}
		\text{ or }
		\left \{
		\begin{matrix*}[l]
			\cos \frac{\pi k_1}{m} = \sin \frac{\pi k_2}{m} \\
			\sin \frac{\pi k_1}{m} = - \cos \frac{\pi k_2}{m}
		\end{matrix*}
		\right \}
		\tag{ \ref{lem:The Action of D_n on B for n Even Appendix}.ix }
		\ .
	\end{align}
	The above suppositions inescapably lead to the following sequence of implications.
	\begin{align} \label{eq:No Coincidence Proof for n Even Appendix}
		\left \{
		\begin{matrix}
			\cos \frac{\pi k_1}{m} + \sin \frac{\pi k_2}{m} = 0 \\
			\sin \frac{\pi k_1}{m} - \cos \frac{\pi k_2}{m} = 0
		\end{matrix}
		\right \}
		&\text{ or }
		\left \{
		\begin{matrix}
			\cos \frac{\pi k_1}{m} - \sin \frac{\pi k_2}{m} = 0 \\
			\sin \frac{\pi k_1}{m} + \cos \frac{\pi k_2}{m} = 0
		\end{matrix}
		\right \}
		\overset
		{ ( \ref{eq:Cos + Sin Theta + Pi/2 Identity} ) }
		{ \Rightarrow }
		\nonumber \\
		\left \{
		\begin{matrix}
			\sin ( \frac{\pi k_1}{m} + \frac{\pi}{2} ) + \sin \frac{\pi k_2}{m} = 0 \\
			\sin \frac{\pi k_1}{m} - \sin ( \frac{\pi k_2}{m} + \frac{\pi}{2} ) = 0
		\end{matrix}
		\right \}
		&\text{ or }
		\left \{
		\begin{matrix}
			\sin ( \frac{\pi k_1}{m} + \frac{\pi}{2} ) - \sin \frac{\pi k_2}{m} = 0 \\
			\sin \frac{\pi k_1}{m} + \sin ( \frac{\pi k_2}{m} + \frac{\pi}{2} ) = 0
		\end{matrix}
		\right \}
		\overset
		{ ( \ref{eq:Sin Theta +- Sin Varphi Identity} ) }
		{ \Rightarrow }
		\nonumber \\
		\left \{
		\begin{matrix}
			2 \sin \left( \frac{\pi (k_1 + k_2)}{2 m} + \frac{\pi}{4} \right) \cos ( \frac{\pi (k_1 - k_2)}{2 m} + \frac{\pi}{4} ) = 0 \\
			2 \cos \left( \frac{\pi (k_1 + k_2)}{2 m} + \frac{\pi}{4} \right) \sin ( \frac{\pi (k_1 - k_2)}{2 m} - \frac{\pi}{4} ) = 0
		\end{matrix}
		\right \}
		&\text{ or }
		\left \{
		\begin{matrix}
			2 \cos \left( \frac{\pi (k_1 + k_2)}{2 m} + \frac{\pi}{4} \right) \sin ( \frac{\pi (k_1 - k_2)}{2 m} + \frac{\pi}{4} ) = 0 \\
			2 \sin \left( \frac{\pi (k_1 + k_2)}{2 m} + \frac{\pi}{4} \right) \cos ( \frac{\pi (k_1 - k_2)}{2 m} - \frac{\pi}{4} ) = 0
		\end{matrix}
		\right \}
		\nonumber
		\tag{ \ref{lem:The Action of D_n on B for n Even Appendix}.x }
	\end{align}
	To proceed further it is convenient to distinguish the following cases.
	\begin{itemize}
		\item	The first case gives the system
				$
				\left \{
				\begin{matrix}
					\sin \left( \frac{\pi (k_1 + k_2)}{2 m} + \frac{\pi}{4} \right) = 0 \\
					\cos \left( \frac{\pi (k_1 + k_2)}{2 m} + \frac{\pi}{4} \right) = 0
				\end{matrix}
				\right \}
				$
				, which is clearly impossible because there is no $\varphi$ such that $\sin \varphi = \cos \varphi = 0$.
		\item	The next case involves the system
				$
				\left \{
				\begin{matrix}
					\sin \left( \frac{\pi (k_1 + k_2)}{2 m} + \frac{\pi}{4} \right) = 0 \\
					\sin ( \frac{\pi (k_1 - k_2)}{2 m} - \frac{\pi}{4} ) = 0
				\end{matrix}
				\right \}
				$
				. Using ( \ref{eq:Sin Addition Subtraction Identity} ), this system can be transformed to the equivalent
				$
				\left \{
				\begin{matrix}
					\sin \left( \frac{\pi (k_1 + k_2)}{2 m} \right) \cos \frac{\pi}{4} +
					\cos \left( \frac{\pi (k_1 + k_2)}{2 m} \right) \sin \frac{\pi}{4} = 0 \\
					\sin \left( \frac{\pi (k_1 - k_2)}{2 m} \right) \cos \frac{\pi}{4} -
					\cos \left( \frac{\pi (k_1 - k_2)}{2 m} \right) \sin \frac{\pi}{4} = 0
				\end{matrix}
				\right \}
				$
				, which, in turn, implies that
				$
				\left \{
				\begin{matrix}
					\tan \left( \frac{\pi (k_1 + k_2)}{2 m} \right) = - 1 \\
					\tan \left( \frac{\pi (k_1 - k_2)}{2 m} \right) = 1
				\end{matrix}
				\right \}
				$.
				The fact that $0 \leq k_1, k_2 \leq m - 1$, implies that $0 \leq \frac{\pi (k_1 + k_2)}{2 m} < \pi$ and $- \frac{\pi}{2} < - \frac{\pi (m - 1)}{2 m} \leq \frac{\pi (k_1 - k_2)}{2 m} \leq \frac{\pi (m - 1)}{2 m} < \frac{\pi}{2}$. Hence, we derive that $\frac{\pi (k_1 + k_2)}{2 m} = \frac{3 \pi}{4}$ and $\frac{\pi (k_1 - k_2)}{2 m} = \frac{\pi}{4}$. By adding the last two equations, we conclude that $\frac{2 \pi k_1}{2 m} = \pi \Rightarrow k_1 = m$, which is also impossible because we know that $k_1 \leq m - 1$.
		\item	The next system
				$
				\left \{
				\begin{matrix}
					\cos ( \frac{\pi (k_1 - k_2)}{2 m} + \frac{\pi}{4} ) = 0 \\
					\cos \left( \frac{\pi (k_1 + k_2)}{2 m} + \frac{\pi}{4} \right) = 0
				\end{matrix}
				\right \}
				$
				can be conveniently transformed via (\ref{eq:Cos Addition Subtraction Identity}) to the equivalent system
				$
				\left \{
				\begin{matrix}
					\cos \left( \frac{\pi (k_1 - k_2)}{2 m} \right) \cos \frac{\pi}{4} -
					\sin \left( \frac{\pi (k_1 - k_2)}{2 m} \right) \sin \frac{\pi}{4} = 0 \\
					\cos \left( \frac{\pi (k_1 + k_2)}{2 m} \right) \cos \frac{\pi}{4} -
					\sin \left( \frac{\pi (k_1 + k_2)}{2 m} \right) \sin \frac{\pi}{4} = 0
				\end{matrix}
				\right \}
				$
				, which implies that
				$
				\left \{
				\begin{matrix}
					\tan \left( \frac{\pi (k_1 - k_2)}{2 m} \right) = 1 \\
					\tan \left( \frac{\pi (k_1 + k_2)}{2 m} \right) = 1
				\end{matrix}
				\right \}
				$.
				The fact that $0 \leq k_1, k_2 \leq m - 1$, implies that $0 \leq \frac{\pi (k_1 + k_2)}{2 m} < \pi$ and $- \frac{\pi}{2} < - \frac{\pi (m - 1)}{2 m} \leq \frac{\pi (k_1 - k_2)}{2 m} \leq \frac{\pi (m - 1)}{2 m} < \frac{\pi}{2}$. Hence, we derive that $\frac{\pi (k_1 + k_2)}{2 m} = \frac{\pi (k_1 - k_2)}{2 m} = \frac{\pi}{4}$. By adding the last two equations, we conclude that $\frac{2 \pi k_1}{2 m} = \frac{\pi}{2} \Rightarrow k_1 = \frac{m}{2}$, which is also impossible because we know from (\ref{eq:Special Case for n Even but Not 4-Multiple}) that $m$ is odd.
		\item	The next system
				$
				\left \{
				\begin{matrix}
					\cos ( \frac{\pi (k_1 - k_2)}{2 m} + \frac{\pi}{4} ) = 0 \\
					\sin ( \frac{\pi (k_1 - k_2)}{2 m} - \frac{\pi}{4} ) = 0
				\end{matrix}
				\right \}
				$
				can be rewritten via (\ref{eq:Cos + Sin Theta + Pi/2 Identity}) as
				$
				\left \{
				\begin{matrix*}[r]
					\cos ( \frac{\pi (k_1 - k_2)}{2 m} + \frac{\pi}{4} ) = 0 \\
					- \cos ( \frac{\pi (k_1 - k_2)}{2 m} + \frac{\pi}{4} ) = 0
				\end{matrix*}
				\right \}
				$,
				i.e., $\cos ( \frac{\pi (k_1 - k_2)}{2 m} + \frac{\pi}{4} ) = 0$. The fact that $0 \leq k_1, k_2 \leq m - 1$, implies that $- \frac{\pi}{4} < - \frac{\pi (m - 1)}{2 m} + \frac{\pi}{4} \leq \frac{\pi (k_1 - k_2)}{2 m} + \frac{\pi}{4} \leq \frac{\pi (m - 1)}{2 m} + \frac{\pi}{4} < \frac{3 \pi}{4}$. Thus, $\frac{\pi (k_1 - k_2)}{2 m} + \frac{\pi}{4} = \frac{\pi}{2} \Rightarrow \frac{\pi (k_1 - k_2)}{2 m} = \frac{\pi}{4} \Rightarrow k_1 - k_2 = \frac{m}{2}$. This is absurd because $k_1 - k_2$ is an integer and $m$ is odd, as we recall from (\ref{eq:Special Case for n Even but Not 4-Multiple}).
		\item	In the next case we encounter the system
				$
				\left \{
				\begin{matrix}
					\cos \left( \frac{\pi (k_1 + k_2)}{2 m} + \frac{\pi}{4} \right) = 0 \\
					\sin \left( \frac{\pi (k_1 + k_2)}{2 m} + \frac{\pi}{4} \right) = 0
				\end{matrix}
				\right \}
				$
				, which is clearly impossible because there is no $\varphi$ such that $\sin \varphi = \cos \varphi = 0$.
		\item	The next case concerns the system
				$
				\left \{
				\begin{matrix}
					\cos \left( \frac{\pi (k_1 + k_2)}{2 m} + \frac{\pi}{4} \right) = 0 \\
					\cos ( \frac{\pi (k_1 - k_2)}{2 m} - \frac{\pi}{4} ) = 0
				\end{matrix}
				\right \}
				$
				that can be transformed via (\ref{eq:Cos Addition Subtraction Identity}) to the equivalent system
				$
				\left \{
				\begin{matrix}
					\cos \left( \frac{\pi (k_1 + k_2)}{2 m} \right) \cos \frac{\pi}{4} -
					\sin \left( \frac{\pi (k_1 + k_2)}{2 m} \right) \sin \frac{\pi}{4} = 0 \\
					\cos \left( \frac{\pi (k_1 - k_2)}{2 m} \right) \cos \frac{\pi}{4} +
					\sin \left( \frac{\pi (k_1 - k_2)}{2 m} \right) \sin \frac{\pi}{4} = 0
				\end{matrix}
				\right \}
				$
				, which gives
				$
				\left \{
				\begin{matrix*}[r]
					\tan \left( \frac{\pi (k_1 + k_2)}{2 m} \right) = 1 \\
					\tan \left( \frac{\pi (k_1 - k_2)}{2 m} \right) = - 1
				\end{matrix*}
				\right \}
				$.
				The fact that $0 \leq k_1, k_2 \leq m - 1$, implies that $0 \leq \frac{\pi (k_1 + k_2)}{2 m} < \pi$ and $- \frac{\pi}{2} < - \frac{\pi (m - 1)}{2 m} \leq \frac{\pi (k_1 - k_2)}{2 m} \leq \frac{\pi (m - 1)}{2 m} < \frac{\pi}{2}$. Therefore, we derive that $\frac{\pi (k_1 + k_2)}{2 m} = \frac{\pi}{4}$ and $\frac{\pi (k_1 - k_2)}{2 m} = - \frac{\pi}{4}$. By subtracting the latter from the former, we derive that $\frac{2 \pi k_2}{2 m} = \frac{\pi}{2} \Rightarrow k_2 = \frac{m}{2}$, which is also impossible because we know from (\ref{eq:Special Case for n Even but Not 4-Multiple}) that $m$ is odd.
		\item	Moving to the next case, we have to deal with the system
				$
				\left \{
				\begin{matrix}
					\sin ( \frac{\pi (k_1 - k_2)}{2 m} + \frac{\pi}{4} ) = 0 \\
					\sin \left( \frac{\pi (k_1 + k_2)}{2 m} + \frac{\pi}{4} \right) = 0
				\end{matrix}
				\right \}
				$
				. Using ( \ref{eq:Sin Addition Subtraction Identity} ), this system can be transformed to the equivalent
				$
				\left \{
				\begin{matrix}
					\sin \left( \frac{\pi (k_1 - k_2)}{2 m} \right) \cos \frac{\pi}{4} +
					\cos \left( \frac{\pi (k_1 - k_2)}{2 m} \right) \sin \frac{\pi}{4} = 0 \\
					\sin \left( \frac{\pi (k_1 + k_2)}{2 m} \right) \cos \frac{\pi}{4} +
					\cos \left( \frac{\pi (k_1 + k_2)}{2 m} \right) \sin \frac{\pi}{4} = 0
				\end{matrix}
				\right \}
				$
				, which, in turn, implies that
				$
				\left \{
				\begin{matrix*}[r]
					\tan \left( \frac{\pi (k_1 - k_2)}{2 m} \right) = - 1 \\
					\tan \left( \frac{\pi (k_1 + k_2)}{2 m} \right) = - 1
				\end{matrix*}
				\right \}
				$.
				The fact that $0 \leq k_1, k_2 \leq m - 1$, implies that $0 \leq \frac{\pi (k_1 + k_2)}{2 m} < \pi$ and $- \frac{\pi}{2} < - \frac{\pi (m - 1)}{2 m} \leq \frac{\pi (k_1 - k_2)}{2 m} \leq \frac{\pi (m - 1)}{2 m} < \frac{\pi}{2}$. Thus, we derive that $\frac{\pi (k_1 + k_2)}{2 m} = \frac{3 \pi}{4}$ and $\frac{\pi (k_1 - k_2)}{2 m} = - \frac{\pi}{4}$. By adding the last two equations, we conclude that $\frac{2 \pi k_1}{2 m} = \frac{\pi}{2} \Rightarrow k_1 = \frac{m}{2}$, which is of course impossible, since we know from (\ref{eq:Special Case for n Even but Not 4-Multiple}) that $m$ is odd.
		\item	Finally, we come to the last case concerning the system
				$
				\left \{
				\begin{matrix}
					\sin ( \frac{\pi (k_1 - k_2)}{2 m} + \frac{\pi}{4} ) = 0 \\
					\cos ( \frac{\pi (k_1 - k_2)}{2 m} - \frac{\pi}{4} ) = 0
				\end{matrix}
				\right \}
				$.
				This system can be rewritten using (\ref{eq:Cos + Sin Theta + Pi/2 Identity}) as
				$
				\left \{
				\begin{matrix*}[r]
					\sin ( \frac{\pi (k_1 - k_2)}{2 m} + \frac{\pi}{4} ) = 0 \\
					\sin ( \frac{\pi (k_1 - k_2)}{2 m} + \frac{\pi}{4} ) = 0
				\end{matrix*}
				\right \}
				$,
				i.e., $\sin ( \frac{\pi (k_1 - k_2)}{2 m} + \frac{\pi}{4} ) = 0$. The fact that $0 \leq k_1, k_2 \leq m - 1$, implies that $- \frac{\pi}{4} < - \frac{\pi (m - 1)}{2 m} + \frac{\pi}{4} \leq \frac{\pi (k_1 - k_2)}{2 m} + \frac{\pi}{4} \leq \frac{\pi (m - 1)}{2 m} + \frac{\pi}{4} < \frac{3 \pi}{4}$. Hence, $\frac{\pi (k_1 - k_2)}{2 m} + \frac{\pi}{4} = 0 \Rightarrow \frac{\pi (k_1 - k_2)}{2 m} = - \frac{\pi}{4} \Rightarrow k_1 - k_2 = - \frac{m}{2}$. This is also absurd because $k_1 - k_2$ is an integer and $m$ is odd, as we recall from (\ref{eq:Special Case for n Even but Not 4-Multiple}).
	\end{itemize}
	Thus, we have shown that the first $m$ kets in the $\ket{\varphi_k}$ sequence are all different from the first $m$ kets in the $\ket{\chi_k}$ sequence, which establishes the validity of (\ref{eq:Orbit of B in D_n for n Even but Not 4-Multiple Appendix Lemma}).
}
\end{proof}

By combining the results of Lemmata \ref{lem:The Action of D_n on B for n 4-Multiple Appendix} and \ref{lem:The Action of D_n on B for n Even Appendix} we can immediately prove Theorem \ref{thr:The Action of D_n on B Appendix Theorem}.

\begin{ManualTheorem}{5.5}[The action of $D_{n}$ on $B$] \label{thr:The Action of D_n on B Appendix Theorem}
	The action of the general dihedral group $D_{n}, n \geq 3,$ on the computational basis $B$ depends on whether $n$ is a multiple of $4$ or $n$ is even but not a multiple of $4$. Specifically,
	\begin{enumerate}
		\item	if $n$ is a multiple of $4$, then the action of the dihedral group $D_{n}$ on the computational basis $B$ is
		\begin{align} \label{eq:Orbit of B in D_n for n 4-Multiple Appendix}
			D_{n} \star \ket{0} = D_{n} \star \ket{1} = D_{n} \star B =
			\{ \cos \frac{2 \pi k}{n} \ket{0} + \sin \frac{2 \pi k}{n} \ket{1} : 0 \leq k < \frac{n}{2} \} \ ,
		\end{align}
		\item	if $n$ is even but not a multiple of $4$, then the action of the dihedral group $D_{n}$ on the computational basis $B$ is
		\begin{align}
			D_{n} \star \ket{0}
			&=
			\{ \cos \frac{2 \pi k}{n} \ket{0} + \sin \frac{2 \pi k}{n} \ket{1} : 0 \leq k < \frac{n}{2} \} \ , \label{eq:Orbit of Ket 0 in D_n for n Even but Not 4-Multiple Appendix}
			\\
			D_{n} \star \ket{1}
			&=
			\{ -\sin \frac{2 \pi k}{n} \ket{0} + \cos \frac{2 \pi k}{n} \ket{1} : 0 \leq k < \frac{n}{2} \} \ , \label{eq:Orbit of Ket 1 in D_n for n Even but Not 4-Multiple Appendix}
			\\
			D_{n} \star B
			&= \{ \cos \frac{2 \pi k}{n} \ket{0} + \sin \frac{2 \pi k}{n} \ket{1} : 0 \leq k < \frac{n}{2} \}
			\cup \{ -\sin \frac{2 \pi k}{n} \ket{0} + \cos \frac{2 \pi k}{n} \ket{1} : 0 \leq k < \frac{n}{2} \} \ . \label{eq:Orbit of B in D_n for n Even but Not 4-Multiple Appendix}
		\end{align}
	\end{enumerate}
\end{ManualTheorem}

We may now give the proof of Theorem \ref{thr:The Fixed Set of M_P in D_n Appendix}.

\begin{ManualTheorem}{5.6}[The fixed set of $\{ I, F \}$ in $D_n$] \label{thr:The Fixed Set of M_P in D_n Appendix}
	The fixed set of $M_P = \{ I, F \}$ in the general dihedral group $D_{n}, n \geq 3,$ depends on whether $n$ is a multiple of $8$ or not.
	\begin{enumerate}
		\item	If $n$ is a multiple of $8$, then:
		\begin{align} \label{eq:The Fixed Set of F in D_n for n 8-Multiple Appendix}
			Fix ( \{ I, F \} ) = Fix ( F ) = \{ \ket{+}, \ket{-} \} \ .
		\end{align}
		\item	In every other case:
		\begin{align} \label{eq:The Fixed Set of F in D_n for n Not 8-Multiple Appendix}
			Fix ( \{ I, F \} ) = Fix ( F ) = \emptyset \ .
		\end{align}
	\end{enumerate}
\end{ManualTheorem}
\begin{proof}[Proof of Theorem \ref{thr:The Fixed Set of M_P in D_n Appendix}]
{\small
	\
	\begin{enumerate}
		\item	Let us first consider the case where $n$ is a multiple of $8$:
				\begin{align} \label{eq:Special Case for n 8-Multiple}
					n = 8 m \ , \ m \geq 1 \ . \tag{ \ref{thr:The Fixed Set of M_P in D_n Appendix}.i } 
				\end{align}
				Consequently, (\ref{eq:Orbit of ket 0 in D_n Appendix}) and (\ref{eq:Orbit of ket 1 in D_n Appendix}) become:
				\begin{align} \label{eq:Orbit of B in D_n for n 8-Multiple Extended Form Appendix}
					\ket{\varphi_k}
					=
					\begin{bmatrix}
						\cos \frac{\pi k}{4 m}
						\\
						\sin \frac{\pi k}{4 m}
					\end{bmatrix}
					\quad {\rm and} \quad
					\ket{\chi_k}
					=
					\begin{bmatrix*}[r]
						- \sin \frac{\pi k}{4 m}
						\\
						\cos \frac{\pi k}{4 m}
					\end{bmatrix*}
					\ . \tag{ \ref{thr:The Fixed Set of M_P in D_n Appendix}.ii }
				\end{align}
				According to (\ref{eq:Orbit of B in D_n for n 4-Multiple Appendix Lemma}) the range of $k$ is $0 \leq k < 4 m$. By setting $k = m$ in (\ref{eq:Orbit of B in D_n for n 4-Multiple Appendix Lemma}) we derive that $\cos \frac{2 \pi m}{8 m} \ket{0} + \sin \frac{2 \pi m}{8 m} \ket{1} = \cos \frac{\pi}{4} \ket{0} + \sin \frac{\pi}{4} \ket{1} = \ket{+}$ belongs to $D_{n} \star B$. Likewise, by setting $k = 3 m$ in (\ref{eq:Orbit of B in D_n for n 4-Multiple Appendix Lemma}) we get that $\cos \frac{2 \pi 3 m}{8 m} \ket{0} + \sin \frac{2 \pi 3 m}{8 m} \ket{1} = \cos \frac{3 \pi}{4} \ket{0} + \sin \frac{3 \pi}{4} \ket{1} = - \ket{-}$ belongs to $D_{n} \star B$. The above calculations show that the states $\ket{+}$ and $\ket{-}$ belong to the orbit of $B$. We already know that $F$ fixes these kets (recall Proposition \ref{prp:The Fixed Set of M_P in D_8 Appendix}). What remains is to prove that $F$ fixes no other state in the orbit of $B$. So, let us suppose to the contrary that $F$ also fixes some ket other than $\ket{+}$ and $\ket{-}$. This means that there exists a $k$, $0 \leq k < 4 m$ but $k \neq m, 3 m$, such that
				\begin{align} \label{eq:F Fixes No Other Ket for n 8-Multiple Appendix}
					F
					\begin{bmatrix}
						\cos \frac{\pi k}{4 m}
						\\
						\sin \frac{\pi k}{4 m}
					\end{bmatrix}
					=
					\begin{bmatrix}
						\sin \frac{\pi k}{4 m}
						\\
						\cos \frac{\pi k}{4 m}
					\end{bmatrix}
					=
					\pm
					\begin{bmatrix*}[r]
						\cos \frac{\pi k}{4 m}
						\\
						\sin \frac{\pi k}{4 m}
					\end{bmatrix*}
					\Rightarrow
					\left \{
					\begin{matrix*}[l]
						\sin \frac{\pi k}{4 m} = \cos \frac{\pi k}{4 m} \\
						\cos \frac{\pi k}{4 m} = \sin \frac{\pi k}{4 m}
					\end{matrix*}
					\right \}
					&\text{ or }
					\left \{
					\begin{matrix*}[l]
						\sin \frac{\pi k}{4 m} = - \cos \frac{\pi k}{4 m} \\
						\cos \frac{\pi k}{4 m} = - \sin \frac{\pi k}{4 m}
					\end{matrix*}
					\right \}
					\Rightarrow
					\nonumber \\
					\tan \left( \frac{\pi k}{4 m} \right) = 1
					\quad &\text{ or } \quad
					\tan \left( \frac{\pi k}{4 m} \right) = - 1
					\tag{ \ref{thr:The Fixed Set of M_P in D_n Appendix}.iii }
					\ .
				\end{align}
				The fact that $0 \leq k < 4 m$, implies that $0 \leq \frac{\pi k}{4 m} < \pi$. Therefore, either $\frac{\pi k}{4 m} = \frac{\pi}{4}$ or $\frac{\pi k}{4 m} = \frac{3 \pi}{4}$. The former equation leads to $k = m$ and the latter to $k = 3 m$, which correspond to kets $\ket{+}$ and $\ket{-}$, respectively. No other values for $k$ arise and, thus, $F$ fixes no other state.
		\item	If $n$ is not a multiple of $8$, then we may distinguish the following cases.
				\begin{itemize}
					\item	$n$ is a multiple of $4$, but not a multiple of $8$. This implies that $n = 4 m$, where $m$ is a positive \emph{odd} integer. Accordingly, (\ref{eq:Orbit of ket 0 in D_n Appendix}) and (\ref{eq:Orbit of ket 1 in D_n Appendix}) become:
							\begin{align} \label{eq:Orbit of B in D_n for n 4-Multiple Extended Form Appendix Theorem}
								\ket{\varphi_k}
								=
								\begin{bmatrix}
									\cos \frac{\pi k}{2 m}
									\\
									\sin \frac{\pi k}{2 m}
								\end{bmatrix}
								\quad {\rm and} \quad
								\ket{\chi_k}
								=
								\begin{bmatrix*}[r]
									- \sin \frac{\pi k}{2 m}
									\\
									\cos \frac{\pi k}{2 m}
								\end{bmatrix*}
								\ . \tag{ \ref{thr:The Fixed Set of M_P in D_n Appendix}.iv }
							\end{align}
							According to (\ref{eq:Orbit of B in D_n for n 4-Multiple Appendix Lemma}) the range of $k$ is $0 \leq k < 2 m$. Let us first assume that there exists a $k$ such that $\frac{\pi k}{2 m} = \frac{\pi}{4}$. But this is absurd because then $k$ must be equal to $\frac{m}{2}$. Similarly, the existence of a $k$ such that $\frac{\pi k}{2 m} = \frac{3 \pi}{4}$ is impossible because then $k$ must be equal to $\frac{3 m}{2}$. The previous calculations establish that $\ket{+}$ and $\ket{-}$ do not belong to the orbit of $B$. We now show that $F$ fixes no ket in the orbit of $B$. If $F$ did fix some ket, then there would be a $k$, $0 \leq k < 2 m$, such that
							\begin{align} \label{eq:F Fixes No Other Ket for n 4-Multiple Appendix}
								F
								\begin{bmatrix}
									\cos \frac{\pi k}{2 m}
									\\
									\sin \frac{\pi k}{2 m}
								\end{bmatrix}
								=
								\begin{bmatrix}
									\sin \frac{\pi k}{2 m}
									\\
									\cos \frac{\pi k}{2 m}
								\end{bmatrix}
								=
								\pm
								\begin{bmatrix*}[r]
									\cos \frac{\pi k}{2 m}
									\\
									\sin \frac{\pi k}{2 m}
								\end{bmatrix*}
								\Rightarrow
								\left \{
								\begin{matrix*}[l]
									\sin \frac{\pi k}{2 m} = \cos \frac{\pi k}{2 m} \\
									\cos \frac{\pi k}{2 m} = \sin \frac{\pi k}{2 m}
								\end{matrix*}
								\right \}
								&\text{ or }
								\left \{
								\begin{matrix*}[l]
									\sin \frac{\pi k}{2 m} = - \cos \frac{\pi k}{2 m} \\
									\cos \frac{\pi k}{2 m} = - \sin \frac{\pi k}{2 m}
								\end{matrix*}
								\right \}
								\Rightarrow
								\nonumber \\
								\tan \left( \frac{\pi k}{2 m} \right) = 1
								\quad &\text{ or } \quad
								\tan \left( \frac{\pi k}{2 m} \right) = - 1
								\tag{ \ref{thr:The Fixed Set of M_P in D_n Appendix}.v }
								\ .
							\end{align}
							The fact that $0 \leq k < 2 m$, implies that $0 \leq \frac{\pi k}{2 m} < \pi$. Therefore, either $\frac{\pi k}{2 m} = \frac{\pi}{4}$ or $\frac{\pi k}{2 m} = \frac{3 \pi}{4}$. The former equation leads to $k = \frac{m}{2}$ and the latter to $k = \frac{3 m}{2}$, which are both impossible. Hence, $F$ fixes no state from the orbit of $B$.
					\item	$n$ is even, but not a multiple of $4$. This implies that $n = 2 m$, where $m$ is a positive \emph{odd} integer. Then (\ref{eq:Orbit of ket 0 in D_n Appendix}) and (\ref{eq:Orbit of ket 1 in D_n Appendix}) become:
							\begin{align} \label{eq:Orbit of B in D_n for n Even but Not 4-Multiple Extended Form Appendix Theorem}
								\ket{\varphi_k}
								=
								\begin{bmatrix}
									\cos \frac{\pi k}{m}
									\\
									\sin \frac{\pi k}{m}
								\end{bmatrix}
								\quad {\rm and} \quad
								\ket{\chi_k}
								=
								\begin{bmatrix*}[r]
									- \sin \frac{\pi k}{m}
									\\
									\cos \frac{\pi k}{m}
								\end{bmatrix*}
								\ . \tag{ \ref{thr:The Fixed Set of M_P in D_n Appendix}.vi }
							\end{align}
							According to (\ref{eq:Orbit of B in D_n for n Even but Not 4-Multiple Appendix Lemma}) the range of $k$ is $0 \leq k < m$. Let us assume that there exists a $k$ such that $\frac{\pi k}{m} = \frac{\pi}{4}$. But this is absurd because then $k$ must be equal to $\frac{m}{4}$. Similarly, the existence of a $k$ such that $\frac{\pi k}{m} = \frac{3 \pi}{4}$ is impossible because then $k$ must be equal to $\frac{3 m}{4}$, since $m$ is odd. The previous calculations establish that $\ket{+}$ and $\ket{-}$ do not belong to the orbit of $B$. We now show that $F$ fixes no ket in the orbit of $B$. If $F$ did fix some ket, then there would be a $k$, $0 \leq k < m$, such that
							\begin{align} \label{eq:F Fixes No Other Ket for n Even Appendix}
								F
								\begin{bmatrix}
									\cos \frac{\pi k}{m}
									\\
									\sin \frac{\pi k}{m}
								\end{bmatrix}
								=
								\begin{bmatrix}
									\sin \frac{\pi k}{m}
									\\
									\cos \frac{\pi k}{m}
								\end{bmatrix}
								=
								\pm
								\begin{bmatrix*}[r]
									\cos \frac{\pi k}{m}
									\\
									\sin \frac{\pi k}{m}
								\end{bmatrix*}
								\Rightarrow
								\left \{
								\begin{matrix*}[l]
									\sin \frac{\pi k}{m} = \cos \frac{\pi k}{m} \\
									\cos \frac{\pi k}{m} = \sin \frac{\pi k}{m}
								\end{matrix*}
								\right \}
								&\text{ or }
								\left \{
								\begin{matrix*}[l]
									\sin \frac{\pi k}{m} = - \cos \frac{\pi k}{m} \\
									\cos \frac{\pi k}{m} = - \sin \frac{\pi k}{m}
								\end{matrix*}
								\right \}
								\Rightarrow
								\nonumber \\
								\tan \left( \frac{\pi k}{m} \right) = 1
								\quad &\text{ or } \quad
								\tan \left( \frac{\pi k}{m} \right) = - 1
								\tag{ \ref{thr:The Fixed Set of M_P in D_n Appendix}.vii }
								\ .
							\end{align}
							The fact that $0 \leq k < m$, implies that $0 \leq \frac{\pi k}{m} < \pi$. Hence, either $\frac{\pi k}{m} = \frac{\pi}{4}$ or $\frac{\pi k}{m} = \frac{3 \pi}{4}$. The former equation leads to $k = \frac{m}{4}$ and the latter to $k = \frac{3 m}{4}$, which are both impossible because $m$ is odd. Hence, $F$ fixes no state from the orbit of $B$.
				\end{itemize}

	\end{enumerate}
}
\end{proof}

The preceding results allow to easily prove Theorem \ref{thr:Q's PQG Winning Strategies in D_{8 n} Appendix}.

\begin{ManualTheorem}{5.7}[The ambient group of the $PQG$ is $D_{8 n}$] \label{thr:Q's PQG Winning Strategies in D_{8 n} Appendix}
	If $M_P = \{ I, F \}$ and $M_Q = D_{8 n}$, i.e., the ambient group of the $PQG$ is $D_{8 n}$, where $n \geq 1$, then the following hold.
	\begin{enumerate}
		\item	Q has exactly two classes of winning and dominant strategies
		\begin{align} \label{eq:D_{8 n} Winning Strategy Classes Appendix}
			\mathcal{C}_{+} = [ (H, H) ] \quad \text{and} \quad \mathcal{C}_{-} = [ (S_{\frac{7 \pi}{8}}, S_{\frac{7 \pi}{8}}) ] \ ,
		\end{align}
		each containing $16$ equivalent strategies.
		\item	The winning state paths corresponding to $\mathcal{C}_{+}$ and $\mathcal{C}_{-}$ are
		\begin{align} \label{eq:D_{8 n} Winning State Paths Appendix}
			\tau_{\mathcal{C}_{+}} = (\ket{0}, \ket{+}, \ket{0}) \quad \text{and} \quad \tau_{\mathcal{C}_{-}} = (\ket{0}, \ket{-}, \ket{0}) \ .
		\end{align}
		\item	Picard has no winning strategy.
	\end{enumerate}
\end{ManualTheorem}
\begin{proof}[Proof of Theorem \ref{thr:Q's PQG Winning Strategies in D_{8 n} Appendix}]
{\small \
	\begin{enumerate}
		\item	The two classes of winning strategies of Q, $\mathcal{C}_{+}$ and $\mathcal{C}_{-}$, which were establish by Theorem \ref{thr:Q's PQG Winning Strategies in D_8 Appendix}, are also present in every dihedral group $D_{8 n}$.

				Let us first suppose that in some dihedral group larger than $D_{8}$ there exists a third class $\mathcal{C}'$ and consider a strategy $\sigma = (A_{1}, A_{2})$ in this class. Then the action of $A_{1}$ must drive the coin into some state other than $\ket{+}$ or $\ket{-}$. However, (\ref{eq:Characteristic Property II of Winning Strategies Appendix}) asserts that $A_1 \ket{0} \in Fix ( \{ F \} )$, which, in view of (\ref{eq:The Fixed Set of F in D_n for n 8-Multiple Appendix}), implies that $A_1 \ket{0} \in \{ \ket{+}, \ket{-} \}$, a contradiction. Hence, there are just two classes of winning strategies $\mathcal{C}_{+}$ and $\mathcal{C}_{-}$. It remains to prove that $\mathcal{C}_{+}$ and $\mathcal{C}_{-}$ do not contain any new winning strategy. So, let us temporarily suppose that $\sigma = (A_{1}, A_{2})$ is a ``new winning'' strategy, that is other than those established in $D_{8}$. Let us first examine the possibility that $A_{1}$ sends the coin to state $\ket{+}$ and assume that $A_{1}$ is other than $H, R_{\frac{2 \pi}{8}}, S_{\frac{5 \pi}{8}}$ and $R_{\frac{10 \pi}{8}}$. If $A_{1}$ is a rotation, then, according to (\ref{eq:Standard Representation r^k in D_n}), there exist $n \geq 1$ and $k, 0 \leq k < 8 n,$ such that
				$A_{1}
				=
				\begin{bmatrix}
					\begin{array}{lr}
						\cos \frac{2 \pi k}{8 n} & - \sin \frac{2 \pi k}{8 n} \\
						\sin \frac{2 \pi k}{8 n} & \cos \frac{2 \pi k}{8 n}
					\end{array}
				\end{bmatrix}
				$. Therefore,
				\begin{align}
					\begin{bmatrix}
						\begin{array}{lr}
							\cos \frac{2 \pi k}{8 n} & - \sin \frac{2 \pi k}{8 n} \\
							\sin \frac{2 \pi k}{8 n} & \cos \frac{2 \pi k}{8 n}
						\end{array}
					\end{bmatrix}
					\begin{bmatrix}
						\begin{array}{lr}
							1 \\
							0
						\end{array}
					\end{bmatrix}
					=
					\pm
					\frac{1}{\sqrt{2}}
					\begin{bmatrix}
						\begin{array}{lr}
							1 \\
							1
						\end{array}
					\end{bmatrix}
					\Rightarrow
					\begin{bmatrix}
						\begin{array}{lr}
							\cos \frac{2 \pi k}{8 n} \\
							\sin \frac{2 \pi k}{8 n}
						\end{array}
					\end{bmatrix}
					=
					\pm
					\frac{1}{\sqrt{2}}
					\begin{bmatrix}
						\begin{array}{lr}
							1 \\
							1
						\end{array}
					\end{bmatrix}
					\Rightarrow
					\nonumber \\
					\left \{
					\begin{matrix}
						\cos \frac{2 \pi k}{8 n} = \frac{1}{\sqrt{2}} \\
						\sin \frac{2 \pi k}{8 n} = \frac{1}{\sqrt{2}}
					\end{matrix}
					\right \}
					\text{ or }
					\left \{
					\begin{matrix}
						\cos \frac{2 \pi k}{8 n} = - \frac{1}{\sqrt{2}} \\
						\sin \frac{2 \pi k}{8 n} = - \frac{1}{\sqrt{2}}
					\end{matrix}
					\right \}
					\ .
					\nonumber
				\end{align}
				The fact that $0 \leq k < 8 n$, implies that $0 \leq \frac{2 \pi k}{8 n} < 2 \pi$. Hence, either $\frac{2 \pi k}{8 n} = \frac{\pi}{4}$ or $\frac{2 \pi k}{8 n} = \frac{5\pi}{4}$. The former equation leads to $k = n$ and the latter to $k = 5 n$.
				This means that
				$A_{1}
				=
				\begin{bmatrix}
					\begin{array}{lr}
						\cos \frac{2 \pi}{8} & - \sin \frac{2 \pi}{8} \\
						\sin \frac{2 \pi}{8} & \cos \frac{2 \pi}{8}
					\end{array}
				\end{bmatrix}
				=
				R_{\frac{2 \pi}{8}}
				$ or
				$A_{1}
				=
				\begin{bmatrix}
					\begin{array}{lr}
						\cos \frac{10 \pi}{8} & - \sin \frac{10 \pi}{8} \\
						\sin \frac{10 \pi}{8} & \cos \frac{10 \pi}{8}
					\end{array}
				\end{bmatrix}
				=
				R_{\frac{10 \pi}{8}}
				$, which contradicts our assumption that $A_{1}$ is different from $R_{\frac{2 \pi}{8}}$ and $R_{\frac{10 \pi}{8}}$.

				We arrive at similar contradictions if we assume that $A_{1}$ is a reflection different from $H$ or $S_{\frac{5 \pi}{8}}$ or that $A_{1}$ drives the coin to state $\ket{-}$. Thus, we conclude that, other than those already existing in $D_{8}$, there are no more winning strategies for Q in the larger dihedral groups $D_{8 n}$.
		\item	Based on the above analysis it is straightforward to see that (\ref{eq:D_{8 n} Winning State Paths Appendix}) holds.
		\item	By Definition \ref{def:Winning and Dominant Strategies}, Picard has no winning strategy because if Q employs one of his winning strategies, Picard has $0.0$ probability to win the game.
	\end{enumerate}
}
\end{proof}

The next two theorem settle the most general case, where the ambient group is $U(2)$.

\begin{ManualTheorem}{5.8}[The fixed set of $\{ I, F \}$ in $U(2)$] \label{thr:The Fixed Set of M_P in U(2) Appendix}
	Under the action of $U(2)$ on the computational basis $B$, the fixed set of $M_P = \{ I, F \}$ is
	\begin{align} \label{eq:The Fixed Set of F in U(2) Appendix}
		Fix ( \{ I, F \} ) = Fix ( F ) = \{ \ket{+}, \ket{-} \} \ .
	\end{align}
\end{ManualTheorem}
\begin{proof}[Proof of Theorem \ref{thr:The Fixed Set of M_P in U(2) Appendix}]
{\small
	In this most general case we must find the eigenvalues and the eigenkets of the flip operator $F$. We may start by formulating the characteristic equation of $F$:
	\begin{align} \label{eq:F Characteristic Equation}
		\det (F - \lambda I)
		=
		\begin{vmatrix}
			\begin{array}{lr}
				- \lambda & 1 \\
				1 & - \lambda
			\end{array}
		\end{vmatrix}
		=
		\lambda^{2} - 1 = (\lambda + 1) (\lambda - 1)
		\Rightarrow
		\lambda = \pm 1
		\ .
		\tag{ \ref{thr:The Fixed Set of M_P in U(2) Appendix}.i }
	\end{align}
	Hence, the two eigenvalues of $F$ are $\lambda = 1$ and $\lambda = -1$. If
	$
	\begin{bmatrix}
		x_1
		\\
		x_2
	\end{bmatrix}
	$ is an eigenket corresponding to the eigenvalue $\lambda = 1$, then $(F - \lambda I) \ket{\psi} = \mathbf{0}$. So, in the first case where the eigenvalue is $1$ we have
	\begin{align} \label{eq:F Eigenket for lambda = 1}
		\begin{bmatrix}
			\begin{array}{lr}
				- 1 & 1 \\
				1 & - 1
			\end{array}
		\end{bmatrix}
		\begin{bmatrix}
			x_1
			\\
			x_2
		\end{bmatrix}
		=
		\begin{bmatrix}
			0
			\\
			0
		\end{bmatrix}
		\Rightarrow
		\left \{
		\begin{matrix}
			- x_1 + x_2 = 0 \\
			x_1 - x_2 = 0
		\end{matrix}
		\right \}
		\ .
		\tag{ \ref{thr:The Fixed Set of M_P in U(2) Appendix}.ii }
	\end{align}
	The general solution is $x_1 = z$ and $x_2 = z$, where $z \in \mathbb{C}$. In matrix form
	$
	\begin{bmatrix}
		x_1
		\\
		x_2
	\end{bmatrix}
	$
	can be written as
	$
	\begin{bmatrix}
		z
		\\
		z
	\end{bmatrix}
	=
	z
	\begin{bmatrix}
		1
		\\
		1
	\end{bmatrix}
	$.
	The normalization condition is satisfied if we take $z = \frac{1}{\sqrt{2}}$. Thus, the eigenket corresponding to the eigenvalue $1$ is
	$
	\frac{1}{\sqrt{2}}
	\begin{bmatrix}
		1
		\\
		1
	\end{bmatrix}
	=
	\ket{+}
	$,
	which can be viewed as a basis for the eigenspace corresponding to the eigenvalue $1$.

	Symmetrically, when the eigenvalue is $\lambda = -1$ we have
	\begin{align} \label{eq:F Eigenket for lambda = -1}
		\begin{bmatrix}
			\begin{array}{lr}
				1 & 1 \\
				1 & 1
			\end{array}
		\end{bmatrix}
		\begin{bmatrix}
			x_1
			\\
			x_2
		\end{bmatrix}
		=
		\begin{bmatrix}
			0
			\\
			0
		\end{bmatrix}
		\Rightarrow
		\left \{
		\begin{matrix}
			x_1 + x_2 = 0 \\
			x_1 + x_2 = 0
		\end{matrix}
		\right \}
		\ .
		\tag{ \ref{thr:The Fixed Set of M_P in U(2) Appendix}.iii }
	\end{align}
	The general solution is $x_1 = z$ and $x_2 = - z$, where $z \in \mathbb{C}$. In matrix form
	$
	\begin{bmatrix}
		x_1
		\\
		x_2
	\end{bmatrix}
	$
	can be written as
	$
	\begin{bmatrix}
		z
		\\
		- z
	\end{bmatrix}
	=
	z
	\begin{bmatrix}
		1
		\\
		- 1
	\end{bmatrix}
	$.
	Again, the normalization condition is satisfied if we take $z = \frac{1}{\sqrt{2}}$. Hence, the eigenket corresponding to the eigenvalue $1$ is
	$
	\frac{1}{\sqrt{2}}
	\begin{bmatrix}
		1
		\\
		- 1
	\end{bmatrix}
	=
	\ket{-}
	$,
	which can be considered as a basis for the eigenspace corresponding to the eigenvalue $-1$.

	We have, therefore, established that $F$ fixes both $\ket{+}$ and $\ket{-}$, since:
	 \begin{align} \label{eq:The Fixed Set of F in U(2) Extended Form Appendix}
	 	\begin{bmatrix}
	 		0 & 1
	 		\\
	 		1 & 0
	 	\end{bmatrix}
	 	\ket{+}
	 	=
	 	\ket{+}
	 	\quad {\rm and} \quad
	 	\begin{bmatrix}
			0 & 1
			\\
			1 & 0
		\end{bmatrix}
		\ket{-}
		=
		- \ket{-}
	\ . \tag{ \ref{thr:The Fixed Set of M_P in U(2) Appendix}.iv }
	\end{align}
}
\end{proof}

\begin{ManualTheorem}{5.9}[The ambient group of the $PQG$ is $U(2)$] \label{thr:Q's PQG Winning Strategies in U(2) Appendix}
	If $M_P = \{ I, F \}$ and $M_Q = U(2)$, i.e., the ambient group of the $PQG$ is $U(2)$, then the following hold.
	\begin{enumerate}
		\item	Q has exactly two classes of winning and dominant strategies, each containing infinite equivalent strategies:
				\begin{align} \label{eq:U(2) Winning Strategy Classes Appendix}
					\mathcal{C}_{+} = [ ( A_{1} (\theta_{1}), A_{2} (\theta_{2}) ) ] \quad \text{and} \quad \mathcal{C}_{-} = [ ( B_{1} (\theta_{3}), B_{2} (\theta_{4}) ) ] \ ,
				\end{align}
				where
				\begin{itemize}
					\item	$A_{1} (\theta_{1})$ is one of $H (\theta_{1}), R_{\frac{2 \pi}{8}} (\theta_{1}), S_{\frac{5 \pi}{8}} (\theta_{1})$ or $R_{\frac{10 \pi}{8}} (\theta_{1})$,
					\item	$A_{2} (\theta_{2})$ is one of $H (\theta_{2}), R_{\frac{14 \pi}{8}} (\theta_{2}), S_{\frac{5 \pi}{8}} (\theta_{2})$ or $R_{\frac{6 \pi}{8}} (\theta_{2})$,
					\item	$B_{1} (\theta_{3})$ is one of $S_{\frac{7 \pi}{8}} (\theta_{3}), R_{\frac{14 \pi}{8}} (\theta_{3}), S_{\frac{3 \pi}{8}} (\theta_{3})$ or $R_{\frac{6 \pi}{8}} (\theta_{3})$,
					\item	$B_{2} (\theta_{4})$ is one of $S_{\frac{7 \pi}{8}} (\theta_{4}), R_{\frac{2 \pi}{8}} (\theta_{4}), S_{\frac{3 \pi}{8}} (\theta_{4})$ or $R_{\frac{10 \pi}{8}} (\theta_{4})$, and
					\item	$\theta_{1}, \theta_{2}, \theta_{3}, \theta_{4}$ are possibly different real parameters.
				\end{itemize}
		\item	The winning state paths corresponding to $\mathcal{C}_{+}$ and $\mathcal{C}_{-}$ are
		\begin{align} \label{eq:U(2) Winning State Paths Appendix}
			\tau_{\mathcal{C}_{+}} = (\ket{0}, \ket{+}, \ket{0}) \quad \text{and} \quad \tau_{\mathcal{C}_{-}} = (\ket{0}, \ket{-}, \ket{0}) \ .
		\end{align}
		\item	Picard has no winning strategy.
	\end{enumerate}
\end{ManualTheorem}
\begin{proof}[Proof of Theorem \ref{thr:Q's PQG Winning Strategies in U(2) Appendix}]
{\small \
	\begin{enumerate}
		\item	The two classes of winning strategies of Q, $\mathcal{C}_{+}$ and $\mathcal{C}_{-}$, which were establish by Theorem \ref{thr:Q's PQG Winning Strategies in D_8 Appendix}, are still present in $U(2)$. However, they now contain infinitely many equivalent strategies. To see why this is so, let us consider a winning strategy $\sigma = (A_{1}, A_{2})$ in $\mathcal{C}_{+}$. We may associate to this strategy the collection of infinitely many strategies $( A_{1} (\theta_{1}), A_{2} (\theta_{2}) )$, where $\theta_{1}, \theta_{2} \in \mathbb{R}$. Every strategy in this collection of strategies is equivalent to $(A_{1}, A_{2})$ because the action of every operator $A \in U(2)$ on a ket $\ket{\psi}$ is the same as the action of $e^{i \theta} A \in U(2)$ on $\ket{\psi}$. The same holds for every strategy in $\mathcal{C}_{-}$. Hence, we may conclude that in $U(2)$ the two classes $\mathcal{C}_{+}$ and $\mathcal{C}_{-}$ contain infinitely many equivalent strategies.

		Let us assume that there exists a third class $\mathcal{C}'$ and consider a strategy $\sigma = (A_{1}, A_{2})$ in this class. Then the action of $A_{1}$ must drive the coin into some state other than $\ket{+}$ or $\ket{-}$. However, (\ref{eq:Characteristic Property II of Winning Strategies Appendix}) asserts that $A_1 \ket{0} \in Fix ( \{ F \} )$, which, in view of (\ref{eq:The Fixed Set of F in U(2) Appendix}), implies that $A_1 \ket{0} \in \{ \ket{+}, \ket{-} \}$, a contradiction. Consequently, there are just two classes of winning strategies $\mathcal{C}_{+}$ and $\mathcal{C}_{-}$.

		It remains to prove that $\mathcal{C}_{+}$ and $\mathcal{C}_{-}$ do not contain any new winning strategy that are not of the form stated in (\ref{eq:U(2) Winning Strategy Classes Appendix}). To arrive at a contradiction, let us suppose that $\sigma = (A_{1}, A_{2})$ is a ``new winning'' strategy. If $A_{1}$
		$
		=
		\begin{bmatrix}
			\begin{array}{lr}
				z_{1} & w_{1} \\
				z_{2} & w_{2}
			\end{array}
		\end{bmatrix}
		$,
		then its columns form an orthonormal basis for $\mathcal{H}_2$ because it is a unitary operator. This means that its rows and columns satisfy the following relations:
		\begin{align}
			z_{1}^{*} z_{1} + z_{2}^{*} z_{2} &= 1 \ ,
			\label{eq:Orhonormal Condition I}
			\tag{ \ref{thr:Q's PQG Winning Strategies in U(2) Appendix}.i }
			\\
			w_{1}^{*} w_{1} + w_{2}^{*} w_{2} &= 1 \ , \quad \text{and}
			\label{eq:Orhonormal Condition II}
			\tag{ \ref{thr:Q's PQG Winning Strategies in U(2) Appendix}.ii }
			\\
			z_{1}^{*} w_{1} + z_{2}^{*} w_{2} &= 0 \ .
			\label{eq:Orhonormal Condition III}
			\tag{ \ref{thr:Q's PQG Winning Strategies in U(2) Appendix}.iii }
		\end{align}
		First, we investigate the possibility that $A_{1}$ sends the coin to state $\ket{+}$, assuming that $A_{1}$ is other than $H (\theta_{1}), R_{\frac{2 \pi}{8}} (\theta_{1}), S_{\frac{5 \pi}{8}} (\theta_{1})$ or $R_{\frac{10 \pi}{8}} (\theta_{1})$. In this case,
		\begin{align} \label{eq:Column 1 of A_1}
			\begin{bmatrix}
				\begin{array}{lr}
					z_{1} & w_{1} \\
					z_{2} & w_{2}
				\end{array}
			\end{bmatrix}
			\begin{bmatrix}
				\begin{array}{lr}
					1 \\
					0
				\end{array}
			\end{bmatrix}
			=
			e^{i \theta_{1}}
			\begin{bmatrix}
				\begin{array}{lr}
					\frac{1}{\sqrt{2}} \\
					\frac{1}{\sqrt{2}}
				\end{array}
			\end{bmatrix}
			\Rightarrow
			\begin{bmatrix}
				\begin{array}{lr}
					z_{1} \\
					z_{2}
				\end{array}
			\end{bmatrix}
			=
			e^{i \theta_{1}}
			\begin{bmatrix}
				\begin{array}{lr}
					\frac{1}{\sqrt{2}} \\
					\frac{1}{\sqrt{2}}
				\end{array}
			\end{bmatrix}
			\ ,
			\ \text{for \ some} \ \theta_{1} \in \mathbb{R} \ .
			\tag{ \ref{thr:Q's PQG Winning Strategies in U(2) Appendix}.iv }
		\end{align}
		If we combine (\ref{eq:Column 1 of A_1}) with (\ref{eq:Orhonormal Condition III}) we derive that
		\begin{align} \label{eq:Relation Between w1 and w2}
			z_{1}^{*} w_{1} + z_{2}^{*} w_{2} = 0
			\Rightarrow
			\frac{e^{- i \theta_{1}}}{\sqrt{2}} (w_{1} +  w_{2}) = 0
			\Rightarrow
			w_{2} = - w_{1}
			\ .
			\tag{ \ref{thr:Q's PQG Winning Strategies in U(2) Appendix}.v }
		\end{align}
		In view of (\ref{eq:Relation Between w1 and w2}), (\ref{eq:Orhonormal Condition III}) becomes
		\begin{align} \label{eq:Modulus of w1 and w2}
			2 | w_{1} |^{2} = 1 \Rightarrow | w_{1} | = \frac{1}{\sqrt{2}}
			\ .
			\tag{ \ref{thr:Q's PQG Winning Strategies in U(2) Appendix}.vi }
		\end{align}
		All complex numbers of the form $e^{i \varphi} \frac{1}{\sqrt{2}}$, where $\varphi \in \mathbb{R}$, are solutions of the equation (\ref{eq:Modulus of w1 and w2}). We may therefore choose $\varphi = \theta_{1}$ and set $w_{1} = \frac{e^{i \theta_{1}}}{\sqrt{2}}$, in which case, $w_{2}$ becomes $- \frac{e^{i \theta_{1}}}{\sqrt{2}}$. Hence, $A_{1}$ is in fact
		$
		e^{i \theta_{1}}
		\begin{bmatrix}
			\begin{array}{lr}
				\frac{1}{\sqrt{2}} & \frac{1}{\sqrt{2}} \\
				\frac{1}{\sqrt{2}} & - \frac{1}{\sqrt{2}}
			\end{array}
		\end{bmatrix}
		$,
		which means that $A_{1}$ is one of $H (\theta_{1}), R_{\frac{2 \pi}{8}} (\theta_{1}), S_{\frac{5 \pi}{8}} (\theta_{1})$ or $R_{\frac{10 \pi}{8}} (\theta_{1})$, in stark contrast to our initial assumption.

		We arrive at similar contradictions if we assume that $A_{1}$ drives the coin to state $\ket{-}$. Thus, we conclude that, other than the strategies described by (\ref{eq:U(2) Winning Strategy Classes Appendix}), there are no more winning strategies for Q.
		\item	Based on the above analysis it is straightforward to see that (\ref{eq:U(2) Winning State Paths Appendix}) holds.
		\item	By Definition \ref{def:Winning and Dominant Strategies}, Picard has no winning strategy because if Q employs one of his winning strategies, Picard has $0.0$ probability to win the game.
	\end{enumerate}
}
\end{proof}

\subsection{Proofs for Section \ref{sec:Extending the PQG}}

We now prove the important Theorem \ref{thr:Picard Lacks a Winning Strategy Appendix}.

\begin{ManualTheorem}{6.1}[Picard lacks a winning strategy] \label{thr:Picard Lacks a Winning Strategy Appendix}
	Picard does not have a winning strategy in any $n$-round game, $n \geq 2$, as long as Q makes at least one move.
\end{ManualTheorem}
\begin{proof}[Proof of Theorem \ref{thr:Picard Lacks a Winning Strategy Appendix}]
	{\small
		Let us first note that according to Definition \ref{def:Extended PQGs} and our assumptions at the beginning of Section \ref{sec:Extending the PQG}, in every $n$-round game with $n \geq 2$, Q makes at least one move. As a matter of fact, any $n$-round game has one of the following forms.
		\begin{enumerate}
			\item	$(Q, P, \dots, Q, P)$, in which case if Q employs the strategy $(H, I, \dots, I)$, the state of the coin prior to measurement will either be $\ket{+}$, if the initial state of the coin is $\ket{0}$, or $\ket{-}$, if the initial state of the coin is $\ket{1}$. This is because $Fix ( \{ I, F \} ) = \{ \ket{+}, \ket{-} \}$, so the coin will stay in one of these states no matter which strategy Picard uses. When the coin is measured, Picard will have exactly $0.5$ probability to win irrespective of which is his target state. Thus, Picard does not possess a winning strategy because, by Definition \ref{def:Winning and Dominant Strategies}, a winning strategy means that he wins the game with probability $1.0$.
			\item	$(P, Q, \dots, P, Q)$, where again, if the same strategy $(H, I, \dots, I)$ is used by Q, will prevent Picard from surely winning the game. The coin will be in one of its basis states (which one depends on the initial state and Picard's first move) when Q acts on it for the first time. His action will drive the coin to state $\ket{+}$ (if the coin was at state $\ket{0}$) or $\ket{-}$ (if the coin was at state $\ket{1}$). This will be the state of the coin prior to measurement no matter which strategy Picard uses because $Fix ( \{ I, F \} ) = \{ \ket{+}, \ket{-} \}$.  When the coin is measured, Picard will have exactly $0.5$ probability to win irrespective of which is his target state. Hence, Picard does not have a winning strategy because, by Definition \ref{def:Winning and Dominant Strategies}, a winning strategy means that he wins the game with probability $1.0$.
			\item	$(Q, P, \dots, Q, P, Q)$, in which case Q has a winning strategy. If the initial state is the same as Q's target state, then Q's winning strategy is $(H, I, \dots, I, H)$; if it is different then Q's winning strategy is $(H, I, \dots, I, FH)$. In this case Picard has precisely $0.0$ probability to win, so he certainly does not possess a winning strategy.
			\item	$(P, Q, \dots, P, Q, P)$, where once more the strategy $(H, I, \dots, I)$ can be used by Q to prevent Picard from surely winning the game. The coin will be in one of its basis states (which one depends on the initial state and Picard's first move) when Q acts on it for the first time. His action will drive the coin to state $\ket{+}$ (if the coin was at state $\ket{0}$) or $\ket{-}$ (if the coin was at state $\ket{1}$). This will be the state of the coin prior to measurement no matter which strategy Picard uses because $Fix ( \{ I, F \} ) = \{ \ket{+}, \ket{-} \}$.  When the coin is measured, Picard will have exactly $0.5$ probability to win irrespective of which is his target state. Therefore, Picard does not have a winning strategy because, by Definition \ref{def:Winning and Dominant Strategies}, a winning strategy means that he wins the game with probability $1.0$.
		\end{enumerate}
	}
\end{proof}

The next couple of theorems give important negative results for Q by specifying the games in which he cannot surely win.

\begin{ManualTheorem}{6.2}[Q lacks a winning strategy when Picard plays last] \label{thr:Q Lacks a Winning Strategy when Picard Plays Last Appendix}
	Q does not have a winning strategy in any $n$-round game, $n \geq 2$, in which Picard makes the last move.
\end{ManualTheorem}
\begin{proof}[Proof of Theorem \ref{thr:Q Lacks a Winning Strategy when Picard Plays Last Appendix}]
	{\small
		Any $n$-round game in which Picard makes the last move has one of the following two forms.
		\begin{enumerate}
			\item	$(Q, P, \dots, Q, P)$, where $n$ is even and both Picard and Q make $\frac{n}{2}$ moves. Let us assume to the contrary that there exists a winning strategy $\sigma_{Q} = ( A_{1}, \dots, A_{\frac{n}{2}} )$ for Q. According to Definition \ref{def:Winning and Dominant Strategies}, the fact that $\sigma_{Q}$ is a winning strategy means that for every strategy of Picard, Q wins the game with probability $1.0$. Since this holds for every strategy of Picard, it must also hold for the strategies $\sigma_{P} = ( I, \dots, I, I )$ and $\sigma'_{P} = ( I, \dots, I, F )$. The former implies that after Q's last action the coin must be at the basis state $\ket{q_Q}$, whereas the latter implies that after Q's last action the coin must be at the opposite basis state, which is absurd. Thus, Q does not possess a winning strategy.
			\item	$(P, Q, \dots, P, Q, P)$, where $n$ is odd, Q makes $\frac{n}{2}$ moves and Picard make $\frac{n}{2} + 1$ moves. In order to arrive at a contradiction, we assume to the contrary that there exists a winning strategy $\sigma_{Q} = ( A_{1}, \dots, A_{\frac{n}{2}} )$ for Q. In view of Definition \ref{def:Winning and Dominant Strategies}, if Q employs $\sigma_{Q}$, he will win the game with probability $1.0$ no matter which strategy Picard chooses. If Picard uses $\sigma_{P} = ( I, \dots, I, I )$, then the fact that $\sigma_{Q}$ is a winning strategy means that after Q's last action the coin must be at the basis state $\ket{q_Q}$. On the other hand, if Picard uses $\sigma_{P} = ( I, \dots, I, F )$, then the fact that $\sigma_{Q}$ is a winning strategy means that after Q's last action the coin must be at the opposite basis state. This contradiction proves that Q does not possess a winning strategy.
		\end{enumerate}
	}
\end{proof}

\begin{ManualTheorem}{6.3}[Q lacks a winning strategy when Picard plays first] \label{thr:Q Lacks a Winning Strategy when Picard Plays First Appendix}
	Q does not have a winning strategy in any $n$-round game, $n \geq 2$, in which Picard makes the first move.
\end{ManualTheorem}
\begin{proof}[Proof of Theorem \ref{thr:Q Lacks a Winning Strategy when Picard Plays First Appendix}]
	{\small
		Any $n$-round game in which Picard makes the first move has one of the following two forms.
		\begin{enumerate}
			\item	$(P, Q, \dots, P, Q)$, where $n$ is even and both Picard and Q make $\frac{n}{2}$ moves. Let us assume to the contrary that there exists a winning strategy $\sigma_{Q} = ( A_{1}, \dots, A_{\frac{n}{2}} )$ for Q. According to Definition \ref{def:Winning and Dominant Strategies}, the fact that $\sigma_{Q}$ is a winning strategy means that for every strategy of Picard, Q wins the game with probability $1.0$. Since this holds for every strategy of Picard, it must also hold for the strategies $\sigma_{P} = ( I, \dots, I, I )$ and $\sigma'_{P} = ( F, \dots, I, I )$. The former implies that
			\begin{align} \label{eq:Q's Winning Strategy Property I}
				A_{\frac{n}{2}} \dots A_{1} \ket{q_0} = \ket{q_Q} \Rightarrow C \ket{q_0} = \ket{q_Q}
				\ ,
				\tag{ \ref{thr:Q Lacks a Winning Strategy when Picard Plays First Appendix}.i }
			\end{align}
			where $C = A_{\frac{n}{2}} \dots A_{1}$. Since the composition of unitary operators produces a unitary operator, we know that $C$ is unitary. On the other hand, the latter implies that
			\begin{align} \label{eq:Q's Winning Strategy Property II}
				A_{\frac{n}{2}} \dots A_{1} F \ket{q_0} = \ket{q_Q} \Rightarrow C F \ket{q_0} = \ket{q_Q}
				\ .
				\tag{ \ref{thr:Q Lacks a Winning Strategy when Picard Plays First Appendix}.ii }
			\end{align}
			If $\ket{q_0} = \ket{0}$, then $F \ket{q_0} = \ket{1}$, whereas if $\ket{q_0} = \ket{1}$, then $F \ket{q_0} = \ket{0}$. If we combine this last result with (\ref{eq:Q's Winning Strategy Property I}) and (\ref{eq:Q's Winning Strategy Property II}), we conclude that
			\begin{align} \label{eq:Q's Winning Strategy Property III}
				C \ket{0} = C \ket{1} = \ket{q_Q}
				\ ,
				\tag{ \ref{thr:Q Lacks a Winning Strategy when Picard Plays First Appendix}.iii }
			\end{align}
			which is of course impossible because $C$ is unitary. Thus, Q does not possess a winning strategy.
			\item	$(P, Q, \dots, P, Q, P)$, where $n$ is odd, Q makes $\frac{n}{2}$ moves and Picard make $\frac{n}{2} + 1$ moves. In order to arrive at a contradiction, we assume to the contrary that there exists a winning strategy $\sigma_{Q} = ( A_{1}, \dots, A_{\frac{n}{2}} )$ for Q. In view of Definition \ref{def:Winning and Dominant Strategies}, if Q employs $\sigma_{Q}$, he will win the game with probability $1.0$ no matter which strategy Picard chooses. If Picard uses $\sigma_{P} = ( I, \dots, I, I )$, then the fact that $\sigma_{Q}$ is a winning strategy means that after Q's last action the coin must be at the basis state $\ket{q_Q}$. On the other hand, if Picard uses $\sigma_{P} = ( I, \dots, I, F )$, then the fact that $\sigma_{Q}$ is a winning strategy means that after Q's last action the coin must be at the opposite basis state. This contradiction proves that Q does not possess a winning strategy.
		\end{enumerate}
	}
\end{proof}

The next Theorem \ref{thr:When Q Has a Winning Strategy Appendix} gives a positive answer to the question of whether Q, and in effect the quantum player, can surely win and under what circumstances.

\begin{ManualTheorem}{6.4}[When Q possesses a winning strategy] \label{thr:When Q Has a Winning Strategy Appendix}
	In any $n$-round game, $n \geq 2$, Q has a winning strategy iff Q makes the first and the last move.
\end{ManualTheorem}
\begin{proof}[Proof of Theorem \ref{thr:When Q Has a Winning Strategy Appendix}]
	{\small
		The one direction, i.e., if Q has a winning strategy then he must make the first and the last move, is an immediate consequence of Theorems \ref{thr:Q Lacks a Winning Strategy when Picard Plays Last Appendix} and \ref{thr:Q Lacks a Winning Strategy when Picard Plays First Appendix}.

		It remains to prove the other direction, that is if Q makes the first and the last move, then Q possesses a winning strategy. If the initial state of the coin $\ket{q_0}$ is the same as the target state of Q, then $\sigma_{Q} = ( H, I, \dots, I, H )$ a winning strategy for Q. Q's first action will drive the coin to either $\ket{+}$ (if the initial state is $\ket{0}$) or $\ket{-}$ (if the initial state is $\ket{1}$). In any case both $\{ \ket{+}, \ket{-} \}$ are fixed by $\{ I, F \}$, as Theorem \ref{thr:The Fixed Set of M_P in U(2) Appendix} asserts. Q's last move will send the coin back to $\ket{q_0}$. If the initial state of the coin $\ket{q_0}$ is different from the target state of Q, then it is trivial to check that $\sigma_{Q} = ( H, I, \dots, I, FH )$ a winning strategy for Q.
	}
\end{proof}

\begin{ManualCorollary}{6.5}[The impact of initial and target states] \label{crl:The Impact of Initial and Target States is Negligible Apendix}
	In any $n$-round game, $n \geq 2$, if Q has a winning strategy, then he has a winning strategy for every combination of initial and target states.
\end{ManualCorollary}
\begin{proof}[Proof of Corollary \ref{crl:The Impact of Initial and Target States is Negligible Apendix}]
	{\small
		An immediate consequence of Theorem \ref{thr:When Q Has a Winning Strategy Appendix}.
	}
\end{proof}

\bibliographystyle{ieeetr}
\bibliography{PQG_DiGBibliography}

\begin{thebibliography}{10}
\bibitem{Andronikos2021}
T.~Andronikos, ``Conditions that enable a player to surely win in sequential
  quantum games,''
\end{thebibliography}

\end{document}